\documentclass{amsart}

\usepackage{amssymb}
\usepackage{amsmath}
\usepackage{amsfonts}
\usepackage{amsthm}
\usepackage{graphicx}
\usepackage{epsfig}
\usepackage{longtable}
\newtheorem{thm}{Theorem}[section]
\newtheorem{lem}[thm]{Lemma}
\newtheorem{cor}[thm]{Corollary}

\theoremstyle{definition}
\newtheorem{defn}[thm]{Definition}
\newtheorem{eg}[thm]{Example}

\theoremstyle{remark}
\newtheorem{rmk}[thm]{Remark}

\numberwithin{equation}{section}

\newcommand{\eps}{\varepsilon}

\newcommand{\DEF}{{:=}}
\newcommand{\FED}{{=:}}

\newcommand{\PT}[1]{\mathbf{#1}}

\newcommand{\re}{\mathop{\mathrm{Re}}}

\DeclareMathOperator{\dd}{\mathrm{d}}

\DeclareMathOperator{\betafcn}{B}
\DeclareMathOperator{\bal}{Bal}
\DeclareMathOperator{\CAP}{cap}

\DeclareMathOperator{\gammafcn}{\Gamma}

\DeclareMathOperator{\kelvin}{K}
\DeclareMathOperator{\kelvinMEAS}{\mathcal{K}}

\DeclareMathOperator{\digammafcn}{\psi}

\DeclareMathOperator{\supp}{supp}

\DeclareMathOperator{\HyperF}{F}
\DeclareMathOperator{\HyperTildeF}{\tilde{F}}
\newcommand{\Hypergeom}[5]{{\sideset{_#1}{_#2}\HyperF\!\left(\substack{\displaystyle#3\\\displaystyle#4};#5\right)}}
\newcommand{\HypergeomReg}[5]{{\sideset{_#1}{_#2}\HyperTildeF\!\left(\substack{\displaystyle#3\\\displaystyle#4};#5\right)}}

\newcommand{\Pochhsymb}[2]{{\left(#1\right)_{#2}}}

\allowdisplaybreaks[1]

\title[Minimal Riesz energy on the sphere]{Minimal Riesz energy on the sphere for axis-supported external fields}
\author{ J. S. Brauchart, P. D. Dragnev\textdagger, and E. B. Saff\textdaggerdbl } 
\thanks{\noindent The research of this author was supported, in part, by the ``Scholar-in-Residence'' program at IPFW, by the Austrian Science Foundation (FWF) under grant S9603-N13, and the Oberwolfach-Leibniz Fellow Programme. \\
\textdaggerdbl The research of this author  was supported, in part,
by the U. S. National Science Foundation under grants DMS-0603828 and DMS-0808093. }

\begin{document}

\address{J. S. Brauchart and E. B. Saff:
Center for Constructive Approximation,
Department of Mathematics,
Vanderbilt University,
Nashville, TN 37240,
USA }
\email{Johann.Brauchart@Vanderbilt.Edu}
\email{Edward.B.Saff@Vanderbilt.Edu}

\address{P. D. Dragnev:
Department of Mathematical Sciences,
Indiana-Purdue University,
Fort Wayne, IN 46805,
USA}
\email{dragnevp@ipfw.edu}

\begin{abstract}
We investigate the minimal Riesz $s$-energy problem for positive
measures on the $d$-dimensional unit sphere $\mathbb{S}^d$ in the
presence of an external field induced by a point charge, and more
generally by a line charge. The model interaction is that of Riesz
potentials $|\PT{x}-\PT{y}|^{-s}$ with $d-2\leq s<d$. For a given axis-supported external field, 
the support and the density of the corresponding extremal measure on $\mathbb{S}^d$ is determined.
The special case $s=d-2$ yields interesting phenomena, which we investigate in detail. A weak$^*$ asymptotic analysis is provided as $s\to(d-2)^+$.
\end{abstract}


\maketitle

\section{Introduction and results}

\subsection{Potential-theoretical preliminaries.}
Let $\mathbb{S}^d \DEF \{ \PT{x} \in \mathbb{R}^{d+1} : |\PT{x}|=1 \}$ 
be the unit sphere in $\mathbb{R}^{d+1}$, where $|\PT{\cdot}|$
denotes the Euclidean norm, and let $\sigma=\sigma_d$ be the unit
Lebesgue surface measure on $\mathbb{S}^d$. Recall that, using
cylindrical coordinates
\begin{equation}
\PT{x} = \left( \sqrt{1-u^2} \; \overline{\PT{x}}, u \right), \qquad -1 \leq u \leq 1, \;  \overline{\PT{x}} \in \mathbb{S}^{d-1}, \label{Cylindrical}
\end{equation}
we can write the decomposition
\begin{equation}
\dd \sigma_{d}(\PT{x}) = \frac{\omega_{d-1}}{\omega_{d}} \left( 1 - u^{2} \right)^{d/2-1} \, \dd u \, \dd \sigma_{d-1}(\overline{\PT{x}}).
\label{Measure.Decomposition}
\end{equation}
Here $\omega_{d}$ is the surface area of $\mathbb{S}^{d}$, and the
ratio of these areas can be evaluated as
\begin{equation}
\frac{\omega_{d}}{\omega_{d-1}} = \int_{-1}^{1} \left( 1 - u^2 \right)^{d/2-1} \dd u =
\frac{\sqrt{\pi}\gammafcn(d/2)}{\gammafcn((d+1)/2)} = 2^{d-1} \frac{\left[\gammafcn(d/2)\right]^{2}}{\gammafcn(d)}. \label{omega.ratio}
\end{equation}

Given a compact set $E\subset \mathbb{S}^d$, consider the class
$\mathcal{M}(E)$ of unit positive Borel measures supported on $E$.
For $0<s<d$, the {\em Riesz $s$-potential} and {\em Riesz
$s$-energy} of a measure $\mu\in \mathcal{M}(E)$ are given,
respectively, by
\begin{equation*}
U_s^\mu (\PT{x}) \DEF \int k_s(\PT{x},\PT{y}) \dd\mu(\PT{y}), \qquad
\mathcal{I}_s (\mu) \DEF \int\int k_s(\PT{x},\PT{y}) \dd\mu(\PT{x})
\dd\mu(\PT{y}) ,
\end{equation*}
where $k_s(\PT{x},\PT{y}) \DEF |\PT{x}-\PT{y}|^{-s}$ is the
so-called {\em Riesz kernel} (for $s=0$ we use the logarithmic
kernel $k_0(\PT{x},\PT{y})\DEF\log(1/|\PT{x}-\PT{y}|)$ instead). The
{\em $s$-energy} of $E$ is $W_s(E)\DEF\inf\{\mathcal{I}_s (\mu) :
\mu\in \mathcal{M}(E)\}$ and if $W_s(E)$ is finite, there is a unique measure $\mu_{E,s}$
achieving this minimal energy, which is called the {\em $s$-extremal measure} on $E$. The {\em $s$-capacity of $E$} is defined as
$\CAP_{s}(E)\DEF1/W_s(E)$ for $s>0$. 
(In the logarithmic case $s=0$ we define $\CAP_0(E) \DEF \exp\{-W_0(E)\}$, cf. \eqref{W.0.d} for $E=\mathbb{S}^d$.)
A property is said to hold {\em quasi-everywhere} (q.e.) if the exceptional set has $s$-capacity
zero. For more details see \cite[Chapter II]{L}. We remind the
reader that the $s$-energy of $\mathbb{S}^d$ is given by
\begin{equation}
W_s(\mathbb{S}^d) = \frac{\gammafcn(d)\gammafcn((d-s)/2)}{2^s\gammafcn(d/2)\gammafcn(d-s/2)}, \qquad 0 < s < d. \label{SphereEnergy}
\end{equation}

The weighted $s$-energy associated with a non-negative lower
semi-continuous {\em external field} $Q : E \to [0,\infty]$ is given
by
\begin{equation}
\mathcal{I}_Q(\mu) \DEF \mathcal{I}_s(\mu) + 2 \int Q(\PT{x}) \dd\mu(\PT{x}) .
\end{equation}
A measure $\mu_Q\in \mathcal{M}(E )$ such that $\mathcal{I}_Q(\mu_Q) = V_Q$, where
\begin{equation}
V_Q \DEF \inf \left\{ \mathcal{I}_Q (\mu)  :  \mu \in \mathcal{M}(E) \right\},
\end{equation}
is called an {\em extremal (or positive equilibrium) measure on $E$ associated with
$Q(\PT{x})$}. The measure $\mu_Q$ is characterized by the
Gauss variational inequalities
\begin{align}
U_s^{\mu_Q}(\PT{x}) + Q(\PT{x}) &\geq F_Q \qquad \text{q.e. on $E$,} \label{geqineq} \\
U_s^{\mu_Q}(\PT{x}) + Q(\PT{x}) &\leq F_Q \qquad \text{everywhere on $\supp(\mu_Q)$,} \label{leqineq}
\end{align}
where
\begin{equation}
F_Q \DEF V_Q - \int Q(\PT{x}) \, \dd \mu_Q(\PT{x}).
\end{equation}
For simplicity, we suppressed in some of the above notation the
dependency on $s$; that is, $\mathcal{I}_Q=\mathcal{I}_{Q,s}$,
$\mu_Q=\mu_{Q,s}$, etc. We note that for suitable external fields
(e.g. continuous on $E=\mathbb{S}^d$), the inequality in
\eqref{geqineq} holds everywhere, which implies that equality holds
in \eqref{leqineq}.

The existence, uniqueness, and characterization-related questions
concerning extremal potentials with external fields in the most
general setting can be found in \cite{Z1}--\cite{Z3}. We remark that the
logarithmic potential with external fields is treated in depth in \cite{ST}.

When $Q \equiv 0$ and ${\rm cap}_s (E)>0$, the extremal measure
$\mu_Q$ is the same as the measure $\mu_E=\mu_{E,s}$.

In \cite{DS} {\em Riesz external fields}
\begin{equation}
Q_{\PT{a},q}(\PT{x}) \DEF Q_{\PT{a},q,s}(\PT{x}) \DEF q \left| \PT{x} - \PT{a} \right|^{-s} \text{on $E=\mathbb{S}^d$, $d-2<s<d$,} \label{RieszField}
\end{equation}
were considered, where $q>0$ and $\PT{a}$ is a fixed point on
$\mathbb{S}^d$. \footnote{The case $d=1$, $s=0$, where $\PT{a}$ is a point on the unit circle was investigated in \cite{LSV}.} The motivation for that investigation was to obtain
new separation results for minimal $s$-energy points on the sphere.
In the current work we extend that investigation to Riesz external fields
$Q_{\PT{a},q}$ with $\PT{a}\not\in \mathbb{S}^d$ and develop a
technique for finding the extremal measure associated with more
general axis-supported external fields.

\subsection{Signed Equilibrium.}

We note that for $d=2$ and $s=1$ it is a standard electrostatic
problem to find the charge density (signed measure) on a charged,
insulated, conducting sphere in the presence of a point charge $q$
placed off the sphere (see \cite[Chapter 2]{J}). This motivates us
to give the following definition (see \cite{BDS}).


\begin{defn} Given a compact subset $E\subset \mathbb{R}^p$ ($p\geq 3$) and
an external field $Q$, we call a signed measure
$\eta_{E,Q}=\eta_{E,Q,s}$ supported on $E$ and of total charge
$\eta_{E,Q}(E)=1$ {\em a signed $s$-equilibrium on $E$ associated with
$Q$} if its weighted Riesz $s$-potential is constant on $E$, that is
\begin{equation}
U_s^{\eta_{E,Q}}(\PT{x}) + Q(\PT{x}) = F_{E,Q} \qquad \text{for all $\PT{x} \in E$.} \label{signedeq}
\end{equation}
\end{defn}


The choice of the normalization $\eta_{E,Q}(E)=1$ is just for
convenience in the applications here. 
Lemma 2.1 below establishes that if a signed $s$-equilibrium
$\eta_{E,Q}$ exists, then it is unique.

In \cite{FABR} Fabrikant et al give a derivation of certain signed Riesz equilibria on suitably parametrized  surfaces in $\mathbb{R}^3$, including spherical caps when $Q(\PT{x}) \equiv 0$. We remark that the determination of signed equilibria is a
substantially easier problem than that of finding non-negative
extremal measures, which is the goal of this paper. However, the solution
to the former problem is useful in solving the latter problem.

Our first result establishes existence of the signed $s$-equilibrium
associated with the Riesz external field $Q_{\PT{a},q}$,
$\PT{a}\not\in \mathbb{S}^d$, defined in \eqref{RieszField}. We
assume that $\PT{a}$ lies above the North Pole $\PT{p}:=(\PT{0},1)$,
that is $\PT{a}=(\PT{0},R)$ and $R>1$ (the case $R<1$ is handled by
inversion).

Throughout, $\Hypergeom{2}{1}{a,b}{c}{z}$ and
$\HypergeomReg{2}{1}{a,b}{c}{z}$ denote the Gauss hypergeometric
function and its regularized form \footnote{which is well-defined
even for $c$ a negative integer} with series expansions
\begin{equation}
\Hypergeom{2}{1}{a,b}{c}{z} \DEF \sum_{n=0}^\infty
\frac{\Pochhsymb{a}{n}\Pochhsymb{b}{n}}{\Pochhsymb{c}{n}}
\frac{z^n}{n!}, \quad \HypergeomReg{2}{1}{a,b}{c}{z} \DEF
\sum_{n=0}^\infty
\frac{\Pochhsymb{a}{n}\Pochhsymb{b}{n}}{\gammafcn(n+c)}
\frac{z^n}{n!}, \qquad |z|<1, \label{HypergeomSeries}
\end{equation}
where $\Pochhsymb{a}{0} \DEF 1$ and $\Pochhsymb{a}{n} \DEF a (a+1)
\cdots (a+n-1)$ for $n\geq1$ is the Pochhammer symbol. The incomplete Beta
function and the Beta function are defined as
\begin{equation} \label{betafnc}
\betafcn(x;\alpha,\beta) \DEF \int_{0}^x v^{\alpha-1} \left( 1 - v \right)^{\beta-1} \dd v, \qquad \betafcn(\alpha,\beta) \DEF \betafcn(1;\alpha,\beta),
\end{equation}
whereas the regularized incomplete Beta function is given by
\begin{equation}
\mathrm{I}(x;a,b) \DEF \betafcn(x;a,b) \big/ \betafcn(a,b). \label{regbetafnc}
\end{equation}


\begin{thm} Let $0<s<d$ and $R>1$. The signed $s$-equilibrium
$\eta_{\PT{a}}=\eta_{\mathbb{S}^d,Q_{\PT{a},q},s}$ on $\mathbb{S}^d$
associated with the Riesz external field $Q_{\PT{a},q}$, $
\PT{a}=R\PT{p}$, is given by
\begin{equation}
\dd \eta_{\PT{a}}(\PT{x}) = \left\{ 1 + \frac{q U_s^\sigma (\PT{a})}{W_s(\mathbb{S}^d)} -
\frac{q\left(R^2-1\right)^{d-s}}{W_s(\mathbb{S}^d)\left|\PT{x}-\PT{a}\right|^{2d-s}}
\right\} \dd \sigma(\PT{x}).\label{signedeqdens}
\end{equation}
Furthermore, $U_s^\sigma (\PT{a})=\int k_s (\PT{a},\PT{y})\, \dd \sigma
(\PT{y})$ has the following representation:
\begin{equation}
U_s^\sigma (\PT{a}) = \left( R + 1 \right)^{-s} \Hypergeom{2}{1}{s/2,d/2}{d}{\frac{4R}{\left(R+1\right)^2}}. \label{SignEqPot(b)}
\end{equation}\label{SignEq}
\end{thm}


We remark that in the Coulomb case $d=2$ and $s=1$, the
representation \eqref{signedeqdens} is well-known from elementary
physics (cf. \cite[p.~61]{J}).


The next result explicitly shows the relationship between
$q$ and $R$ so that $\mu_{Q_{\PT{a},q}}$ coincides with the signed
equilibrium and has as support the entire sphere.


\begin{cor} \label{Cor}
Let $0<s<d$ and $R=|\PT{a}|>1$. Then $\supp(\mu_{Q_{\PT{a},q}})=\mathbb{S}^d$ if and only if
\begin{align}
\frac{W_s(\mathbb{S}^d)}{q} 
&\geq \frac{\left(R+1\right)^{d-s}}{\left(R-1\right)^d} - U_s^\sigma(\PT{a})  \label{eqsignedA}\\
&= \frac{1}{\left(R+1\right)^s}\sum_{k=0}^\infty \left[ 1 - \frac{\Pochhsymb{s/2}{k}}{\Pochhsymb{d}{k}} \right] \frac{\Pochhsymb{d/2}{k}}{k!} \left[ \frac{4R}{\left(R+1\right)^2} \right]^k. \label{eqsigned}
\end{align}
In such a case, $\mu_{Q_{\PT{a},q}}=\eta_{\PT{a}}$.
\end{cor}

\begin{rmk}
Observe that the function of $R$ in \eqref{eqsigned} is strictly decreasing for $R>1$. Thus, for any fixed charge $q$ there is
a critical $R_q$ given by equality in \eqref{eqsignedA}, such that
for $R\geq R_q$ the extremal support is the entire sphere.
\end{rmk}

\subsection{The Newtonian case $s=d-1$.}


The following example deals with the classical case of a
Newtonian potential (relative to the manifold dimension). The
example answers a question of A. A. Gonchar; namely, how far from
the unit sphere should a unit point charge be placed so that the
support of the extremal measure associated with the external
field exerted by the charge be the entire sphere?

\begin{eg}
Let $d\geq2$, $s=d-1$, $q=1$ and $\PT{a}=(\PT{0},R)$. Then $W_s(\mathbb{S}^d)=1$ (cf. \eqref{SphereEnergy}) and from the
mean-value property for harmonic functions we can write
\begin{equation*}
U_{s}^\sigma(\mathbf{a}) = \frac{1}{R^{d-1}} \qquad \text{for $R\geq1$.}
\end{equation*}
Thus \eqref{eqsigned} in this case is equivalent to the inequality
\begin{equation}
1 \geq \frac{R + 1}{\left( R - 1 \right)^d} - \frac{1}{R^{d-1}} \qquad \text{or} \qquad 1 \geq \frac{\rho + 2}{\rho^d} - \frac{1}{\left(\rho+1\right)^{d-1}},
\end{equation}
where $\rho$ measures the distance between the unit charge and the
surface of the sphere. Equality holds, if $\rho$ is an algebraic
number satisfying
\begin{equation} \label{P}
P(d;\rho) \DEF \left( \rho^d - 2 - \rho \right) \left( \rho + 1 \right)^{d-1} + \rho^d = 0,
\end{equation}
or on expanding the polynomial $P(d;\rho)$,
\begin{equation} \label{algebr.equ}
\sum_{m=0}^{d-1} \binom{d-1}{m} \rho^{m+d} - \sum_{m=0}^{d-1} \left[ \binom{d}{m} + \binom{d-1}{m} \right] \rho^m = 0.
\end{equation}
The monic polynomial\footnote{Properties of these polynomials will be investigated in a future publication.} $P(d;\rho)$ with integer coefficients has odd degree $2d-1$. Furthermore, $P(d;1)<0$ and hence $P(d;\rho)$ has at least one positive root; but, by Descartes' Sign Rule, this is the only positive root. This simple root $\rho_+$ must be in the interval $(1,2]$, since $P(d;\rho)>0$ for $\rho>2$. Asymptotic analysis shows that
\begin{equation}
\rho_+ = 1 + \left( \log 3 \right) / d  + \mathcal{O}(1/d^2) \qquad \text{as $d\to\infty$.}
\end{equation}
%
%
%

Of particular interest is the case when $d=2$. Then one easily
computes that the distance between the point charge and the surface
of the sphere is given precisely by the golden ratio
\begin{equation*}
\rho_+ = ( 1 + \sqrt{5} ) / 2.
\end{equation*}
We note that the fact that the inequality $R-1\geq \rho_+$ 
implies $\supp(\mu_{Q_{\PT{a},1}})=\mathbb{S}^2$ follows
from an elementary physics argument.
\end{eg}

\subsection{The Mhaskar-Saff $\mathcal{F}_s$-functional and the extremal support.}

An important tool in our analysis is the Riesz analog of the {\it
Mhaskar-Saff $F$-functional} from classical logarithmic potential in
the plane (see \cite{MS} and \cite[Chapter IV, p. 194]{ST}).

\begin{defn} Given a compact subset $K\subset \mathbb{S}^d$ of
positive $s$-capacity, we define the {\it
$\mathcal{F}_s$-functional} of the set $K$ as
\begin{equation} \label{Functional}
\mathcal{F}_s(K) \DEF W_s(K) + \int Q(\PT{x}) \, \dd \mu_K(\PT{x}),
\end{equation}
where $W_s(K)$ is the $s$-energy of $K$ and $\mu_K$ is the $s$-extremal
measure (without external field) on $K$.
\end{defn}

\begin{rmk} We caution the reader that \eqref{Functional} is the negative of the
$F$-functional defined in \cite{MS} and \cite{ST}.
\end{rmk}

\begin{rmk} \label{FuncSignedEqRel}
When $d-2<s<d$, there is a remarkable relationship between the signed
equilibrium and the $\mathcal{F}_s$-functional. Namely, if the
signed $s$-equilibrium on a compact set $K$ associated with $Q$ exists,
then $\mathcal{F}_s (K)=F_{K,Q}$, where  $F_{K,Q}$ is the constant
from \eqref{signedeq}. Indeed, if $\eta_{K,Q}$ exists, we integrate
\eqref{signedeq} with respect to $\mu_K$ and interchange the order
of integration to obtain the asserted equality.
\end{rmk}

\begin{rmk}
With the notion of the functional $\mathcal{F}_s$ at hand we can restate the results of Theorem \ref{SignEq} and Corollary \ref{Cor} as follows: For $0<s<d$ and $R>1$ the signed $s$-equilibrium $\eta_{\PT{a}}=\eta_{\mathbb{S}^d,Q_{\PT{a},q},s}$ on $\mathbb{S}^d$ associated with  $Q_{\PT{a},q}$, $\PT{a}=R\PT{p}$, is given by
\begin{equation}
\dd \eta_{\PT{a}}(\PT{x}) = \frac{1}{W_s(\mathbb{S}^d)} \left\{ \mathcal{F}_s(\mathbb{S}^d) - q \left(R^2-1\right)^{d-s} \big/ \left|\PT{x}-\PT{a}\right|^{2d-s}
\right\} \dd \sigma(\PT{x}). \label{signedeqdens.2}
\end{equation}
Moreover, $\supp(\mu_{Q,s})=\mathbb{S}^d$ (that is $\mu_{Q,s}=\eta_{\PT{a}}$) if and only if 
\begin{equation}
\mathcal{F}_s(\mathbb{S}^d) \geq q \left( R + 1 \right)^{d-s} \big/ \left( R - 1 \right)^{d}.
\end{equation}
\end{rmk}

The following optimization property is the main motivation for
introducing the $\mathcal{F}_s$-functional.
\begin{thm} \label{FuncOptThm}
Let $d-2\leq s<d$ with $s>0$ and $Q$ be an external field on
$\mathbb{S}^d$. Then the $\mathcal{F}_s$-functional is minimized for
$S_Q \DEF \supp (\mu_Q)$.
\end{thm}

The next theorem provides sufficient conditions on a general
external field $Q$ that guarantee that the extremal support
$S_Q$ is a spherical zone or a spherical cap. 

\begin{thm} Let $d-2\leq s<d$ with $s>0$ and the external field $Q:\mathbb{S}^d\to [0,\infty]$ be
rotationally invariant about the polar axis; that is,
$Q(\PT{z})=f(\xi)$, where $\xi$ is the altitude of 
$\PT{z}=(\sqrt{1-\xi^2}\; \overline{\PT{z}},\xi)$ (see
\eqref{Cylindrical}). Suppose that $f$ is a convex function on
$[-1,1]$. Then the support of the $s$-extremal measure $\mu_Q$ on $\mathbb{S}^d$ is a
spherical zone; namely, there are numbers $-1\leq t_1\leq t_2\leq 1$
such that
\begin{equation}
\supp(\mu_Q)=\Sigma_{t_1,t_2}:=\{ (\sqrt{1-u^2}\,
\overline{\PT{x}},u)\ :\ t_1\leq u \leq t_2, \, \overline{\PT{x}}\in
\mathbb{S}^{d-1} \}. \label{eqsupp}
\end{equation}
Moreover, if additionally $f$ is increasing, then $t_1=-1$ and the
support of $\mu_Q$ is a spherical cap centered at the South Pole.
\label{ConnThm}
\end{thm}

It is easy to see that the external field $Q_{\PT{a},q}(\PT{z}) = q | 1 - 2 R \xi + R^2 |^{-s/2}$ is rotationally invariant about the
polar axis and is an increasing and convex function of the altitude
$\xi$ of $\PT{z}$. Therefore, from Theorem \ref{ConnThm} we conclude that
the support of the extremal measure $\mu_{Q_{\PT{a},q}}$ on $\mathbb{S}^d$ is a spherical
cap. In view of Theorem \ref{FuncOptThm} we thus need only to minimize
the $\mathcal{F}_s$-functional over the collection of spherical caps
centered at the South Pole in order to determine $S_Q$. For this purpose, in consideration of
Remark \ref{FuncSignedEqRel}, we first seek an explicit
representation for the signed equilibria for these spherical caps.

Denote by $\Sigma_t$ the spherical cap centered at the South Pole
\begin{equation}
\Sigma_t \DEF \Sigma_{-1,t}, \label{Sigma}
\end{equation}
(cf. \eqref{eqsupp}), and let $\eta_t$ be the signed $s$-equilibrium on
$\Sigma_t$ associated with $Q_{\PT{a},q}$. Using M. Riesz's approach
to $s$-balayage as presented in \cite[Chapter IV]{L}, we introduce
the following $s$-balayage measures onto $\Sigma_t$:
\begin{equation}
\epsilon_t = \epsilon_{t,s} \DEF \bal_s(\delta_{\PT{a}},\Sigma_t),
\qquad \nu_t = \nu_{t,s} \DEF \bal_s(\sigma,\Sigma_t), \label{bal}
\end{equation}
where $\delta_{\PT{a}}$ is the unit Dirac-delta measure at ${\bf
a}$. Recall that given a measure $\nu$ and a compact set $K$ (of the
sphere $\mathbb{S}^d$), the balayage measure $\hat{\nu}:=\bal_s(\nu,K)$ preserves the Riesz $s$-potential of $\nu$ onto the
set $K$ and diminishes it elsewhere (on the sphere $\mathbb{S}^d$).
We remark that in what follows an important role is played by the
function
\begin{equation}
\Phi_s(t) \DEF W_s(\mathbb{S}^d) \left( 1 + q \left\|\epsilon_t\right\| \right) \big/ \left\|\nu_t\right\|, \qquad d-2<s<d.
\label{FuncPhi}
\end{equation}

The next assertion is an immediate consequence of the definition of
the balayage measures in \eqref{bal}. In Lemmas \ref{NuLem} and
\ref{EpsilonLem} below we present explicit formulas for their densities.
Their norms are calculated in Lemmas \ref{LemNuNorm} and
\ref{norm.eps.r}, respectively. Below we combine these formulas to
give an explicit form for the density of the signed $s$-equilibrium. The
only statements requiring further proof is the formula for the
weighted $s$-potential \eqref{eq:weighted.outside} when
$\xi>t$. We shall do this in Section \ref{sec:6}.

\begin{thm} Let $d-2<s<d$. The signed $s$-equilibrium $\eta_t$
on the spherical cap $\Sigma_t \subset \mathbb{S}^d$ associated with $Q_{\PT{a},q}$ is
given by
\begin{equation} 
\eta_t = \displaystyle{\frac{1+ q \| \epsilon_t \|}{\| \nu_t \|}\nu_t - q \epsilon_t}. \label{eta}
\end{equation}
It is absolutely continuous in the sense that for $\PT{x} = ( \sqrt{1-u^2} \overline{\PT{x}}, u) \in \Sigma_t$,
\begin{equation*}
\dd \eta_{t}(\PT{x}) = \eta_{t}^{\prime}(u)
\frac{\omega_{d-1}}{\omega_{d}} \left( 1 - u^2 \right)^{d/2-1} \dd u
\dd\sigma_{d-1}(\overline{\PT{x}}),
\end{equation*}
where (with $R=|\PT{a}|$ and $r = \sqrt{R^2 - 2 R t + 1}$)
\begin{equation*}
\begin{split} 
&\eta_{t}^{\prime}(u) = \frac{1}{W_s(\mathbb{S}^d)} \frac{\gammafcn(d/2)}{\gammafcn(d-s/2)}
\left( \frac{1-t}{1-u} \right)^{d/2} \left( \frac{t-u}{1-t} \right)^{(s-d)/2} \\
&\phantom{=\times}\times \Bigg\{ \Phi_s (t)
\HypergeomReg{2}{1}{1,d/2}{1-(d-s)/2}{\frac{t-u}{1-u}}  \\
&\phantom{=\times\pm}-  \frac{q\left( R + 1 \right)^{d-s}}{r^{d}}
\HypergeomReg{2}{1}{1,d/2}{1-(d-s)/2}{\frac{\left(R-1\right)^{2}}{r^{2}}
\, \frac{t-u}{1-u}}  \Bigg\}.
\end{split}
\end{equation*}

Furthermore, if $\PT{z} = ( \sqrt{1-\xi^2}\; \overline{\PT{z}},
\xi)\in \mathbb{S}^d$, the weighted $s$-potential is given by
\begin{align}
U_s^{\eta_t}(\PT{z})+Q_{\PT{a},s}(\PT{z}) &= \Phi_s(t), \qquad \PT{z} \in \Sigma_t, \\
\begin{split}
U_s^{\eta_t}(\PT{z})+Q_{\PT{a},s}(\PT{z}) &= \Phi_s(t) + q \frac{1}{\rho^s} \mathrm{I}\left(\frac{(R+1)^2}{r^2} \frac{\xi-t}{1+\xi};
\frac{d-s}{2}, \frac{s}{2} \right) \\
&\phantom{=\pm}- \Phi_s(t) \mathrm{I}\left(\frac{\xi-t}{1+\xi};
\frac{d-s}{2}, \frac{s}{2}\right), \qquad \PT{z} \in \mathbb{S}^d \setminus \Sigma_t, \label{eq:weighted.outside}
\end{split}
\end{align}
where $\rho=\sqrt{R^2-2R\xi+1}$ and $\mathrm{I}(x;a,b)$ is the regularized incomplete Beta function (see
\eqref{regbetafnc}). 
\label{SignEqThm}
\end{thm}

The corresponding statement for the case $s=d-2$ is given in Theorem \ref{ExcepThm}.

\begin{rmk} \label{rmk:F.s.EQ.Phi.s}
According to Remark \ref{FuncSignedEqRel} we have from Theorem
\ref{SignEqThm} that $\mathcal{F}_s(\Sigma_t) = \Phi_s(t) $.
Concerning the minimization of this function, we derive the following result.
\end{rmk}


\begin{thm}
Let $d-2<s<d$. For the external field $Q_{\PT{a},q}(\PT{x})$,
$\PT{a}=(\PT{0},R)$, $R>1$, the function $\Phi_s (t)$ has precisely
one global minimum $t_0\in (-1,1]$. This minimum is either the
unique solution $ t_0\in (-1,1)$ of the equation
\begin{equation}
\Phi_s (t) = \frac{q\left( R + 1 \right)^{d-s}}{\left( R^2 - 2 R t + 1 \right)^{d/2}},
\label{FuncEqCond}
\end{equation}
or $t_0=1$ when such a solution does not exist. Moreover, $t_0=\max
\{ t : \eta_t \geq 0 \}$.

The extremal measure $\mu_{Q_{\PT{a},q}}$ on $\mathbb{S}^d$ is given
by $\eta_{t_0}$ (see \eqref{eta}), and has as support the spherical
cap $\Sigma_{t_0}$.  \label{MainThm}
\end{thm}

Note that, in view of formulas \eqref{NormEpsB} and \eqref{NuNormA} for $\| \epsilon_t\|$ and $\| \nu_t \|$ given
below, equation \eqref{FuncEqCond} can be written in terms of
hypergeometric functions.

\begin{rmk} The restriction on the parameter $s$ arises in the process of applying the balayage method and the principle of domination. It is a
topic for further investigation to extend the range of $s$ for which
the conclusion of Theorem \ref{MainThm} remains true.
\end{rmk}

Figure \ref{fig1} gives an overview of the qualitative behavior of the weighted $s$-potential of the signed $s$-equilibrium measure $\eta_t$ on $\mathbb{S}^d$ associated with the external field $Q$ and its density with respect to $\sigma_d|_{\Sigma_t}$ for $s$ in the range $d-2<s<d$ and the choices $t<t_0$, $t=t_0$ and $t>t_0$. We remark that the derivative with respect to $\xi$ of the weighted $s$-potential becomes $\pm\infty$ as $\xi\to t^+$ for $t\neq t_0$ and vanishes for $t=t_0<1$ (cf. Remark \ref{rmk:derivative.weighted.pot}).
\begin{figure}[ht]
\begin{center}
\begin{minipage}{0.5\linewidth}
\centerline{\includegraphics[scale=1]{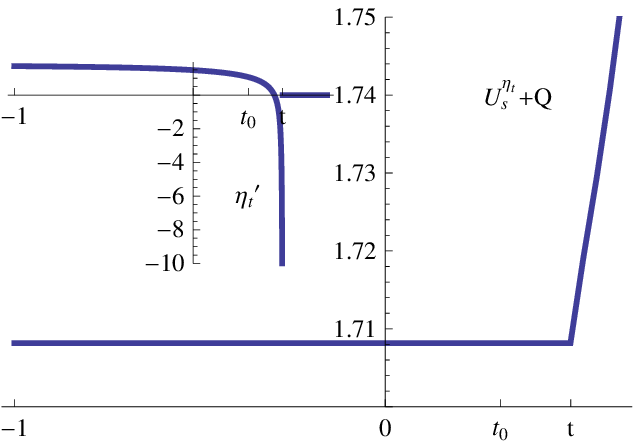}}
\end{minipage}
\begin{minipage}{0.4\linewidth}
\small
\begin{align*}
t &> t_0, \\
U_s^{\eta_t}(\PT{z}) + Q(\PT{z}) &\geq \mathcal{F}_s(\Sigma_t), \quad \text{on $\mathbb{S}^d\setminus\Sigma_t$,} \\
U_s^{\eta_t}(\PT{z}) + Q(\PT{z}) &= \mathcal{F}_s(\Sigma_t), \quad \text{on $\Sigma_t$,} \\
\eta_t^\prime &\ngeq 0, \quad \text{on $\Sigma_t$.} 
\end{align*}
\end{minipage}
\begin{minipage}{0.5\linewidth}
\centerline{\includegraphics[scale=1]{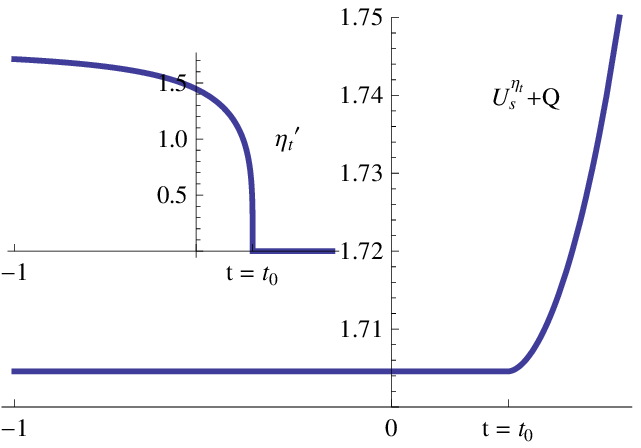}}
\end{minipage}
\begin{minipage}{0.4\linewidth}
\small
\begin{align*}
t &= t_0, \\
U_s^{\eta_t}(\PT{z}) + Q(\PT{z}) &\geq \mathcal{F}_s(\Sigma_t), \quad \text{on $\mathbb{S}^d\setminus\Sigma_t$,} \\
U_s^{\eta_t}(\PT{z}) + Q(\PT{z}) &= \mathcal{F}_s(\Sigma_t), \quad \text{on $\Sigma_t$,} \\
\eta_t^\prime &\geq 0, \quad \text{on $\Sigma_t$.} 
\end{align*}
\end{minipage}
\begin{minipage}{0.5\linewidth}
\centerline{\includegraphics[scale=1]{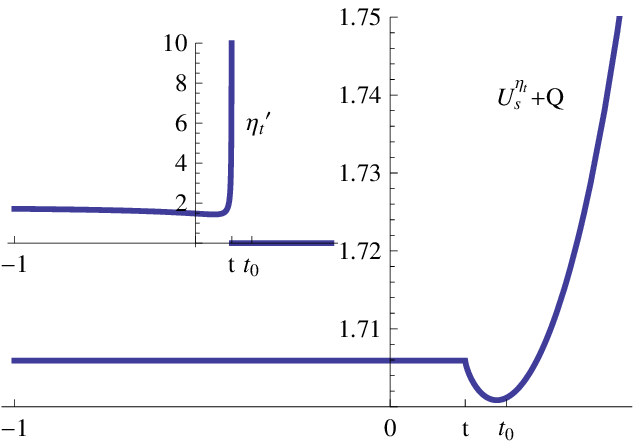}} 
\end{minipage}
\begin{minipage}{0.4\linewidth}
\small
\begin{align*}
t &< t_0, \\
U_s^{\eta_t}(\PT{z}) + Q(\PT{z}) &\ngeq \mathcal{F}_s(\Sigma_t), \quad \text{on $\mathbb{S}^d\setminus\Sigma_t$,} \\
U_s^{\eta_t}(\PT{z}) + Q(\PT{z}) &= \mathcal{F}_s(\Sigma_t), \quad \text{on $\Sigma_t$,} \\
\eta_t^\prime &\geq 0, \quad \text{on $\Sigma_t$.} 
\end{align*}
\end{minipage}
\caption{\label{fig1} The weighted $s$-potential of $\eta_t$ for different choices of $t$ ($t>t_0$, $t=t_0$, and $t<t_0$) versus altitude $\xi$ of $\PT{z}$ for $d=2$, $s=1/2$, $q=1$, and $R=3/2$, cf. Theorems \ref{SignEqThm} and \ref{MainThm}. Insets show the respective density $\eta_t^\prime$.}
\end{center}
\end{figure}

\subsection{The exceptional case $s=d-2$} 

In this case M. Riesz's approach \cite[Chapter IV]{L} has to be modified. 
Somewhat surprisingly it turns out, as shown in Lemmas \ref{NuLem:s.EQ.d-2} and \ref{EpsLem:s.EQ.d-2}, that the $s$-balayage measures from \eqref{bal}
\begin{equation}
\label{barBal}
\overline{\epsilon}_t \DEF \epsilon_{t,d-2} = \bal_{d-2}(\delta_{\PT{a}},\Sigma_t),
\qquad  \overline{\nu}_t \DEF \nu_{t,d-2} = \bal_{d-2}(\sigma,\Sigma_t)
\end{equation}
exist and both have a component that is uniformly distributed on the boundary of $\Sigma_t$. 
Moreover, unlike the case $d-2<s<d$, the density for $\mu_{Q_{\PT{a},q}}$, where $s=d-2$, does not vanish on the boundary of its support.
More precisely, on setting
\begin{equation*}
\beta_t (\PT{x}) \DEF \delta_t (u) \cdot \sigma_{d-1}(\overline{\PT{x}}), \qquad \PT{x}=(\sqrt{1-u^2}\,\overline{\PT{x}},u),
\end{equation*}
we obtain the following result.

\begin{thm} Let $d\geq3$. The signed $s$-equilibrium $\overline{\eta}_t$ on the spherical cap $\Sigma_t$
associated with $\overline{Q}_{\PT{a},q}(\PT{x})=q \, | \PT{x} - \PT{a} |^{2-d}$ is given by
\begin{equation*}
\overline{\eta}_t = \frac{\overline{\Phi}_{d-2}(t)}{W_s(\mathbb{S}^d)} \overline{\nu}_t - q \overline{\epsilon}_t, \qquad \overline{\Phi}_{d-2}(t) \DEF  W_s(\mathbb{S}^d) \frac{1+q\left\|\overline{\epsilon}_t\right\|}{\left\|\overline{\nu}_t\right\|},
\end{equation*}
where $\overline{\nu}_t$ and $\overline{\epsilon}_t$ are given in \eqref{barBal} and can be written as
\begin{align} \label{etabar}
\begin{split}
\dd \overline{\eta}_t(\PT{x}) &= \frac{1}{W_{d-2}(\mathbb{S}^d)} \left[ \overline{\Phi}_{d-2}(t) - \frac{q \left( R^2 - 1 \right)^2}{\left( R^2 - 2 R u + 1 \right)^{d/2+1}} \right] \dd \sigma_{d}\big|_{\Sigma_t}(\PT{x}) \\
&\phantom{=\pm}+ \frac{1-t}{2} \left( 1 - t^2 \right)^{d/2-1} \left[ \overline{\Phi}_{d-2}(t) - \frac{q \left( R + 1 \right)^2}{\left( R^2 - 2 R t + 1 \right)^{d/2}} \right] \dd \beta_t(\PT{x}).
\end{split}
\end{align}

For any fixed $t\in(-1,1)$, the following weak$^*$ convergence holds:
\begin{equation} \label{weak.star.conv}
\nu_{t,s}\stackrel{*}{\longrightarrow} \overline{\nu}_t , \qquad \epsilon_{t,s} \stackrel{*}{\longrightarrow}
\overline{\epsilon}_t, \qquad \text{as $s\to(d-2)^+$.}
\end{equation}

The function $\overline{\Phi}_{d-2}(t)$ has precisely one global minimum
$t_0\in (-1,1]$. This minimum is either the unique solution $ t_0\in (-1,1)$ of the equation
\begin{equation}
\overline{\Phi}_{d-2} (t)  = \frac{q \left( R + 1 \right)^2}{\left( R^2 - 2 R t + 1 \right)^{d/2}},
\label{FuncEqCondEx}
\end{equation}
or $t_0=1$ when such a solution does not exist. Moreover, $t_0=\max\{ t: \overline{\eta}_t \geq 0 \}$.

The extremal measure $\mu_{\overline{Q}_{\PT{a},q}}$  on $\mathbb{S}^d$
is given by
\begin{equation} \label{etabarzero}
\dd \mu_{\overline{Q}_{\PT{a},q}}(\PT{x}) = \dd \overline{\eta}_{t_0}(\PT{x}) = \frac{\overline{\Phi}_{d-2}(t_0)}{W_{d-2}(\mathbb{S}^d)} \left[ 1 - \frac{\left( R - 1 \right)^2 \left( R^2 - 2 R t_0 + 1 \right)^{d/2}}{\left( R^2 - 2 R u + 1 \right)^{d/2+1}}  \right] \dd \sigma_{d}\big|_{\Sigma_{t_0}}(\PT{x}), 
\end{equation}
and has as support the spherical
cap $\Sigma_{t_0}$.  \label{ExcepThm}
\end{thm}

In Lemmas \ref{NuLem:s.EQ.d-2} and \ref{EpsLem:s.EQ.d-2} we give the $s$-potentials of the balayage measures $\overline{\nu}_t$ and $\overline{\epsilon}_t$ from which the weighted $s$-potential of $\overline{\eta}_t$ at every $\PT{z}\in\mathbb{S}^d$ can be easily obtained.

\begin{rmk}
As can be seen from \eqref{etabar}, depending on the sign of the coefficient of $\beta_t$, the signed $s$-equilibrium $\overline{\eta}_t$ has positive or negative charge on $\partial\Sigma_t$ unless $t=t_0$, in which case the charge on the boundary disappears (see Figure \ref{fig2}). 
\end{rmk}

Next, we describe the results when $d=2$ and $s=0$. The external field in this case is $\overline{Q}(\PT{x}) = \overline{Q}_{\PT{a},q}(\PT{x}) = q \log ( 1 / | \PT{x} - \PT{a} | )$. The total mass of the balayage measures $\overline{\nu}_{t,0}$ and $\overline{\epsilon}_{t,0}$ is preserved, so $\|\overline{\nu}_{t,0}\|=\|\overline{\epsilon}_{t,0}\|=1$. Thus, the function $\overline{\Phi}_{d-2}(t)$ reduces to $\overline{\Phi}_{d-2}(t) = 1 + q$. The Mhaskar-Saff functional $\mathcal{F}_0(K)$ from \eqref{Functional}, now defined for compact sets $K\subset\mathbb{S}^2$ with positive logarithmic capacity $\CAP_0(K) = \exp\{-W_0(K)\}$, uses the logarithmic energy
\begin{equation*} 
W_0(K) = \lim_{s\to0^+} \frac{\dd W_s(K)}{\dd s} \Big|_{s=0}. 
\end{equation*}
However, $\mathcal{F}_0(\Sigma_t)$ is no longer equal to $\overline{\Phi}_{d-2}(t)$ (cf. Remark \ref{rmk:F.s.EQ.Phi.s} and Lemma \ref{lem:Mhaskar-Saff.log}). 
For $K=\mathbb{S}^2$ we have $W_0(\mathbb{S}^2) = 1 / 2 - \log 2 < 0$.
Since Theorem \ref{ConnThm} can be extended to $s=0$ if $d=2$, we deduce that $S_{\overline{Q}}\DEF \supp(\mu_{\overline{Q}})$ will be a spherical cap $\Sigma_{t_0}$. Direct calculations show that the Mhaskar-Saff functional $\mathcal{F}_0$ for spherical caps is still minimized for $S_{\overline{Q}}$. 
Figure \ref{fig2} shows the qualitative behavior for the weighted potential in the logarithmic case. (Note, that for $t\neq t_0$ the tangent line to the graph of the weighted logarithmic potential at $\xi \to t^+$ is {\bf not} vertical like in the case $d-2<s<d$ (cf. Figure \ref{fig1}), but it becomes horizontal if $t=t_0<1$.) 

\begin{figure}[ht]
\begin{center}
\begin{minipage}{0.5\linewidth}
\centerline{\includegraphics[scale=1]{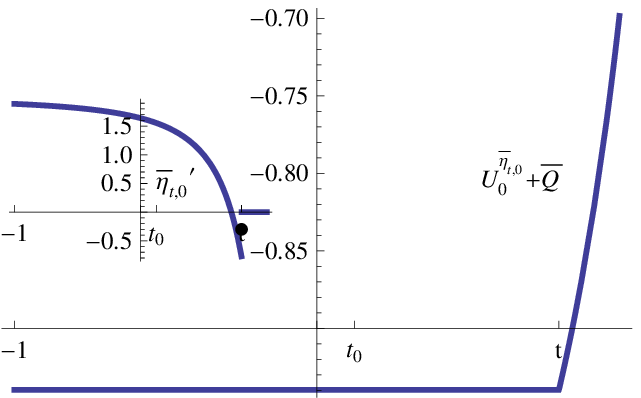}}
\end{minipage}
\begin{minipage}{0.4\linewidth}
\small
\begin{align*}
t &> t_0, \\
U_0^{\overline{\eta}_{t,0}}(\PT{z}) + \overline{Q}(\PT{z}) &\geq \mathcal{F}_0(\Sigma_t) \quad \text{on $\mathbb{S}^d\setminus\Sigma_t$,} \\
U_0^{\overline{\eta}_{t,0}}(\PT{z}) + \overline{Q}(\PT{z}) &= \mathcal{F}_0(\Sigma_t) \quad \text{on $\Sigma_t$,} \\
\overline{\eta}_{t,0}^\prime &\ngeq 0 \quad \text{on $\Sigma_t$.} 
\end{align*}
\end{minipage}
\begin{minipage}{0.5\linewidth}
\centerline{\includegraphics[scale=1]{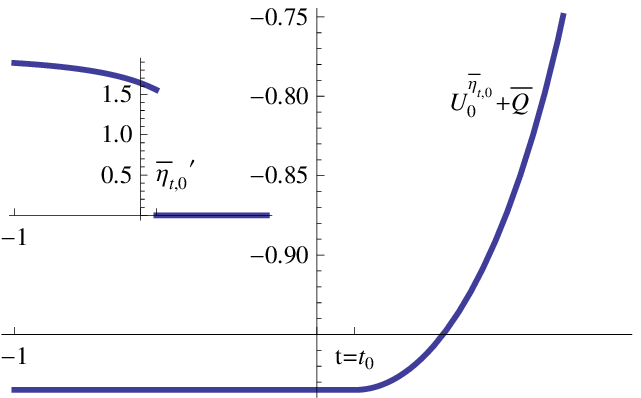}}
\end{minipage}
\begin{minipage}{0.4\linewidth}
\small
\begin{align*}
t &= t_0, \\
U_0^{\overline{\eta}_{t,0}}(\PT{z}) + \overline{Q}(\PT{z}) &\geq \mathcal{F}_0(\Sigma_t) \quad \text{on $\mathbb{S}^d\setminus\Sigma_t$,} \\
U_0^{\overline{\eta}_{t,0}}(\PT{z}) + \overline{Q}(\PT{z}) &= \mathcal{F}_0(\Sigma_t) \quad \text{on $\Sigma_t$,} \\
\overline{\eta}_{t,0}^\prime &\geq 0 \quad \text{on $\Sigma_t$.} 
\end{align*}
\end{minipage}
\begin{minipage}{0.5\linewidth}
\centerline{\includegraphics[scale=1]{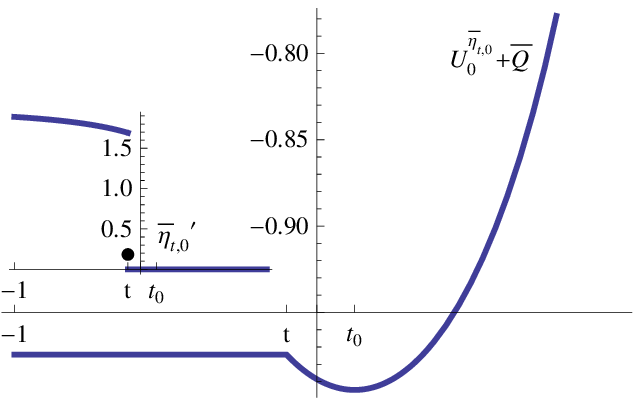}} 
\end{minipage}
\begin{minipage}{0.4\linewidth}
\small
\begin{align*}
t &< t_0, \\
U_0^{\overline{\eta}_{t,0}}(\PT{z}) + \overline{Q}(\PT{z}) &\ngeq \mathcal{F}_0(\Sigma_t), \quad \text{on $\mathbb{S}^d\setminus\Sigma_t$,} \\
U_0^{\overline{\eta}_{t,0}}(\PT{z}) + \overline{Q}(\PT{z}) &= \mathcal{F}_0(\Sigma_t), \quad \text{on $\Sigma_t$,} \\
\overline{\eta}_{t,0}^\prime &\geq 0, \quad \text{on $\Sigma_t$.} 
\end{align*}
\end{minipage}
\caption{\label{fig2} The weighted logarithmic potential of $\overline{\eta}_{t,0}$ for different choices of $t$ ($t>t_0$, $t=t_0$, and $t<t_0$) versus altitude $\xi$ of $\PT{z}$ for $d=2$, $s=0$, $q=1$, and $R=2$, cf. Theorem \ref{ExcepThm.log}. Insets show the respective density $\overline{\eta}_{t,0}^\prime$. The black dot indicates the component on the boundary.}
\end{center}
\end{figure}

\begin{thm} \label{ExcepThm.log} Let $d=2$ and $s=0$. The signed $s$-equilibrium $\overline{\eta}_{t,0}$ on the spherical cap $\Sigma_t$
associated with $\overline{Q}_{\PT{a},q}(\PT{x})=q \log ( 1 / | \PT{x} - \PT{a} | )$ is given by
\begin{equation*}
\overline{\eta}_{t,0} = \left( 1 + q \right) \overline{\nu}_{t,0} - q \overline{\epsilon}_{t,0},
\end{equation*}
where $\overline{\nu}_{t,0} = \bal_0(\sigma_2, \Sigma_t)$ and $\overline{\epsilon}_{t,0} = \bal_0(\delta_{\PT{a}},\Sigma_t)$. It can be written as
\begin{align} \label{etabar.log}
\begin{split}
\dd \overline{\eta}_{t,0}(\PT{x}) &= \left[ 1 + q - \frac{q \left( R^2 - 1 \right)^2}{\left( R^2 - 2 R u + 1 \right)^{2}} \right] \dd \sigma_{2}\big|_{\Sigma_t}(\PT{x}) \\
&\phantom{=\pm}+ \frac{1-t}{2} \left[ 1 + q - \frac{q \left( R + 1 \right)^2}{R^2 - 2 R t + 1 } \right] \dd \beta_t(\PT{x}),
\end{split}
\end{align}
where $\PT{x} = ( \sqrt{1 - u^2} \, \overline{\PT{x}}, u )$, $\overline{\PT{x}} \in \mathbb{S}^{1}$.

The weighted logarithmic potential of $\overline{\eta}_{t,0}$ satisfies
\begin{align*} 
U_0^{\overline{\eta}_{t,0}}(\PT{z}) + \overline{Q}_{\PT{a},q}(\PT{z}) &= \mathcal{F}_0(\Sigma_t), \qquad \PT{z} \in \Sigma_t, \\ 
U_0^{\overline{\eta}_{t,0}}(\PT{z}) + \overline{Q}_{\PT{a},q}(\PT{z}) &= \mathcal{F}_0(\Sigma_t) + \frac{1}{2} \log\frac{1+t}{1+\xi} + \frac{q}{2} \log \frac{R^2 - 2 R t + 1}{R^2 - 2 R \xi + 1}, \qquad \PT{z} \in \mathbb{S}^2 \setminus \Sigma_t, 
\end{align*}
where $\mathcal{F}_0(\Sigma_t)$ is given below in Lemma {\rm \ref{lem:Mhaskar-Saff.log}}. 

The Mhaskar-Saff functional $\mathcal{F}_0$ is minimized for $\Sigma_{t_0}$, where either $t_0\in(-1,1]$ is the unique solution of the equation
\begin{equation} \label{log.t.rel}
1 + q = q \left( R + 1 \right)^2 \big/ \left( R^2 - 2 R t + 1 \right),
\end{equation}
or $t_0=1$ if such a solution does not exists. Moreover, $t_0=\max\{ t: \overline{\eta}_{t,0} \geq 0 \}$.

The logarithmic extremal measure $\mu_{\overline{Q}_{\PT{a},q}}$ on $\mathbb{S}^2$ 
is given by
\begin{equation} \label{etabarzero.log}
\dd \mu_{\overline{Q}_{\PT{a},q}}(\PT{x}) = \dd \overline{\eta}_{t_0,0}(\PT{x}) = \left[ 1 + q - \frac{q \left( R^2 - 1 \right)^2}{\left( R^2 - 2 R u + 1 \right)^{2}} \right] \dd \sigma_{2}\big|_{\Sigma_{t_0}}(\PT{x}), 
\end{equation}
and has as support the spherical cap $\Sigma_{t_0}$.
\end{thm}

\begin{rmk}
Given $R$ and $q$, relation \eqref{log.t.rel} immediately enables us to find the support $\Sigma_{t_0}$ of the logarithmic extremal measure $\mu_{\overline{Q}_{\PT{a},q}}$ on $\mathbb{S}^2$: 
\begin{equation}
t_0 = \min\left\{ 1, \frac{\left( R + 1 \right)^2}{2 R \left( 1 + q \right)} - 1 \right\} = \min\left\{ 1, 1 - \frac{4 q R - \left( R - 1 \right)^2}{2 R \left( 1 + q \right)} \right\}.
\end{equation}
\end{rmk}

\begin{rmk}
In general, the density $\overline{\eta}_{t_0,0}^\prime(u)$ in \eqref{etabarzero.log} does not vanish on the boundary of $\Sigma_{t_0}$. In fact, if $t_0\in(-1,1)$, then 
\begin{equation}
\lim_{u\to t_0} \overline{\eta}_{t_0,0}^\prime(u) = 1 + q - \frac{q \left( R^2 - 1 \right)^2}{\left( R^2 - 2 R t_0 + 1 \right)^{2}} = \frac{1+q}{q} \cdot \frac{4 q R - \left( R - 1 \right)^2}{\left( R + 1 \right)^2} > 0.
\end{equation}
\end{rmk}

\subsection{Axis-supported external fields}
It is well known that the balayage of a measure can be represented
as a superposition of balayages of Dirac-delta measures. Using this, 
we extend our results to external fields that are axis-supported $s$-potentials.

\begin{defn} We call an external field $Q$ {\it positive-axis supported}, if
\begin{equation}
Q(\PT{x})= \int \frac{\dd \lambda(R)}{|\PT{x}-R\PT{p}|^s}, \qquad \PT{x}\in \mathbb{S}^d \label{AxisSupp},
\end{equation}
for some finite positive measure $\lambda$ supported on a compact
subset of $(0,\infty)$.
\end{defn}

\begin{rmk} Since
\begin{equation}
\bal_s(\delta_{(1/R)\PT{p}},\mathbb{S}^d) = R^s \bal_s(\delta_{R\PT{p}},\mathbb{S}^d), \label{BalInversion}
\end{equation}
we can restrict ourselves to measures $\lambda$ with support in
$[1,\infty)$. It is possible to generalize the setting to fields
supported on both the negative and positive polar axis as well. This
generalization shall be reserved for a later occasion.
\end{rmk}

We begin with a result that establishes the existence of the signed equilibrium measure $\tilde{\eta}_{\lambda}$ on $\mathbb{S}^d$ associated with the axis-supported external field $Q$. Furthermore, a necessary and sufficient condition for coincidence of $\tilde{\eta}_{\lambda}$ and the extremal measure $\mu_Q$ on $\mathbb{S}^d$ is given.

\begin{thm} \label{SignEq.axis.supp}
Let $0<s<d$ and $Q$ be as in \eqref{AxisSupp} with $\supp(\lambda)\subset[1,\infty)$. Then
\begin{equation}
\dd \tilde{\eta}_{\lambda}(\PT{x}) = \frac{1}{W_{s}(\mathbb{S}^d)} \left\{ \mathcal{F}_s(\mathbb{S}^d) - \int \frac{\left( R^2 - 1 \right)^{d-s}}{\left( R^2 - 2 R u + 1 \right)^{d-s/2}} \dd \lambda(R) \right\} \dd \sigma(\PT{x}).
\end{equation}
Moreover, $\supp(\mu_Q)=\mathbb{S}^d$ (that is $\mu_Q=\tilde{\eta}_{\lambda}$) if and only if
\begin{equation}
\mathcal{F}_s(\mathbb{S}^d) \geq \int \frac{\left( R + 1 \right)^{d-s}}{\left( R - 1 \right)^d} \dd \lambda(R).
\end{equation}
\end{thm}

The above result, appropriately modified, also holds for the logarithmic case. We shall use the Mhaskar-Saff functional for the logarithmic case
\begin{equation}
\mathcal{F}_0(K) \DEF W_0(K) + \int Q(\PT{x}) \dd \mu_{K,0}(\PT{x}), 
\end{equation}
where $K$ is a compact subset of $\mathbb{S}^d$ with finite logarithmic energy $W_0(K)$ and $\mu_{K,0}$ is the logarithmic extremal measure on $K$ (without external field). In particular,
\begin{equation} \label{W.0.d}
W_0(\mathbb{S}^d) = \lim_{s\to0^+} \frac{\dd W_s(\mathbb{S}^d)}{\dd s} \Big|_{s=0} = - \log 2 - \frac{1}{2} \digammafcn(d/2) - \frac{1}{2} \digammafcn(d),
\end{equation}
where $\digammafcn(z) := \gammafcn^\prime(z) / \gammafcn(z)$ denotes the digamma function.

\begin{thm} \label{SignEq.axis.supp.log}
Let $d=2$, $s=0$, and $Q$ be the positive-axis supported external field 
\begin{equation}
Q(\PT{z}) = \int \log \frac{1}{\left| \PT{z} - \PT{a} \right|} \dd \lambda(R)
\end{equation}
with $\supp(\lambda)\subset[1,\infty)$. Then the signed logarithmic equilibrium measure $\tilde{\eta}_{\lambda,0}$ on $\mathbb{S}^d$ associated with $Q$ is given by
\begin{equation}
\dd \tilde{\eta}_{\lambda,0}(\PT{x}) = \left\{ 1 + \left\| \lambda \right\| - \int \frac{\left( R^2 - 1 \right)^{d}}{\left( R^2 - 2 R u + 1 \right)^{d}} \dd \lambda(R) \right\} \dd \sigma(\PT{x}).
\end{equation}
Its weighted logarithmic potential is given by
\begin{equation}
U_0^{\tilde{\eta}_{\lambda,0}}(\PT{z}) + Q(\PT{z}) = \mathcal{F}_0(\mathbb{S}^d), \qquad \PT{z} \in \mathbb{S}^d.
\end{equation}
Moreover, $\supp(\mu_Q)=\mathbb{S}^d$ (that is the logarithmic extremal measure $\mu_Q$ on $\mathbb{S}^d$ coincides with $\tilde{\eta}_{\lambda,0}$) if and only if
\begin{equation}
1 + \left\| \lambda \right\| \geq \int \frac{\left( R + 1 \right)^{d}}{\left( R - 1 \right)^d} \dd \lambda(R).
\end{equation}
\end{thm}

The next assertion deals with the signed equilibrium measure $\tilde{\eta}_{t}$ on a spherical cap $\Sigma_t$ for $Q$ of the form \eqref{AxisSupp}.

\begin{thm} \label{SignEqThm2}
Let $d-2<s<d$ and $Q$ be as in \eqref{AxisSupp} with
$\supp(\lambda)\subset[1,\infty)$. The
signed $s$-equilibrium $\tilde{\eta}_t$ on the spherical cap $\Sigma_t$
associated with $Q$ is given by
\begin{equation}
\tilde{\eta}_t = \frac{\tilde{\Phi}_s(t)}{W_s(\mathbb{S}^d)} \, \nu_t - \tilde{\epsilon}_t, \qquad \tilde{\Phi}_s(t) \DEF W_s(\mathbb{S}^d) \frac{1+ \left\| \tilde{\epsilon}_t \right\|}{\left\| \nu_t \right\|}, \label{etatilde}
\end{equation}
where $\nu_t$ is defined in \eqref{bal} and
\begin{equation}
\tilde{\epsilon}_t \DEF \bal_s(\lambda,\Sigma_t) =\int \bal_s(\delta_{R\PT{p}},\Sigma_t)\, \dd\lambda(R).
\end{equation}
The signed $s$-equilibrium $\tilde{\eta}_t$ can be written as 
\begin{equation*}
\dd \tilde{\eta}_t(\PT{x}) = \tilde{\eta}_{t}^{\prime}(u,R) \frac{\omega_{d-1}}{\omega_d} (1-u^2)^{d/2-1}\, \dd u \dd \sigma_{d-1} (\overline{\PT{x}}), \qquad \PT{x} \in \Sigma_t,
\end{equation*}
where
\begin{equation*}
\begin{split} 
&\tilde{\eta}_{t}^{\prime}(u,R) = \frac{1}{W_s(\mathbb{S}^d)} \frac{\gammafcn(d/2)}{\gammafcn(d-s/2)}
\left( \frac{1-t}{1-u} \right)^{d/2} \left( \frac{t-u}{1-t} \right)^{(s-d)/2} \\
&\phantom{=\times}\times \Bigg\{ \tilde{\Phi}_s(t)
\HypergeomReg{2}{1}{1,d/2}{1-(d-s)/2}{\frac{t-u}{1-u}}  \\
&\phantom{=\times\pm}- \int \frac{\left( R + 1 \right)^{d-s}}{\left( R^2 - 2 R t + 1 \right)^{d/2}}
\HypergeomReg{2}{1}{1,d/2}{1-(d-s)/2}{\frac{\left(R-1\right)^{2}}{R^2 - 2 R t + 1}
\, \frac{t-u}{1-u}} \, \dd \lambda(R)  \Bigg\}. 
\end{split}
\end{equation*}
Furthermore, the function $\tilde{\Phi}_s$
has precisely one global minimum in $(-1,1]$. This minimum is either the unique solution $t_\lambda\in(-1,1)$ of the equation
\begin{equation}
\tilde{\Phi}_s(t) = \int \frac{\left( R + 1 \right)^{d-s}}{\left( R^2 - 2 R t + 1 \right)^{d/2}} \dd \lambda (R),
\end{equation}
or $t_\lambda = 1$ when such a solution does not exist. Moreover, $t_\lambda \DEF \max \{ t : \tilde{\eta}_t \geq 0 \}$, $\mu_Q=\tilde{\eta}_{t_\lambda}$, and $\supp(\mu_Q)=\Sigma_{t_\lambda}$, where $\mu_Q$ is the extremal
measure on $\mathbb{S}^d$ associated with $Q$. 
\end{thm}

Theorem \ref{SignEqThm2} can be also extended to the case $s=d-2$ and $d\geq3$. We present

\begin{thm} \label{SignEqThm2.s.EQ.d-2}
Let $s=d-2$, $d\geq3$ and 
 $Q$ be as in \eqref{AxisSupp} with
$\supp(\lambda)\subset[1,\infty)$. The
signed $s$-equilibrium $\tilde{\overline{\eta}}_t$ on the spherical cap $\Sigma_t$
associated with $Q$ is given by
\begin{equation}
\tilde{\overline{\eta}}_t = \frac{\tilde{\overline{\Phi}}_{d-2}(t)}{W_s(\mathbb{S}^d)} \overline{\nu}_t - \tilde{\overline{\epsilon}}_t, \qquad \tilde{\overline{\Phi}}_{d-2}(t) \DEF W_s(\mathbb{S}^d) \frac{ 1 + \left\| \tilde{\overline{\epsilon}}_t \right\|}{\left\| \overline{\nu}_t \right\|}, \label{etatilde.s.EQ.d-2}
\end{equation}
where $\overline{\nu}_t$ is defined in \eqref{barBal} and
\begin{equation}
\tilde{\overline{\epsilon}}_t \DEF \bal_{d-2}(\lambda,\Sigma_t) = \int \bal_{d-2}(\delta_{R\PT{p}},\Sigma_t)\, \dd\lambda(R).
\end{equation}
The signed $s$-measure $\tilde{\overline{\eta}}_t$ can be written as
\begin{equation} \label{tilde.bar.eta.thm}
\dd \tilde{\overline{\eta}}_t(\PT{x}) = g(u) \dd \sigma_d \big|_{\Sigma_t}(\PT{x}) + h(u) \dd \beta_t(\PT{x}),
\end{equation}
where, when using Lemmas \ref{NuLem:s.EQ.d-2} and \ref{EpsLem:s.EQ.d-2}, we have for $-1 \leq u \leq t$ 
\begin{align}
g(u) &= \frac{1}{W_{d-2}(\mathbb{S}^d)} \left[ \tilde{\overline{\Phi}}_{d-2}(t) - \int \frac{\left( R^2 - 1 \right)^2}{\left( R^2 - 2 R u + 1 \right)^{d/2+1}} \dd \lambda(R) \right], \label{eta.prime.prime.thm} \\
h(u) &= \frac{1-t}{2} \left[ \tilde{\overline{\Phi}}_{d-2}(t) - \int \frac{\left( R + 1 \right)^2}{\left( R^2 - 2 R t + 1 \right)^{d/2}} \dd \lambda(R) \right] \left( 1 - t^2 \right)^{d/2-1}. \label{eta.prime.prime.prime.thm}
\end{align}

For any fixed $t\in(-1,1)$, the following weak$^*$ convergence holds:
\begin{equation} \label{weak.star.axis-supp}
\tilde{\epsilon}_{t,s} \stackrel{*}{\longrightarrow} \tilde{\overline{\epsilon}}_t \qquad \text{as $s\to(d-2)^+$.}
\end{equation}

The function $\tilde{\overline{\Phi}}_{d-2}$ has precisely one global minimum in $(-1,1]$. This minimum is either the unique solution $t_\lambda\in(-1,1)$ of the equation
\begin{equation}
\tilde{\overline{\Phi}}_{d-2}(t) = \int \frac{\left( R + 1 \right)^{2}}{\left( R^2 - 2 R t + 1 \right)^{d/2}} \dd \lambda (R),
\end{equation}
or $t_\lambda = 1$ when such a solution does not exist. Moreover, $t_\lambda \DEF \max \{ t : \tilde{\overline{\eta}}_t \geq 0 \}$, $\mu_Q=\tilde{\overline{\eta}}_{t_\lambda}$, and $\supp(\mu_Q)=\Sigma_{t_\lambda}$, where $\mu_Q$ is the extremal
measure on $\mathbb{S}^d$. 
\end{thm}

Next, we describe the results when $d=2$ and $s=0$. The external field in this case is 
\begin{equation} \label{log.axis.ext.field}
\tilde{\overline{Q}}(\PT{x}) = \tilde{\overline{Q}}_{\PT{a},q}(\PT{x}) = \int \log \frac{1}{\left| \PT{x} - R \PT{p} \right|} \dd \lambda(R), \qquad \PT{x} \in \mathbb{S}^2,
\end{equation}
for some finite positive measure $\lambda$ supported on a compact subset of $[1,\infty)$.
%

We show a result, which generalizes Theorem \ref{ExcepThm.log}. 

\begin{thm} \label{ExcepThm.log.axis} Let $d=2$ and $s=0$. Let $\tilde{\overline{Q}}$ be as in \eqref{log.axis.ext.field} with
$\supp(\lambda)\subset[1,\infty)$. The signed logarithmic equilibrium $\tilde{\overline{\eta}}_{t,0}$ on the spherical cap $\Sigma_t$
associated with $\tilde{\overline{Q}}$ is given by
\begin{equation}
\tilde{\overline{\eta}}_{t,0} = \left( 1 + \left\| \lambda \right\| \right) \overline{\nu}_{t,0} - \tilde{\overline{\epsilon}}_{t,0},
\end{equation}
where $\overline{\nu}_{t,0} = \bal_0(\sigma_2, \Sigma_t)$, $\overline{\epsilon}_{t,0} = \bal_0(\delta_{R\PT{p}},\Sigma_t)$, and
\begin{equation}
\tilde{\overline{\epsilon}}_{t,0} \DEF \bal_{0}(\lambda,\Sigma_t) = \int \bal_0(\delta_{R\PT{p}},\Sigma_t) \, \dd\lambda(R).
\end{equation}
It can be written as
\begin{align} \label{etabar.log.axis}
\begin{split}
\dd \tilde{\overline{\eta}}_{t,0}(\PT{x}) &= \left[ 1 + \left\| \lambda \right\| - \int \frac{\left( R^2 - 1 \right)^2}{\left( R^2 - 2 R u + 1 \right)^{2}} \dd \lambda(R) \right] \dd \sigma_{2}\big|_{\Sigma_t}(\PT{x}) \\
&\phantom{=\pm}+ \frac{1-t}{2} \left[ 1 + \left\| \lambda \right\| - \int \frac{\left( R + 1 \right)^2}{R^2 - 2 R t + 1 } \dd \lambda(R) \right] \dd \beta_t(\PT{x}).
\end{split}
\end{align}
The weighted logarithmic potential of $\tilde{\overline{\eta}}_{t,0}$ satisfies
\begin{align*}
U_0^{\tilde{\overline{\eta}}_{t,0}}(\PT{z}) + \tilde{\overline{Q}}(\PT{z}) &= W_0(\Sigma_t) + \int \tilde{\overline{Q}} \dd \mu_{\Sigma_t,0} \FED \tilde{\overline{\mathcal{F}}}_0(\Sigma_t), \quad \PT{z} \in \Sigma_t, \\
U_0^{\tilde{\overline{\eta}}_{t,0}}(\PT{z}) + \tilde{\overline{Q}}(\PT{z}) &= \tilde{\overline{\mathcal{F}}}_0(\Sigma_t) + \frac{1}{2} \log \frac{1+t}{1+\xi} + \int \frac{1}{2} \log \frac{R^2 - 2 R t + 1}{R^2 - 2 R \xi + 1} \dd \lambda(R), \quad \PT{z} \in \mathbb{S}^2 \setminus \Sigma_t.
\end{align*}

The Mhaskar-Saff functional $\tilde{\overline{\mathcal{F}}}_0$ (explicitly given in \eqref{Mashkar-Saff.functional.log.case}) is minimized for $\Sigma_{t_\lambda}$, where either $t_\lambda\in(-1,1)$ is the unique solution of the equation
\begin{equation} \label{log.t.rel.axis}
1 + \left\| \lambda \right\| = \int \frac{\left( R + 1 \right)^2}{R^2 - 2 R t + 1} \dd \lambda(R),
\end{equation}
or $t_\lambda=1$ if such a solution does not exists. Moreover, $t_\lambda=\max\{ t: \tilde{\overline{\eta}}_t \geq 0 \}$, $\supp(\mu_{\tilde{\overline{Q}}})=\Sigma_{t_\lambda}$, and $\mu_{\tilde{\overline{Q}}}=\tilde{\overline{\eta}}_{t_\lambda,0}$.
\end{thm}

\begin{rmk}
In general, the density $\tilde{\overline{\eta}}_{t_\lambda,0}^\prime(u)$ with respect to $\sigma_2|_{\Sigma_{t_\lambda}}$ of the extremal measure $\mu_Q$ on $\mathbb{S}^2$ in Theorem \ref{ExcepThm.log.axis} does not vanish on the boundary of $\Sigma_{t_\lambda}$. In fact, if $t_\lambda\in(-1,1)$, then 
\begin{equation}
\lim_{u\to t_\lambda} {\tilde{\overline{\eta}}}^\prime_{t_\lambda,0}(u) = 1 + \left\| \lambda \right\| - \int \frac{\left( R^2 - 1 \right)^2 \dd \lambda(R)}{\left( R^2 - 2 R t_\lambda + 1 \right)^{2}} = \int \frac{2 R \left( 1 - t_\lambda \right) \dd \lambda(R)}{\left( R^2 - 2 R t_\lambda + 1 \right)^{2}} > 0.
\end{equation}
\end{rmk}

The remainder of this paper is structured as follows. In Section \ref{sec:2} we show the
uniqueness of the signed equilibrium and prove Theorem \ref{SignEq}
and Corollary \ref{Cor}. In Section \ref{sec:3} a suitable Kelvin transform of
points and measures is considered and explicit formulas for the
densities of the measures in \eqref{eta} are found in Lemmas
\ref{NuLem} and \ref{EpsilonLem}. The norms of these measures are
computed in Section \ref{sec:4}. The proofs of Theorems \ref{FuncOptThm},
\ref{ConnThm}, and \ref{MainThm} are given in Section \ref{sec:5}. 
The weighted $s$-potential of the signed $s$-equilibrium is given in 
Section \ref{sec:6}. Section \ref{sec:7} considers the special case $s=d-2$ and the proofs 
of Theorems \ref{ExcepThm} and \ref{ExcepThm.log} are provided. Finally, in Section \ref{sec:8} we prove the generalization of the results to axis-supported external fields.

\section{Signed equilibrium associated with an external field}
\label{sec:2}

First, we consider some preliminaries on the Kelvin transformation
(spherical inversion) of points and measures.
Inversion in a sphere is a basic technique in electrostatics 
(method of electrical images, cf. Jackson~\cite{J}) and in general 
in potential theory (cf. Kellog~\cite{kellog:1967} and Landkof~\cite{L}). 
Kelvin transformation (of a function) is linear, preserves harmonicity 
(in the classical case), and preserves positivity.
We shall make use of this method and of balayage to conveniently infer 
representations of the signed equilibrium associated with an external field 
from known results.

\subsection{The Kelvin transformation.}
\label{sec:kelvin.transf.}
Let us denote by $\kelvin_R$ the Kelvin transformation
(stereographic projection) with center $\PT{a}=(\PT{0},R)$ and
radius $\sqrt{R^2-1}$, that is for any point $\PT{x}\in
\mathbb{R}^{d+1}$ the image $\PT{x}^*\DEF\kelvin_R(\PT{x})$ lies on
a ray stemming from $\PT{a}$, and passing through $\PT{x}$ such that
\begin{equation}
\left| \PT{x} - \PT{a} \right| \cdot \left| \PT{x}^* - \PT{a} \right| = R^2 - 1. \label{KelTr}
\end{equation}
Thus, the transformation of the distance is given by the formula
\begin{equation}
\left| \PT{x}^{*} - \PT{y}^{*} \right| = \left( R^{2} - 1
\right)\frac{\left| \PT{x} - \PT{y} \right|}{\left| \PT{x} - \PT{a}
\right| \left| \PT{y} - \PT{a} \right|}, \qquad
\PT{x},\PT{y}\in\mathbb{S}^{d}. \label{KelTrDist}
\end{equation}

It is easy to see that $\kelvin_R (\mathbb{S}^d) =\mathbb{S}^d$,
where $\kelvin_R$ sends the spherical cap \linebreak
$A_R\DEF\{(\sqrt{1-u^2}\,\overline{\PT{x}},u):1/R \leq u \leq 1,
\overline{\PT{x}} \in \mathbb{S}^{d-1} \}$ to
$B_R\DEF\{(\sqrt{1-u^2}\,\overline{\PT{x}},u) : -1\leq u \leq 1/R,
\overline{\PT{x}} \in \mathbb{S}^{d-1} \}$ and vice versa, with the
points on the boundary being fixed. In particular, the North Pole
$\PT{p}=(\PT{0},1)$ goes to the South Pole $\PT{q}\DEF(\PT{0},-1)$.
The image of $\PT{x}=(\sqrt{1-u^2}\,\overline{\PT{x}},u)$ is
$\PT{x}^*=(\sqrt{1-({u^*})^2}\,\overline{\PT{x}},u^*)$, where the
relation between $u$ and $u^*$ is given by
\begin{equation}
1+u^* = \frac{\left( R + 1 \right)^2}{R^2-2Ru+1} \left( 1 - u \right). \label{u_rel}
\end{equation}
The last equation is derived from the similar
triangles proportion
\begin{equation*}
\left|\PT{x}^*-\PT{q}\right| \big/ \left|\PT{q}-\PT{a}\right| =
\left|\PT{x}-\PT{p}\right| \big/ \left|\PT{x}-\PT{a}\right|
\end{equation*}
and the formulas $\left|\PT{x}^*-\PT{q}\right|^2 = 2 \left( 1 + u^{*} \right)$, $\left|\PT{x}-\PT{p}\right|^2 = 2 \left( 1 - u \right)$, $\left|\PT{q}-\PT{a}\right| = R + 1$, and $\left|\PT{x}-\PT{a}\right|^2 = R^2 - 2 R u + 1$.
Finally, we point out that
\begin{equation}
\left| \PT{x}^* - \PT{a} \right|^{-d} \dd\sigma(\PT{x}^*) = \left| \PT{x} - \PT{a} \right|^{-d} \dd\sigma(\PT{x}),
\label{KelTrMeas}
\end{equation}
which can be easily seen from the relation $\left( \PT{x}^* - \PT{a} \right) \big/ \left| \PT{x}^* - \PT{a} \right| = \left( \PT{x} - \PT{a} \right) \big/ \left| \PT{x} - \PT{a} \right|$.

Next, we recall that given a measure $\lambda$ with no point mass at
$\PT{a}$, its Kelvin transformation (associated with a fixed $s$)
$\lambda^*=\kelvinMEAS_{R,s} (\lambda)$ is a measure defined by
\begin{equation}
\dd\lambda^* (\PT{x}^*) \DEF
\frac{\left(R^2-1\right)^{s/2}}{\left|\PT{x}-\PT{a}\right|^s}
\dd\lambda(\PT{x}). \label{KelMeas}
\end{equation}
The $s$-potentials of the two measures are related as follows (see, for
example, \cite[Section 5, Equation (5.1)]{DS}
\begin{equation}
U_s^{\lambda^*}(\PT{x}^*) = \int \frac{\dd\lambda^*(\PT{y}^*)}{\left|\PT{x}^*-\PT{y}^*\right|^s} = \int \frac{\left|\PT{x}-\PT{a}\right|^s \dd\lambda(\PT{y})}{\left(R^2-1\right)^{s/2} \left|\PT{x}-\PT{y} \right|^s} = \frac{\left|\PT{x}-\PT{a}\right|^s} {\left(R^2-1\right)^{s/2}}
U_s^\lambda(\PT{x}). \label{KelPot}
\end{equation}
Note that the Kelvin transformation has the duality property $\kelvinMEAS_{R,s} (\lambda^* (\PT{x}^*))=\lambda(\PT{x})$.

\subsection{Signed equilibrium.}

We first establish the uniqueness of the signed equilibrium,
provided it exists.

\begin{lem} \label{lem:uniqueness}
Let $0\leq s<d$. If a signed $s$-equilibrium $\eta_{E,Q}$ exists, then it is unique.
\end{lem}

\begin{proof}
The lemma easily follows from the positivity of the $s$-energy of
signed measures. Indeed, suppose $\eta_1$ and $\eta_2$ are two
signed $s$-equilibria on $E$ associated with the same external field
$Q$. Then
\begin{equation*}
U_s^{\eta_1}(\PT{x})+Q(\PT{x})= F_1 , \quad U_s^{\eta_2}(\PT{x})+Q(\PT{x})= F_2 \qquad \text{for all $\PT{x} \in E$.}
\end{equation*}
Subtracting the two equations and integrating with respect to $\eta_1-\eta_2$ we obtain
\begin{equation*}
\mathcal{I}_s (\eta_1-\eta_2) = \int \left[ U_s^{\eta_1}(\PT{x})- U_s^{\eta_2}(\PT{x}) \right] \dd(\eta_1-\eta_2) (\PT{x}) = 0,
\end{equation*}
and from  \cite[Theorem~1.15]{L} we conclude that $\eta_1=\eta_2$ (see also \cite[Section 5]{G}).
When $d=2$ and $s=0$ instead of \cite[Theorem~1.15]{L} we could use
\cite[Theorem~4.1]{Sim} to prove the assertion of the Lemma.
When $d>2$ and $s=0$ we could use \cite[p.~6]{Riesz}. 
Note that $\eta_1-\eta_2$ is the difference of two signed measures with total charge $1$.
\end{proof}
%

We are now in a position to find the signed equilibrium for the
external field $Q_{\PT{a},q}$ defined by a point charge $q$ at
$\PT{a}$ (see \eqref{RieszField}). 

\begin{proof}[Proof of Theorem \ref{SignEq}] Let
\begin{equation*}
\epsilon_{\PT{a}} \DEF \frac{\left( R^2 - 1 \right)^{d-s}}{W_s(\mathbb{S}^d)
\left| \PT{x} - \PT{a} \right|^{2d-s}} \, \dd \sigma(\PT{x}), \qquad \sigma = \sigma_d.
\end{equation*}
We apply the Kelvin transformation \eqref{KelTr} to the $s$-potential
\begin{equation*}
U^{\epsilon_{\PT{a}}}_s(\PT{z}) = \int_{\mathbb{S}^d}\frac{\left( R^2 - 1 \right)^{d-s}}{W_s(\mathbb{S}^d)\left| \PT{z} - \PT{x} \right|^s
\left| \PT{x} - \PT{a} \right|^{2d-s}}\, \dd\sigma(\PT{x}).
\end{equation*}
From \eqref{KelTrDist} and \eqref{KelTrMeas} (recall that $\kelvin_R(\mathbb{S}^d) =\mathbb{S}^d$) we obtain
\begin{equation*}
U^{\epsilon_{\PT{a}}}_s(\PT{z}) = \left| \PT{z} - \PT{a} \right|^{-s}
\int_{\mathbb{S}^d} \frac{1}{W_s(\mathbb{S}^d) \left| \PT{z}^* - \PT{x}^* \right|^s} \,
\dd\sigma(\PT{x}^*) = \frac{1}{\left|\PT{z}-\PT{a}\right|^s},
\end{equation*}
where we used that $U^\sigma_s (\PT{z^*})=W_s(\mathbb{S}^d)$ for all
$\PT{z^*}\in \mathbb{S}^d$. Hence, $\epsilon_{\PT{a}}=\epsilon_1$
(see \eqref{bal}). For $\eta_{\PT{a}}$ defined in
\eqref{signedeqdens}, we therefore derive
\begin{equation*}
U^{\eta_{\PT{a}}}_s(\PT{z}) + Q_{\PT{a},q}(\PT{z}) = W_s(\mathbb{S}^d) + q U^\sigma_s(\PT{a}), \qquad \text{for all $\PT{z}\in \mathbb{S}^d$.}
\end{equation*}
In addition, one similarly finds
\begin{equation*}
\int_{\mathbb{S}^d} \frac{\left( R^2 - 1 \right)^{d-s}}{\left| \PT{x} - \PT{a} \right|^{2d-s}} \, \dd \sigma(\PT{x}) = \int_{\mathbb{S}^d} \frac{1}{\left| \PT{x}^* - \PT{a} \right|^s} \, \dd \sigma(\PT{x}^*) = U^\sigma_s(\PT{a}),
\end{equation*}
and consequently $\eta_{\PT{a}} (\mathbb{S}^d)=1$. Therefore,
$\eta_{\PT{a}}$ is the required signed $s$-equilibrium. 

Finally, to derive \eqref{SignEqPot(b)}, using \eqref{Measure.Decomposition} and
\eqref{omega.ratio}, we evaluate
\begin{eqnarray}
U_s^\sigma (\PT{a})&=&\int_{\mathbb{S}^d}
\frac{1}{|\PT{x}-\PT{a}|^s}\, \dd \sigma_d
(x)=\frac{\omega_{d-1}}{\omega_d} \int_{-1}^1
\frac{(1-u^2)^{d/2-1}}{(R^2-2Ru+1)^{s/2}}\, \dd u \nonumber\\
&=& (R+1)^{-s}\, \Hypergeom{2}{1}{s/2,d/2}{d}{\frac{4R}{(R+1)^2}}.
\label{Pot(a)2}
\end{eqnarray}
In the last step we used the standard substitution $2 v = 1 + u$ and the integral representation of the hypergeometric function \cite[Eq.~15.3.1]{ABR}.
\end{proof}

The proof of Corollary \ref{Cor} is an easy consequence of the
uniqueness of the extremal measure associated with an external
field.

\begin{proof}[Proof of Corollary \ref{Cor}]
We observe that the (strictly decreasing) density in
\eqref{signedeqdens} is at minimum on $\mathbb{S}^d$ when
$\PT{x}=\PT{p}$. So, non-negativity at the North Pole implies that
the signed equilibrium is positive everywhere else on
$\mathbb{S}^d$, in which case it coincides with the extremal
measure on $\mathbb{S}^d$. On the other hand, if
$\supp(\mu_{Q_{\PT{a},q}})=\mathbb{S}^d$, then the variational
inequalities \eqref{geqineq} and \eqref{leqineq} yield
$\mu_{Q_{\PT{a},q}}=\eta_{\PT{a}}$; and the density in
\eqref{signedeqdens} is again non-negative at $\PT{p}$. What remains
to show is that \eqref{eqsignedA} is equivalent to
\begin{equation*}
1 + \frac{q U_s^\sigma (\PT{a})}{W_s(\mathbb{S}^d)} -
\frac{q\left(R^2-1\right)^{d-s}}{W_s(\mathbb{S}^d)
\left|\PT{p}-\PT{a}\right|^{2d-s}} \geq 0,
\end{equation*}
which can be easily seen by using $| \PT{p} - \PT{a} | = R - 1$.
Finally, using the series expansion of \eqref{SignEqPot(b)} and
\begin{equation*}
\frac{\left(R+1\right)^d}{\left(R-1\right)^d} = \left[ 1 -
\frac{4R}{\left(R+1\right)^2} \right]^{-d/2} = \sum_{k=0}^\infty
\frac{\Pochhsymb{d/2}{k}}{k!} \left[ \frac{4R}{\left(R+1\right)^2}
\right]^k,
\end{equation*}
we derive \eqref{eqsigned}.
\end{proof}

\section{The $s$-balayage measures $\nu_t$ and $\epsilon_t$}
\label{sec:3}

In this section we show that for $s$ in the range $d-2<s<d$, the
measures $\nu_t$ and $\epsilon_t$ are absolutely continuous with respect to the normalized area surface measure $\sigma_d$ (restricted to the spherical cap $\Sigma_t)$ and we find their densities.

\subsection{The balayage measures}
We now focus on the two balayage measures in \eqref{bal}.
The second one, $\nu_t$, has already been found in \cite[Section~3,
Equations~(3.19) and (4.6)]{DS}. It is an absolutely continuous measure on
$\Sigma_t$ (see \eqref{Sigma}), given by the following formula:
\begin{equation}
\dd\nu_{t}(\PT{x}) = \left( 1 + J_{t}(\PT{x}) \right)
\frac{\omega_{d-1}}{\omega_{d}} \left( 1 - u^{2} \right)^{d/2-1}
 \dd u \dd\sigma_{d-1}(\overline{\PT{x}}), \label{nu.r.1}
\end{equation}
where
\begin{equation*}
\begin{split}
J_{t}(\PT{x}) &\DEF \frac{1}{\gammafcn((d-s)/2)\gammafcn(1-(d-s)/2)}
\left( \frac{1-t}{1-u}\right)^{d/2} \left( \frac{t-u}{1-t}
\right)^{(s-d)/2} \\
&\phantom{=\times}\times \int_{0}^{1} v^{d/2-1} \left( 1 - v
\right)^{1+(d-s)/2-1} \left( 1 - \frac{1-t}{1-u} v \right)^{-1} \dd
v.
\end{split}
\end{equation*}

It is convenient to obtain a closed form for $J_{t}(\PT{x})$ in terms of
hypergeometric functions. By \cite[Eq.~15.3.1]{ABR}
\begin{equation*}
\begin{split}
J_{t}(\PT{x}) &\DEF \frac{\gammafcn(d/2)\gammafcn(1+(d-s)/2)}{\gammafcn((d-s)/2)\gammafcn(1-(d-s)/2)\gammafcn(1+d-s/2)}  \\
&\phantom{=\times}\times \left( \frac{1-t}{1-u}\right)^{d/2} \left(
\frac{t-u}{1-t} \right)^{(s-d)/2}
\Hypergeom{2}{1}{1,d/2}{1+d-s/2}{\frac{1-t}{1-u}}.
\end{split}
\end{equation*}
The application of \cite[Eq.~15.3.6]{ABR} yields an expansion near $u=t$,
\begin{equation*}
\begin{split}
J_{t}(\PT{x}) &= - 1 +
\frac{\gammafcn(d/2)}{\gammafcn(d-s/2)} \left( \frac{1-t}{1-u} \right)^{d/2} \left( \frac{t-u}{1-t} \right)^{(s-d)/2} \HypergeomReg{2}{1}{1,d/2}{1-(d-s)/2}{\frac{t-u}{1-u}}.
\end{split}
\end{equation*}
Substituting the last relation into \eqref{nu.r.1} and simplifying
we derive the following lemma.

\begin{lem} \label{NuLem}
Let $d-2<s<d$. The measure $\nu_t=\bal_s(\sigma,\Sigma_t)$ is given by
\begin{equation}
\dd\nu_{t}(\PT{x}) = \nu_{t}^{\prime}(u)
\frac{\omega_{d-1}}{\omega_{d}} \left( 1 - u^{2} \right)^{d/2-1} \dd
u \dd\sigma_{d-1}(\overline{\PT{x}}), \qquad \PT{x} \in \Sigma_t,  \label{NuR}
\end{equation}
where the density $\nu_{t}^{\prime}(u)$ is given by
\begin{equation}
\begin{split} \label{nu.dens}
\nu_{t}^{\prime}(u) &\DEF
\frac{\gammafcn(d/2)}{\gammafcn(d-s/2)} \left( \frac{1-t}{1-u} \right)^{d/2} \left( \frac{t-u}{1-t} \right)^{(s-d)/2} \HypergeomReg{2}{1}{1,d/2}{1-(d-s)/2}{\frac{t-u}{1-u}}.
\end{split}
\end{equation}
\end{lem}


To determine the $s$-balayage $\epsilon_t$, we recall the formulas for the Kelvin
transformation of measures and the relation of the corresponding potentials (see
\eqref{KelMeas} and \eqref{KelPot}). Let $\lambda^*$ be the
extremal measure on the set $\Sigma_t^*:=\kelvin_R (\Sigma_t)$,
normalized so that its potential $U_s^{\lambda^{*}} (\PT{x}^*)=1$
for $\PT{x}^* \in \Sigma_t^*$. Then, using \eqref{KelTr} and
\eqref{KelPot} we derive just as in \cite[Section 3, Equation (3.7)]{DS} that
\begin{equation}
\epsilon_t (\PT{x}) = \left( R^2 - 1 \right)^{-s/2} \kelvinMEAS_{R,s}(\lambda^*(\PT{x}^*)). \label{e_r}
\end{equation}
Since the image $\Sigma_t^*$ of $\Sigma_t$ is also a spherical cap,
this time centered at the North Pole, we can utilize a formula similar to
\eqref{NuR} for its extremal measure. If
$\Sigma_t=\{\PT{x} : -1\leq u \leq t\}$, then $\Sigma_t^*=\{\PT{x} :
1\geq u^* \geq t^* \}$, where $u^*$ and $t^*$ are related to $u$ and
$t$ by \eqref{u_rel}. If we set $\nu_t^* \DEF
{\bal}(\sigma,\Sigma_t^* )$, then $\lambda^*=\nu_t^*
/W_{s}(\mathbb{S}^{d})$; hence we get
\begin{equation}
\dd\lambda^*(\PT{x}^*) = (\lambda^*)^{\prime}(u^*)
\frac{\omega_{d-1}}{\omega_{d}} \left[ 1 - (u^*)^{2} \right]^{d/2-1}
\dd u^* \dd\sigma_{d-1} (\overline{\PT{x}}^*), \label{lambda}
\end{equation}
where the density is given by
\begin{equation*}
\begin{split}
(\lambda^*)^{\prime}(u^*)
&\DEF \frac{1}{W_{s}(\mathbb{S}^{d})} \,
\frac{\gammafcn(d/2)}{\gammafcn(d-s/2)} \left( \frac{1+t^*}{1+u^*} \right)^{d/2}
\left( \frac{u^*-t^*}{1+t^*} \right)^{(s-d)/2} \\
&\phantom{=\times}\times
\HypergeomReg{2}{1}{1,d/2}{1-(d-s)/2}{\frac{u^*-t^*}{1+u^*}}.
\end{split}
\end{equation*}
(We remark that the last formula (up to a multiplicative constant) for the special case $d=2$ was first derived by Fabrikant et al \cite{FABR}.)
From \eqref{u_rel} we get
\begin{equation}
\frac{1+u^*}{1+t^*} = \frac{R^2 - 2 R t + 1}{R^2 - 2 R u + 1} \cdot
\frac{1-u}{1-t}, \label{change}
\end{equation}
from which it follows that
\begin{equation}
\left[ 1 - (u^*)^2 \right]^{d/2-1} \dd u^* = \left( \frac{R^2-1}{R^2
- 2 R u + 1} \right)^d \left( 1 - u^2 \right)^{d/2-1} \dd u.
\label{measconv}
\end{equation}
Substituting \eqref{change} and \eqref{measconv} in \eqref{lambda}
and using \eqref{e_r} and \eqref{KelMeas} we obtain the next lemma.

\begin{lem} Let $d-2<s<d$. The measure $\epsilon_t=\bal_s(\delta_{\PT{a}},\Sigma_t)$ is
given by
\begin{equation} \label{dd.eps.t}
\dd\epsilon_{t}(\PT{x}) = \epsilon_{t}^{\prime}(u)
\frac{\omega_{d-1}}{\omega_{d}} \left( 1 - u^2 \right)^{d/2-1} \dd u
\dd\sigma_{d-1}(\overline{\PT{x}}), \qquad \PT{x} \in \Sigma_t,
\end{equation}
and setting $r^2 \DEF R^2-2Rt+1$, the density is given by
\begin{equation}
\begin{split} \label{eps.dens}
&\epsilon_{t}^{\prime}(u) \DEF \frac{1}{W_{s}(\mathbb{S}^{d})} \,
\frac{\gammafcn(d/2)}{\gammafcn(d-s/2)} \frac{\left( R + 1
\right)^{d-s}}{r^{d}}
\left(  \frac{1-t}{1-u} \right)^{d/2} \\
&\phantom{=\times}\times \left( \frac{t-u}{1-t} \right)^{(s-d)/2}
\HypergeomReg{2}{1}{1,d/2}{1-(d-s)/2}{\frac{\left(R-1\right)^{2}}{r^{2}}
\, \frac{t-u}{1-u}}.
\end{split}
\end{equation}
\label{EpsilonLem}
\end{lem}


\subsection{Positivity of the signed equilibrium of a spherical cap}
The following lemma establishes a condition for positivity of the
signed equilibrium 
\begin{equation*}
\dd\eta_{t}(\PT{x}) = \eta_{t}^{\prime}(u)
\frac{\omega_{d-1}}{\omega_{d}} \left( 1 - u^2 \right)^{d/2-1} \dd u
\dd\sigma_{d-1}(\overline{\PT{x}}).
\end{equation*}

\begin{lem} \label{lem4.1}
Let $d-2<s<d$. If for some $\gamma>0$ we have
$\eta_{t}^{\prime}(u)\geq 0$ for $u \in (t-\gamma,t)$, then
\begin{equation} \label{Cond}
\Phi_s(t) \geq q \left( R + 1 \right)^{d-s} \big/ r^d, \qquad r^2 = R^2 - 2 R t + 1,
\end{equation}
and, consequently, $\eta_{t}^{\prime}(u)>0$ for all $-1\leq u< t<1$.
\end{lem}

\begin{proof}
By equation \eqref{eta} in Theorem \ref{SignEqThm} (which easily
follows from the balayage properties), the definition of $\Phi_s(t)$
(cf. \eqref{FuncPhi}), and the Lemmas \ref{NuLem} and
\ref{EpsilonLem} we get
\begin{equation}
\begin{split}
\eta_{t}^{\prime}(u) &= \frac{1}{W_s(\mathbb{S}^d)}
\frac{\gammafcn(d/2)}{\gammafcn(d-s/2)}
\left( \frac{1-t}{1-u} \right)^{d/2} \left( \frac{t-u}{1-t} \right)^{(s-d)/2} \\
&\phantom{=\times}\times \Bigg\{ \Phi_s(t)
\HypergeomReg{2}{1}{1,d/2}{1-(d-s)/2}{\frac{t-u}{1-u}}  \\
&\phantom{=\times\pm}-  \frac{q\left( R + 1 \right)^{d-s}}{r^{d}}
\HypergeomReg{2}{1}{1,d/2}{1-(d-s)/2}{\frac{\left(R-1\right)^{2}}{r^{2}}
\, \frac{t-u}{1-u}}  \Bigg\}. \label{et.r.a}
\end{split}
\end{equation}
Using \eqref{et.r.a} and the non-negativity hypothesis for $\eta_t^\prime(u)$, we get
\begin{equation*}
\begin{split}
\lim_{u\to t^{-}} \left[ \left( t - u \right)^{(d-s)/2}
\eta_{t}^{\prime}(u) \right]
&= \frac{1}{W_s(\mathbb{S}^d)} \frac{\gammafcn(d/2)}{\gammafcn(d-s/2)\gammafcn(1-(d-s)/2)} \\
&\phantom{=\pm\times}\times \left( 1 - u \right)^{(d-s)/2} \Bigg\{
\Phi_s(t)  - \frac{ q \left( R + 1 \right)^{d-s}}{r^{d}} \Bigg\}
\geq 0.
\end{split}
\end{equation*}
In particular, the expression in braces is non-negative for $d-2<s<d$.

For $R\neq1$ we have $(R-1)^{2}<r^{2}$. Thus, the first
hypergeometric function in \eqref{et.r.a} is strictly larger then
the second one for all $-1\leq u<t$ and $d-2<s<d$. Hence, using
$\Phi_s(t) \geq q (R+1)^{d-s} / r^d$, we have
\begin{equation*}
\begin{split}
\eta_{t}^{\prime}(u) &> \frac{1}{W_s(\mathbb{S}^d)}
\frac{\gammafcn(d/2)}{\gammafcn(d-s/2)}
\left( \frac{1-t}{1-u} \right)^{d/2} \left( \frac{t-u}{1-t} \right)^{(s-d)/2} \\
&\phantom{=\times}\times
\HypergeomReg{2}{1}{1,d/2}{1-(d-s)/2}{\frac{t-u}{1-u}} \left\{
\Phi_s(t) - \frac{q\left( R + 1 \right)^{d-s}}{r^{d}} \right\} \geq 0,
\end{split}
\end{equation*}
which shows that $\eta_{t}^{\prime}(u)>0$ for all $-1\leq u<t$.
\end{proof}

\begin{rmk} \label{rmk:diff.Phi.rhs}
We note that in the limit $R\to1$ relation \eqref{Cond}
becomes the same as in \cite[Eq.~(5.9)]{DS}. It also follows
from the proof of Lemma \ref{lem4.1} that the sign of the difference
$\Phi_s(t) - q (R+1)^{d-s} / r^d$ is determined by the sign of
$\eta_t^\prime(u)$ near the boundary of the spherical cap
$\Sigma_t$, that is for $u$ near $t^-$, and vice versa.
\end{rmk}

\begin{rmk} \label{CondEqRmk}
Equality in relation \eqref{Cond} yields $\lim_{u\to
t^{-}}\eta_{t}^{\prime}(u)=0$. This follows from \eqref{et.r.a} and
the identity
\begin{equation*}
\begin{split}
&\HypergeomReg{2}{1}{1,d/2}{1-(d-s)/2}{\frac{t-u}{1-u}} -
\HypergeomReg{2}{1}{1,d/2}{1-(d-s)/2}{\frac{\left(R-1\right)^{2}}{r^{2}} \, \frac{t-u}{1-u}} \\
&\phantom{=}= \sum_{n=1}^{\infty} \frac{\Pochhsymb{d/2}{n}}{\gammafcn(n+1-(d-s)/2)} \left\{ 1 - \left[
\left( R - 1 \right) / r \right]^{2n} \right\} \left(
\frac{t-u}{1-u} \right)^{n}.
\end{split}
\end{equation*}
\end{rmk}

\section{The norms $\| \epsilon_t \|$ and $\| \nu_t \|$.}
\label{sec:4}

In this section we compute the norms of the measures in \eqref{eta}.

\begin{lem} \label{norm.eps.r}
Let $d-2<s<d$. Then
\begin{equation}
\begin{split} \label{NormEpsB}
\left\| \epsilon_{t} \right\|
&= \frac{2^{1-d} \gammafcn(d)}{\gammafcn(d-s/2)\gammafcn(s/2)} \frac{\left( R + 1 \right)^{d-s}}{W_s(\mathbb{S}^d)} \int_{-1}^{t} \frac{ \left( 1 + u \right)^{s/2-1} \left(
1 - u \right)^{d-s/2-1} }{ \left( R^2 - 2 R u + 1 \right)^{d/2} } \dd u.
\end{split}
\end{equation}
\end{lem}

\begin{proof}
From \eqref{dd.eps.t} and \eqref{eps.dens}
\begin{equation*}
\begin{split}
\left\| \epsilon_{t} \right\| &= \frac{\omega_{d-1}}{\omega_{d}}
\int_{-1}^{t} \epsilon_{t}^{\prime}(u) \left( 1 - u^2
\right)^{d/2-1} \dd u = \frac{\gammafcn(d/2)}{\gammafcn(d-s/2)}
\frac{\omega_{d-1}}{\omega_{d}}
\frac{\left( R + 1 \right)^{d-s}}{W_{s}(\mathbb{S}^{d}) r^{d}} \\
&\phantom{=\pm}\times \left( 1 - t \right)^{d-s/2} \int_{-1}^{t}
\left( 1 + u \right)^{d/2-1} \left( t - u \right)^{1-(d-s)/2-1}
\left( 1 - u \right)^{-1} \\
&\phantom{=\pm\times\times}\times
\HypergeomReg{2}{1}{1,d/2}{1-(d-s)/2}{\frac{\left(R-1\right)^{2}}{r^{2}}
\, \frac{t-u}{1-u}} \dd u.
\end{split}
\end{equation*}
We now apply Lemma \ref{IntLemB}.
\begin{equation*}
\begin{split}
\left\| \epsilon_{t} \right\| &= 2^{(d-s)/2-1} \frac{\gammafcn(d/2)}{\gammafcn(d-s/2)} \frac{\gammafcn(d/2)}{\gammafcn(s/2)} 
\frac{\omega_{d-1}}{\omega_{d}} \frac{\left( R + 1 \right)^{d-s}}{W_{s}(\mathbb{S}^{d}) r^{d}}
\left( 1 - t \right)^{d/2} \left( 1 + t \right)^{s/2} \\
&\phantom{=\pm}\times \left( 1 - x y \right)^{-d/2} \int_0^1
v^{s/2-1} \left( 1 - x v \right)^{d-s/2-1} \left( 1 - \frac{x \left(
1 - y \right)}{1 - x y} v \right)^{-d/2} \dd v,
\end{split}
\end{equation*}
where $x=(1+t)/2$ and $y=(R-1)^{2}/r^{2}$. Substituting
\begin{equation*}
1 - x y = \frac{\left( R + 1 \right)^{2}}{r^{2}} \, \frac{1-t}{2},
\qquad \frac{x \left( 1 - y \right)}{1 - x y} = \frac{4R}{\left( R + 1 \right)^{2}} \, \frac{1+t}{2}.
\end{equation*}
and \eqref{omega.ratio} we get the Euler-type integral of an Appell function
\cite[Eq.~5.8(5)]{ERD}
\begin{equation*}
\begin{split}
&\left\| \epsilon_{t} \right\| = \frac{2^{-s/2} \gammafcn(d)}{\gammafcn(d-s/2)\gammafcn(s/2)} \frac{1}{W_s(\mathbb{S}^d)} 
\left( R + 1 \right)^{-s} \left( 1 + t \right)^{s/2} \\
&\phantom{=\pm}\times \int_{0}^{1} u^{s/2-1} \left( 1 - \frac{1+t}{2} u \right)^{d-s/2-1} \left( 1 - \frac{4R}{\left(R+1\right)^{2}} \frac{1+t}{2} u \right)^{-d/2} \dd u.
\end{split}
\end{equation*}
A change of variables $1+v=(1+t) u$ yields \eqref{NormEpsB}.
\end{proof}


\begin{lem} \label{LemNuNorm}
Let $d-2<s<d$. Then
\begin{align}
\left\|\nu_{t}\right\|
&= \frac{2^{1-d} \gammafcn(d)}{\gammafcn(d-s/2)\gammafcn(s/2)} \int_{-1}^t \left( 1 + u \right)^{s/2-1} \left( 1 - u \right)^{d-s/2-1} \dd u  \label{NuNormB} \\
&= 1 - \mathrm{I}\left((1-t)/2;d-s/2,s/2\right), 
\label{NuNormA}
\end{align}
where $\mathrm{I}(x;a,b)$ denotes the regularized incomplete Beta function (cf. \eqref{regbetafnc}).
\end{lem}

\begin{proof}
We proceed as in the proof of Lemma \ref{norm.eps.r}. In fact, the
densities $\epsilon_{t}^{\prime}$ and $\nu_{t}^{\prime}$ differ by a
multiplicative factor $\left( R + 1
\right)^{d-s}/[W_{s}(\mathbb{S}^{d}) r^{d}]$ and a factor
$(R-1)^{2}/r^{2}$ in the argument of the hypergeometric function.
From \eqref{NuR}, \eqref{nu.dens}, and Lemma \ref{IntLemB}
\begin{equation*}
\left\| \nu_{t} \right\| =
\frac{\gammafcn(d)}{\gammafcn(d-s/2)\gammafcn(s/2)} \left( \frac{1 +
t}{2} \right)^{s/2} \int_0^1 v^{s/2-1} \left( 1 - \frac{1+t}{2} v
\right)^{d-s/2-1} \dd v.
\end{equation*}
A change of variable $1+u=(1+t)v$ yields \eqref{NuNormB}.

A manipulation of the integral (extending the integral over the
complete interval $[-1,1]$ and using the standard substitution
$2v=1-u$) yields \eqref{NuNormA}.
\end{proof}

\section{The extremal support and measure: Proofs of Theorems
\ref{FuncOptThm}, \ref{ConnThm}, and \ref{MainThm}.}
\label{sec:5}

Our first proof deals with the minimization property of $S_Q$.
\vskip 1mm

\begin{proof}[Proof of Theorem \ref{FuncOptThm}] Let $K$ be any compact
subset of $\mathbb{S}^d$ with positive $s$-capacity. For the
considered range of the parameter $s$, we have that the potential of
the extremal measure $\mu_K=\mu_{K,s}$ satisfies the following
(in)equalities
\begin{equation}
U^{\mu_K}_s (\PT{x}) = W_s (K) \quad \text{q.e. on $K$,} \qquad U^{\mu_K}_s (\PT{x})
\leq W_s (K) \quad \text{on $\mathbb{S}^d$.} \label{EqCond}
\end{equation}
This follows trivially from the general theory (see \cite[Chapter
II]{L}) for \linebreak ${d-1\leq s<d}$, with the inequality holding
on the entire space $\mathbb{R}^{d+1}$. To derive \eqref{EqCond} for
the extended range, we observe that for $K=\mathbb{S}^d$ this is
obvious ($\mu_K=\sigma_d$). 

If $\mathbb{S}^d\setminus K$ is non-empty, there is a spherical cap $\Sigma$ that contains $K$. The $s$-potential of $\mu_\Sigma$ equals $W_s(\Sigma)$ everywhere on $\Sigma$, so the measure $\nu \DEF [ W_s(K) / W_s(\Sigma) ] \mu_K$ has a potential that equals $W_s(K)$ on $\Sigma$. Since $U_s^{\mu_K}(\PT{x})\leq W_s(K)$ on $\supp(\mu_K)$ (see \cite[p. 136(b)]{L}), we could derive the inequality in \eqref{EqCond} by comparing the potentials of $\mu_K$ and $\nu$ and applying the
restricted version of the Principle of Domination as given in
\cite[Lemma 5.1]{DS} (for $s=d-2$ we adapt the argument in Lemma 5.1 using \cite[Theorem~1.27]{L}). Since $U^{\mu_K}_s(\PT{x})\geq W_s(K)$ q.e. on $K$ (see \cite[p. 136(a)]{L}), we conclude the equality in \eqref{EqCond} as well.

Clearly, $\mathcal{F}_s (S_Q)=F_Q$ (see \eqref{geqineq} and
\eqref{leqineq}). We now show that for any compact set $K\subset
\mathbb{S}^d$ with positive $s$-capacity we have $\mathcal{F}_s
(K)\geq \mathcal{F}_s (S_Q)$. Indeed, let us integrate
\eqref{geqineq} with respect to $\mu_K$. Since $\mu_K$ has finite
energy, the inequality holds also $\mu_K$-a.e. and we conclude that
\begin{equation*}
\int U_s^{\mu_Q}(\PT{x}) \, \dd\mu_K (\PT{x}) + \int Q(\PT{x}) \, \dd\mu_K (\PT{x}) \geq  F_Q.
\end{equation*}
Using the inequality in \eqref{EqCond} we write
\begin{equation*}
\int U_s^{\mu_Q}(\PT{x}) \, \dd\mu_K (\PT{x}) = \int U_s^{\mu_K}(\PT{x}) \, \dd\mu_Q(\PT{x})\leq W_s(K),
\end{equation*}
which proves our claim.
\end{proof}

Next, we prove sufficient conditions on $Q$, that guarantee that the
extremal support is a spherical zone (cap).

\begin{proof}[Proof of Theorem \ref{ConnThm}] The convexity assumption on
$f(\xi)$ implies that $Q(\PT{z})$ is continuous and the existence
and uniqueness of the extremal measure $\mu_Q$ follows from
standard potential-theoretical arguments (see \cite{Z1}, \cite{Z2}).
The rotational invariance of the external field implies that the
extremal support is also rotationally invariant. Hence, there is
a compact set $A\subset [-1,1]$ and an integrable function $g:A\to
\mathbb{R}^+$, such that the extremal support is given by
\begin{equation*}
\supp(\mu_Q ) = \left\{ (\sqrt{1-u^2} \; \overline{\PT{x}},u)  : u
\in A, \overline{\PT{x}}\in \mathbb{S}^{d-1} \right\},
\end{equation*}
and the extremal measure is
\begin{equation*}
\dd\mu_Q(\PT{x}) = g(u) \dd u \dd\sigma_{d-1}(\overline{\PT{x}}), \qquad u \in A.
\end{equation*}

What we have to show is that $A$ is connected. For this purpose we
adapt the argument given in \cite{MS}. Suppose $A$ is not connected.
Then there is an interval $[\alpha,\beta] \subset (-1,1)$, such that
$[\alpha,\beta]\cap A=\{\alpha,\beta\}$. Let $A^-\DEF A\cap
[-1,\alpha]$ and $A^+\DEF A\cap [\beta,1]$. Letting
\begin{align*}
\PT{x} &= ( \sqrt{1-u^2} \, \overline{\PT{x}}, u ), \qquad u\in A^-\cup A^+, \, \overline{\PT{x}} \in \mathbb{S}^{d-1}, \\
\PT{z} &= ( \sqrt{1-\xi^2} \, \overline{\PT{z}}, \xi ), \qquad
\xi\in(\alpha,\beta), \, \overline{\PT{z}} \in \mathbb{S}^{d-1},
\end{align*}
we represent the weighted $s$-potential as follows:
\begin{align}
U_s^{\mu_Q} (\PT{z})+Q (\PT{z}) &= \int_A g(u) \left( \int_{\mathbb{S}^{d-1}} \frac{ \dd\sigma_{d-1}(\overline{\PT{x}})}{\left|\PT{z}-\PT{x}\right|^s} \right) \, \dd u  + Q(\PT{z}) \notag \\
&\FED \int_{A^-} g(u) \kappa (u,\xi) \, \dd u + \int_{A^+} g(u) \kappa (u,\xi) \, \dd u + f(\xi), \label{EqPot}
\end{align}
where the kernel $\kappa (u,\xi)$ has been evaluated in
\cite[Section 4]{DS} for the case $\xi>u$ ($u\in A^-$) to be
\begin{align}
\kappa (u,\xi)
&\DEF \int_{\mathbb{S}^{d-1}} \frac{\dd \sigma_{d-1}
(\overline{\PT{x}})}{\left| \PT{z} - \PT{x} \right|^s} \label{KappaKernel1''} \\
&=  \left( 1 - u \right)^{-s/2} \left( 1 + \xi \right)^{-s/2}
\Hypergeom{2}{1}{s/2,1-(d-s)/2}{d/2}{\frac{1+u}{1-u}
\frac{1-\xi}{1+\xi}} \label{KappaKernel1'} \\ &= \sum_{k=0}^\infty
\frac{\Pochhsymb{s/2}{k} \Pochhsymb{1-(d-s)/2}{k} \left( 1 + u
\right)^k}{\Pochhsymb{d/2}{k} k! \left( 1 - u \right)^{k+s/2}}
\frac{\left(1-\xi\right)^k}{\left(1+\xi\right)^{k+s/2}}.
\label{KappaKernel1}
\end{align}
By symmetry we derive that when $\xi<u$ ($u\in A^+$)
\begin{equation}
\kappa (u,\xi) = \sum_{k=0}^\infty \frac{\Pochhsymb{s/2}{k}
\Pochhsymb{1-(d-s)/2}{k} \left(1-u\right)^k}{\Pochhsymb{d/2}{k} k!
\left(1+u\right)^{k+s/2}}
\frac{\left(1+\xi\right)^k}{\left(1-\xi\right)^{k+s/2}}.
\label{KappaKernel2}
\end{equation}
It is easy to verify that the functions
\begin{equation*}
\left(1-\xi\right)^k / \left(1+\xi\right)^{k+s/2}, \quad 
\left(1+\xi\right)^k / \left(1-\xi\right)^{k+s/2}, \qquad k
= 0, 1, 2, \dots,
\end{equation*}
are strictly convex for $\xi\in (-1,1)$. Hence, from
\eqref{KappaKernel1} and \eqref{KappaKernel2} we derive that the
kernel $\kappa(u,\xi)$ is a convex function in $\xi$ on
$(\alpha,\beta)$ for any fixed $u\in A^-\cup A^+$. Therefore, using
the convexity of $f(\xi)$ we deduce that the weighted $s$-potential is
strictly convex on $[\alpha,\beta]$. This clearly contradicts the
inequalities \eqref{geqineq} and \eqref{leqineq}, which proves
\eqref{eqsupp}.

Now suppose that, in addition, $f(\xi)$ is also increasing. If
$t_1>-1$, for $u\in [t_1 , t_2]$ and $\xi\in (-1,t_1)$, the kernel
is calculated using \eqref{KappaKernel2}, in which case we easily
obtain that $\partial\kappa(u,\xi)/\partial \xi >0$. This
yields that the weighted $s$-potential is strictly increasing on
$[-1,t_1]$, which contradicts \eqref{geqineq} and \eqref{leqineq}
similarly.
\end{proof}

\begin{proof}[Proof of Theorem \ref{MainThm}] The external field is given by
\begin{equation*}
Q_{\PT{a},q}(\PT{z}) = q \big/ \left| \PT{a} - \PT{z} \right|^s = q
\left| R^2 - 2 R \xi + 1 \right|^{-s/2} \FED f(\xi),
\end{equation*}
where $\PT{z}=(\sqrt{1-\xi^2}\,\overline{\PT{z}},\xi )$, $\xi \in
[-1,1]$, $\overline{\PT{z}}\in \mathbb{S}^{d-1}$. We easily verify
that $f'(\xi)>0$ and $f''(\xi)>0$ for $\xi \in [-1,1]$. According to
Theorem \ref{ConnThm}, the extremal support associated with
$Q_{\PT{a},q}$ is a spherical cap. So, by Theorem \ref{FuncOptThm}
we have to minimize the $\mathcal{F}_s$-functional among all
spherical caps centered at the South Pole.

Recall that (see \eqref{FuncPhi} and Remark \ref{rmk:F.s.EQ.Phi.s})
\begin{equation*}
\mathcal{F}_s (\Sigma_t ) = \Phi_s(t) = W_s(\mathbb{S}^d) \left( 1 + q \left\| \epsilon_t \right\| \right) \big / \left\| \nu_t \right\|.
\end{equation*}
Applying the Quotient Rule and using \eqref{NormEpsB} and \eqref{NuNormB} and the Fundamental Theorem of Calculus, we get (note that $\| \nu_t \| > 0$ for $t>-1$ and $\| \nu_t \|^\prime > 0$ for $-1 < t < 1$)
\begin{align}
\frac{\dd \Phi_{s}}{\dd t} 
&= \frac{q \left\| \epsilon_t \right\|^\prime \left\| \nu_t \right\| - \left( 1 + q \left\| \epsilon_t \right\| \right) \left\| \nu_t \right\|^\prime}{\left\| \nu_t \right\|^2 / W_{s}(\mathbb{S}^d)} = - \frac{\left\| \nu_t \right\|^\prime}{\left\| \nu_t \right\|} \left[ \Phi_{s}(t) - q \, W_s(\mathbb{S}^d) \frac{\left\| \epsilon_t \right\|^\prime}{\left\| \nu_t \right\|^\prime} \right]  \notag \\ 
&= - \frac{\left\| \nu_t \right\|^\prime}{\left\| \nu_t \right\|} \left[ \Phi_{s}(t) - \frac{q \left( R + 1 \right)^{d-s}}{r^{d}} \right] \FED  - \frac{\left\| \nu_t \right\|^\prime}{\left\| \nu_t \right\|} \Delta(t), \label{nec.cond.A}
\end{align}
where $r = r(t) = \sqrt{R^2 - 2 R t + 1}$. 
Observe, that $\Delta(t) \to \infty$ as $t\to-1$. 
Hence, there is a largest $t_0\in(-1,1]$ such that $\Delta(t) > 0$ on $(-1,t_0)$. 
If $t_0=1$, then $\Phi_{s}(t)$ is strictly decreasing on $(-1,1)$ and attains its minimum at $t=1$. 
%
%
We note that $\Delta(1) \geq 0$ is equivalent to the condition in Corollary \ref{Cor}.
If $t_0<1$, then by continuity $\Delta(t_0) = 0$. Clearly, $\Phi_{s}^\prime(t) < 0$ on $(-1,t_0)$ and $\Phi_{s}^\prime(t_0) = 0$. 
Suppose, $\Phi_{s}^\prime(\tau)=0$ for some $\tau\in(-1,1)$. Then $\Delta(\tau) = 0$. 
Applying the product rule we get 
\begin{align*}
\frac{\dd^2 \Phi_{s}}{\dd t^2} (\tau) 
&= - \frac{\left\| \nu_t \right\|^\prime}{\left\| \nu_t \right\|} \left[ \Phi_{s}^\prime(t) - \frac{d\,q \left( R + 1 \right)^{d-s} R}{r^{d+2}} \right]\Bigg|_{t=\tau} = \frac{\left\| \nu_t \right\|^\prime}{\left\| \nu_t \right\|}  \frac{d\,q \left( R + 1 \right)^{d-s} R}{r^{d+2}} \Bigg|_{t=\tau} > 0.
\end{align*}
Hence, any zero of $\Phi_{s}^\prime$ is a minimum of $\Phi_{s}$. Since $\Phi_{s}$ is twice continuously differentiable on $(-1,1)$ (see Lemmas \ref{norm.eps.r} and \ref{LemNuNorm}), the later observation implies that $\Phi_{s}$ has only one local minimum in $(-1,1)$, namely $t_0$, which has to be also a global minimum. Observe, that $\Phi_{s}^\prime(t)<0$ for $t\in(-1,t_0)$ and $\Phi_{s}^\prime(t)>0$ for $t\in(t_0,1)$. 
%
%
%
From \eqref{nec.cond.A} we conclude that $\Delta(t) > 0$ on $(-1,t_0)$ and $\Delta(t) < 0$ on $(t_0,1)$.
This shows that $\Phi_s(t)$ has precisely one global minimum in $(-1,1]$, which is either the unique solution $t_0 \in (-1,1)$ of the equation $\Delta(t)=0$ if it exists, or $t_0=1$. Moreover, $\Delta(t) \geq 0$ if and only if $t \leq t_0$. By Lemma \ref{lem4.1} and Remark \ref{rmk:diff.Phi.rhs} we have $t_0 = \max \{ t : \eta_t \geq 0 \}$. 
Clearly, $S_{Q_{\PT{a},q}}=\Sigma_{t_0}$, from the minimization property.
Since the signed equilibrium for $\Sigma_{t_0}$ is a positive measure,
by the uniqueness of the extremal measure we derive that
$\mu_{Q_{\PT{a},q}}=\eta_{t_0}$.
\end{proof}

\section{The weighted $s$-potential of $\eta_t$ on $\mathbb{S}^d\setminus\Sigma_t$: Alternative proof of Theorem \ref{MainThm}}
\label{sec:6}

In this section we complete the proof of Theorem \ref{SignEqThm},
namely formula \eqref{eq:weighted.outside} on $\mathbb{S}^d\setminus
\Sigma_t$. The $s$-potential of $\eta_t$ is given by
\begin{equation*}
U_s^{\eta_t}(\PT{z}) = \int \frac{\dd \eta_t(\PT{x})}{\left| \PT{z} - \PT{x} \right|^s} = \frac{\omega_{d-1}}{\omega_d} \int_{-1}^t \kappa(u,\xi) \,
\eta_t^\prime(u)  \left( 1 - u^2 \right)^{d/2-1} \dd u ,
\end{equation*}
where $\PT{z} = ( \sqrt{1-\xi^2}\; \overline{\PT{z}}, \xi)$,
$\xi>t$, and the kernel $\kappa(u,\xi)$ is given in \eqref{KappaKernel1'}.
%
The densities $\epsilon_t^\prime$ and $\nu_t^\prime$ of the balayage measures $\epsilon_t$ and $\nu_t$ in \eqref{eta}
have in common that they can be written as (cf.
Lemmas \ref{NuLem} and \ref{EpsilonLem})
\begin{equation}
\gamma_t^\prime(u) = C \left( \frac{1-t}{1-u} \right)^{d/2} \left(
\frac{t-u}{1-t} \right)^{(s-d)/2}
\sum_{n=0}^\infty \frac{\Pochhsymb{d/2}{n}}{\gammafcn(n+1-(d-s)/2)} \left( c_t^2 \frac{t-u}{1-u} \right)^n 
\end{equation}
with appropriately chosen constants $C$ and $c_t$. Hence, it is
sufficient to study the $s$-potential of $\dd \gamma_t =
\gamma_t^\prime \dd\sigma_d|_{\Sigma_t}$.

Using the series representation \eqref{KappaKernel1} of $\kappa(u,\xi)$ and integrating
term-wise we get
\begin{align*}
U_s^{\gamma_t}(\PT{z}) &= \frac{C \, \omega_{d-1} / \omega_d}{\left(
1 + \xi \right)^{s/2}} \sum_{m=0}^\infty \sum_{n=0}^\infty
\frac{\Pochhsymb{s/2}{m}\Pochhsymb{1-(d-s)/2}{m}\Pochhsymb{d/2}{n}}{m!
\Pochhsymb{d/2}{m} \gammafcn(n+1-(d-s)/2)} \left[
\frac{1-\xi}{1+\xi} \right]^m c_t^{2n} \, \mathcal{H}_{m,n}(t;u),
\end{align*}
where $\mathcal{H}_{m,n}(t;u)$ is the integral
\begin{align*}
\mathcal{H}_{m,n}(t;u)
&= \int_{-1}^t \left( \frac{1-t}{1-u} \right)^{d/2} \left( \frac{t-u}{1-t} \right)^{(s-d)/2} \left( \frac{1 + u}{1 - u} \right)^m \left( \frac{t-u}{1-u} \right)^n \frac{\left( 1 - u^2 \right)^{d/2-1}}{ \left( 1 - u \right)^{s/2}} \dd u \\
&= \left( 1 - t \right)^{d-s/2} \int_{-1}^t \frac{\left( t - u
\right)^{n-(d-s)/2} \left( 1 + u \right)^{m+d/2-1}}{\left( 1 - u
\right)^{m+n+1+s/2}} \dd u.
\end{align*}
By \cite[Eq. 2.2.6(9)]{PBMI} 
\begin{equation*}
\mathcal{H}_{m,n}(t;u) = \frac{\gammafcn(m+d/2)\gammafcn(n+1-(d-s)/2)}{\gammafcn(m+n+1+s/2)} \frac{\left( 1 - t \right)^{d-s/2} \left( 1 + t \right)^{m+n+s/2}}{\left( 1 - t \right)^{m+d/2} \left( 1 + t \right)^{n+1-(d-s)/2}}.
\end{equation*}
Putting everything together, we arrive at
\begin{align*}
&U_s^{\gamma_t}(\PT{z}) = 2^{d-s-1} C  \frac{\omega_{d-1}}{\omega_d}
\frac{\gammafcn(d/2)}{\gammafcn(1+s/2)}
\left( \frac{1 - t}{2} \right)^{(d-s)/2} \left( \frac{1+t}{1+\xi} \right)^{s/2} \\
&\phantom{=\pm}\times \sum_{m=0}^\infty \sum_{n=0}^\infty
\frac{\Pochhsymb{s/2}{m}\Pochhsymb{1}{n}\Pochhsymb{1-(d-s)/2}{m}
\Pochhsymb{d/2}{n}}{ \Pochhsymb{1+s/2}{m+n}m!n!} \left(
\frac{1-\xi}{1+\xi} \frac{1 + t}{1 - t} \right)^m \left( c_t^2
\frac{1+t}{2} \right)^n.
\end{align*}
The double sum in the last expression is, in fact,  the series
expansion of the generalized $\HyperF_3$-hypergeometric function
(cf. \cite[Eq. 7.2.4(3)]{PBMIII}) 
\begin{equation*}
\Hypergeom{{}}{3}{a,a^\prime,b,b^\prime}{c}{w,z}  \DEF
\sum_{m=0}^\infty \sum_{n=0}^\infty
\frac{\Pochhsymb{a}{m}\Pochhsymb{a^\prime}{n}\Pochhsymb{b}{m}\Pochhsymb{b^\prime}{n}}{\Pochhsymb{c}{m+n}
m! n!} w^m z^n, \qquad |w|,|z|<1.
\end{equation*}
Moreover, the $\HyperF_3$-function in question is of the form
\cite[Eq. 7.2.4(76)]{PBMIII}
\begin{equation*}
\Hypergeom{{}}{3}{a,c-a,b,c-b}{c}{w,z} = \left( 1 - z
\right)^{a+b-c} \Hypergeom{2}{1}{a,b}{c}{w + z - w z}.
\end{equation*}
Let $r$ be the distance between the  point charge $q$ and any point on
the boundary circle of the spherical cap $\Sigma_t$ (that is
$r^2=R^2-2R t+1$) and $\rho$ be the distance between the point
charge $q$ and $\PT{z}$ on $\mathbb{S}^{d}\setminus\Sigma_t$ (that
is $\rho^2=|\PT{z}-\PT{a}|^2=R^2-2R\xi+1$). For
$C=\gammafcn(d/2)/\gammafcn(d-s/2)$, $c_t=1$ and using
\eqref{omega.ratio}, we have
\begin{equation}
U_s^{\nu_t}(\PT{z}) = W_s(\mathbb{S}^d) A_{s,d} \left(
\frac{1+t}{1+\xi} \right)^{s/2}
\Hypergeom{2}{1}{s/2,1-(d-s)/2}{1+s/2}{\frac{1+t}{1+\xi}}.
\end{equation}
For $C = ( 1 / W_s(\mathbb{S}^d) ) \gammafcn(d/2) / \gammafcn(d-s/2)
( R + 1 )^{d-s} / r^{d}$ and $c_t^2 = ( R - 1 )^2 / r^2$, we get
\begin{equation}
U_s^{\epsilon_t}(\PT{z}) = A_{s,d} \frac{1}{r^{s}} \left(
\frac{1+t}{1+\xi} \right)^{s/2}
\Hypergeom{2}{1}{s/2,1-(d-s)/2}{1+s/2}{\frac{\rho^2}{r^2}\frac{1+t}{1+\xi}}.
\end{equation}
The normalization constant $A_{s,d}$ is given by
\begin{equation*}
A_{s,d} \DEF
\frac{\gammafcn(d/2)}{\gammafcn((d-s)/2)\gammafcn(1+s/2)} = 1 \Big/
\Hypergeom{2}{1}{s/2,1-(d-s)/2}{1+s/2}{1}.
\end{equation*}
(The above last relation holds by \cite[Eq.~15.1.20]{ABR}.) The
relations
\begin{equation*}
\frac{1+t}{1+\xi} = \frac{\left( R + 1 \right)^2 - r^2}{\left( R + 1
\right)^2 - \rho^2},  \qquad \frac{\xi-t}{1+\xi} =
\frac{r^2-\rho^2}{\left( R + 1 \right)^2 - \rho^2}
\end{equation*}
allow to express all formulas in terms of distances to the point
charge $q$ exerting the external field.  Note that the
hypergeometric functions above represent incomplete beta functions
(see \eqref{betafnc}). When using the regularized incomplete beta
function $\mathrm{I}(x;a,b)$ (see
\eqref{regbetafnc}), the $s$-potentials can be also written as
\begin{equation}
U_s^{\nu_t}(\PT{z}) = W_s(\mathbb{S}^d)
\mathrm{I}(\frac{1+t}{1+\xi};  \frac{s}{2}, \frac{d-s}{2}), \qquad
U_s^{\epsilon_t}(\PT{z}) = \frac{1}{\rho^s}
\mathrm{I}(\frac{\rho^2}{r^2} \frac{1+t}{1+\xi}; \frac{s}{2},
\frac{d-s}{2}),
\end{equation}
which are valid for $\PT{z}\in\mathbb{S}^d \setminus \Sigma_t$.
Hence, we obtain
\begin{align*}
U_s^{\eta_t}(\PT{z}) &= \frac{\Phi_s(t)}{W_s(\mathbb{S}^d)}
U_s^{\nu_t}(\PT{z}) - q U_s^{\epsilon_t}(\PT{z}) \\ &= \Phi_s(t) \,
\mathrm{I}\left(\frac{1+t}{1+\xi}; \frac{s}{2}, \frac{d-s}{2}\right)
- q \frac{1}{\rho^s} \mathrm{I}\left(\frac{\rho^2}{r^2}
\frac{1+t}{1+\xi}; \frac{s}{2}, \frac{d-s}{2}\right),
\end{align*}
By means of the functional equation $\mathrm{I}(x;a,b) = 1 -
\mathrm{I}(1-x;b,a)$, it follows that the weighted $s$-potential of
$\eta_t$ for any $-1<t<1$ at $\PT{z}$ in
$\mathbb{S}^d\setminus\Sigma_t$ is given by
%
\begin{equation*}
\begin{split} 
&U_s^{\eta_t}(\PT{z}) + Q(\PT{z}) = \Phi_s(t) + \Bigg\{
\frac{q}{\rho^s}  \mathrm{I}(\frac{\left(R+1\right)^2}{r^2}
\frac{\xi-t}{1+\xi}; \frac{d-s}{2}, \frac{s}{2}) - \Phi_s(t)
\mathrm{I}(\frac{\xi-t}{1+\xi}; \frac{d-s}{2}, \frac{s}{2}) \Bigg\},
\end{split}
\end{equation*}
which proves \eqref{eq:weighted.outside}. 

Next, we provide an alternative proof of Theorem \ref{MainThm}.
Using the (series) expansion
\begin{equation*}
\mathrm{I}(z; a, b) = \frac{\gammafcn(a+b)}{\gammafcn(b)} z^a \left(
1 - z \right)^b \HypergeomReg{2}{1}{1,a+b}{a+1}{z},
\end{equation*}
we obtain for $\xi>t>-1$ the relation 
\begin{equation*}
\begin{split}
&U_s^{\eta_t}(\PT{z}) + Q(\PT{z}) = \Phi_s(t) +
\frac{\gammafcn(d/2)}{\gammafcn(s/2)} \left( \frac{\xi-t}{1+\xi} \right)^{(d-s)/2} \left( \frac{1+t}{1+\xi} \right)^{s/2} \\
&\phantom{=}\times \sum_{n=0}^\infty
\frac{\Pochhsymb{d/2}{n}}{\gammafcn(n+1+(d-s)/2)} \left(
\frac{\xi-t}{1+\xi} \right)^n \left\{ \frac{q \left( R + 1
\right)^{d-s}}{r^d} \left[ \frac{R^2 + 2 R + 1}{R^2 - 2 R t + 1}
\right]^n - \Phi_s(t) \right\}.
\end{split}
\end{equation*}
%
%
%
%
If $q (R+1)^{d-s}/r^d \geq \Phi_s(t)$, then the above infinite
series is a positive function for $1\geq \xi > t$. An immediate
consequence in such a case is the inequality
\begin{equation}
U_s^{\eta_t}(\PT{z}) + Q(\PT{z}) > \Phi_s(t), \qquad \PT{z} \in \mathbb{S}^d \setminus \Sigma_t.
\label{WeightedPotIneq}
\end{equation}
In particular, the last relation holds when $t=t_0$ is a solution of
$q (R+1)^{d-s}/r^d = \Phi_s(t)$. But then from Lemma \ref{lem4.1} we
have that the signed equilibrium is a positive measure. Since it
satisfies the Gauss variational (in)equalities \eqref{geqineq} and
\eqref{leqineq}, this is the extremal measure associated with
$Q$. Easily, we derive that $t_0=\max \{ t\ :\ \eta_t \geq 0 \}$.


\begin{rmk} \label{rmk:derivative.weighted.pot}
An interesting observation is that for $t=t_0$ we could factor
$(\xi-t)/(1+\xi)$ (to get $[(\xi-t)/(1+\xi)]^{1+(d-s)/2}$) and using
product rule, it follows that
\begin{equation}
\frac{\partial}{\partial \xi} \left\{ U_s^{\eta_t}(\PT{z}) +
Q(\PT{z}) \right\} \Big|_{\xi\to t^+} = 0.
\end{equation}
It can be also shown that for $q (R+1)^{d-s}/r^d \neq \Phi_s(t)$ one has
\begin{equation*}
\begin{split}
&\frac{\partial}{\partial \xi} \left\{ U_s^{\eta_t}(\PT{z}) + Q(\PT{z}) \right\} = \frac{\gammafcn(d/2)}{\gammafcn((d-s)/2)\gammafcn(s/2)} \left\{ \frac{q \left( R + 1 \right)^{d-s}}{r^d} - \Phi_s(t) \right\} \\
&\phantom{=\pm}\times \left( 1 + t \right)^{(s-d)/2} \left( \xi - t
\right)^{(d-s)/2-1} + \mathcal{O}( \left( \xi - t \right)^{(d-s)/2})
\qquad \text{as $\xi\to t^+$.}
\end{split}
\end{equation*}
Thus, the partial derivative with respect to $\xi$ of the weighted
$s$-potential of the signed equilibrium $\eta_t$ is singular at the
boundary of $\Sigma_t$ when approaching it from the ``outside'' if
$t$ is not a solution of the equilibrium condition. The sign of this
partial derivative is determined by the difference in curly braces, see Figure \ref{fig1}.
\end{rmk}

\section{The exceptional case $s=d-2$: Proof of Theorems \ref{ExcepThm} and \ref{ExcepThm.log}}
\label{sec:7}

The proof of Theorem \ref{ExcepThm} will be split into several Lemmas.
We first find the $s$-balayage of a point charge
$\PT{y}=(\sqrt{1-v^2} \, \overline{\PT{y}},v)\in \mathbb{S}^d\setminus
\Sigma_t$ onto $\Sigma_t$. Set
\begin{equation*}
\epsilon_{\PT{y}} = \epsilon_{\PT{y},t,d-2} \DEF \bal_{d-2}(\delta_{\PT{y}},\Sigma_t).
\end{equation*}
To determine $\epsilon_{\PT{y}}$ we proceed as in \cite[Section
3]{DS} (see also \cite[Chapter IV]{L}). We apply an inversion
(stereographical projection) with center $\PT{y}$ and radius
$\sqrt{2}$. The image of $\mathbb{S}^d$ is a hyperplane passing
through the origin. The image of $\Sigma_t$ is a hyperdisc of radius
$\tau=\sqrt{1-t^2}/(v-t)$. The $(d-2)$-extremal measure on this
$d$-dimensional hyperdisc is the normalized (unit) uniform surface
measure on its boundary $\dd \lambda^*(\PT{x}^*) = \tau^{d-1} \dd \sigma_{d-1} ((\PT{x}^*-\PT{b}^*) / \tau)$, where $\PT{b}^*$ is the center of this hyperdisc. The potential of $\lambda^*$ is found to be
\begin{equation*}
U_{d-2}^{\lambda^*}(\PT{x}^*) = \tau \int_{\mathbb{S}^{d-1}} \frac{\dd \sigma_{d-1} ((\PT{x}^*-\PT{b}^*) / \tau)}{\left| \left( \PT{z}^* - \PT{b}^* \right) / \tau - \left( \PT{x}^* - \PT{b}^* \right) / \tau \right|^{d-2}} = \tau \, W_{d-2}(\mathbb{S}^{d-1}) = \tau.
\end{equation*}
 Using the Kelvin transformation of this measure as given in Section 2.1 (cf.
\eqref{KelMeas} and \eqref{KelPot} with $R^2-1=2$), we compute that
\begin{equation}
\dd \epsilon_{\PT{y}} (\PT{x}) = 2 \left( v - t \right) \left( 1 - t^2 \right)^{d/2-1} \frac{\dd \sigma_{d-1}(\overline{\PT{x}})}{\left| \PT{x} - \PT{y} \right|^d}, \qquad \PT{x}\in \partial \Sigma_t. \label{BalY}
\end{equation}

In \cite[Section 3, Eq.~(3.12)]{DS} the corresponding point
charge balayage was calculated for $d-2<s<d$,
\begin{equation}
\dd \epsilon_{\PT{y},s} (\PT{x}) = \frac{2\sin(\pi(d-s)/2)}{\pi} \left( \frac{v-t}{t-u} \right)^{(d-s)/2} \left( 1 - u^2 \right)^{d/2-1} \frac{\dd
u \; \dd \sigma_{d-1}(\overline{\PT{x}})}{\left| \PT{x} - \PT{y} \right|^d}, \quad \PT{x} \in  \Sigma_t. \label{BalYs}
\end{equation}

The following lemma establishes the relationship between
$\epsilon_{\PT{y},s}$ and $\epsilon_{\PT{y}}$.

\begin{lem} \label{WeakConv}
Let $d \geq 3$.
Let $\displaystyle{ \dd \gamma_s \DEF \frac{\sin(\pi(d-s)/2)}{\pi(t-u)^{(d-s)/2}} \dd u}$, $-1\leq u\leq t$. Then
$\|\gamma_s\|\to 1$ and $\gamma_s \stackrel{*}{\to} \delta_t$, as
$s\to (d-2)^+$. Consequently, $\epsilon_{\PT{y},s} \stackrel{*}{\to} \epsilon_{\PT{y}}$, as $s\to (d-2)^+$.
\end{lem}

\begin{proof} We compute 
\begin{equation*}
\left\| \gamma_s \right\| = \int_{-1}^t \frac{\sin(\pi(d-s)/2)}{\pi(t-u)^{(d-s)/2}} \dd u=
\frac{\sin(\pi\left( 1 - (d-s)/2 \right))}{\pi(1-(d-s)/2)}(1+t)^{1-(d-s)/2} .
\end{equation*}
Clearly, $\| \gamma_s \|\leq 2$ and $\|\gamma_s\|\to 1$ as $s\to (d-2)^+$. Let $f$ be a
continuous function on $[-1,t]$. Then what we have to prove is that
\begin{equation*}
\lim_{s\to (d-2)^+} \int_{-1}^t \frac{\sin(\pi(d-s)/2)}{\pi(t-u)^{(d-s)/2}} f(u) \dd u = f(t).
\end{equation*}
By $\|\gamma_s\|\to 1$ as $s\to (d-2)^+$, this is equivalent to
\begin{equation} \label{int.difference}
\lim_{s\to (d-2)^+} \int_{-1}^t \frac{\sin(\pi(d-s)/2)}{\pi(t-u)^{(d-s)/2}} \left[ f(u) - f(t) \right] \dd u = 0.
\end{equation}

Let $\epsilon>0$. From the continuity of $f$ it follows that there exists a $\delta>0$ such that $| f(u) - f(t) | < \epsilon / 4$ whenever $| u - t | < \delta$. For $s$ sufficiently close to $(d-2)^+$ we estimate that
\begin{equation*}
\left| \int_{-1}^{t-\delta} \left[ f(u) - f(t) \right] \frac{\sin(\pi(d-s)/2)}{\pi(t-u)^{(d-s)/2}} \dd u \right| 
\leq 2\left\| f \right\|_{[-1,t]} \frac{\sin \pi (d-s)/2 }{\pi \delta} < \epsilon / 2,
\end{equation*}
and
\begin{equation*}
\left| \int_{t-\delta}^t \left[ f(u) - f(t) \right] \frac{\sin(\pi(d-s)/2)}{\pi(t-u)^{(d-s)/2}} \dd u \right| 
< \frac{\epsilon}{4} \left\| \gamma_s \right\| \leq \epsilon / 2.
\end{equation*}
Therefore, 
\begin{equation*}
\left| \int_{-1}^t \left[ f(u) - f(t) \right] \frac{\sin(\pi(d-s)/2)}{\pi(t-u)^{(d-s)/2}} \dd u \right| < \epsilon,
\end{equation*}
which proves \eqref{int.difference}.

Suppose now that $f(\PT{x})$, where $\PT{x}=(\sqrt{1-u^2}\,\overline{\PT{x}},u)$, is a continuous function on $\mathbb{S}^d$. Then as $s\to(d-2)^+$ we have 
\begin{align*}
&\lim \int_{\Sigma_t} f \dd \epsilon_{\PT{y},s} = \lim \int_{-1}^t \left( \int_{\mathbb{S}^{d-1}} f(\PT{x}) \frac{\dd \sigma_{d-1}(\overline{\PT{x}})}{\left| \PT{x} - \PT{y} \right|^d} \right) 2 \left( v - t \right)^{(d-s)/2} \left( 1 - u^2 \right)^{d/2-1} \dd \gamma_s(u) \\
&\phantom{=}= 2 \left( v - t \right) \left( 1 - t^2 \right)^{d/2-1} \left( \int_{\mathbb{S}^{d-1}} f(\PT{x}) \frac{\dd \sigma_{d-1}(\overline{\PT{x}})}{\left| \PT{x} - \PT{y} \right|^d} \right)\Bigg|_{u=t} = \int_{\Sigma_t} f \dd \epsilon_{\PT{y}},
\end{align*}
which completes the proof of the lemma.
\end{proof}

Next, we determine the balayage measures in \eqref{barBal}. We shall use that $\beta_t$, which is the unit charge uniformly distributed on the boundary of $\Sigma_t$, has $(d-2)$-potential
\begin{equation} \label{boundary.s.pot}
U_{d-2}^{\beta_t}(\PT{z}) = \int_{\mathbb{S}^{d-1}, u=t} \frac{\dd\sigma_{d-1}(\overline{\PT{x}})}{\left| \PT{z} - \PT{x} \right|^{d-2}} = 
\begin{cases}
\left( 1 - t \right)^{1-d/2} \left( 1 + \xi \right)^{1-d/2} & \text{if $\xi\geq t$,} \\
\left( 1 + t \right)^{1-d/2} \left( 1 - \xi \right)^{1-d/2} & \text{if $\xi< t$,}
\end{cases} 
\end{equation}
where $\PT{z}=(\sqrt{1-\xi^2}\,\overline{\PT{z}},u)\in\mathbb{S}^d$. This follows from \eqref{KappaKernel1} and \eqref{KappaKernel2}.

\begin{lem} \label{NuLem:s.EQ.d-2} Let $d \geq 3$. The measure $\overline{\nu}_t=\bal_{d-2}(\sigma_d,\Sigma_t)$ is given by
\begin{equation}
\dd \overline{\nu}_{t}(\PT{x}) = \dd \sigma_{d} \big|_{\Sigma_t}(\PT{x}) + W_{d-2}(\mathbb{S}^{d}) \frac{1-t}{2} \left( 1 - t^2 \right)^{d/2-1} \dd \delta_t(u) \dd\sigma_{d-1}(\overline{\PT{x}}). \label{NuR:s.EQ.d-2}
\end{equation}
The $(d-2)$-potential of $\overline{\nu}_{t}$ is given by
\begin{align}
U_{d-2}^{\overline{\nu}_t}(\PT{z}) &= W_{d-2}(\mathbb{S}^d), \qquad \PT{z} \in \Sigma_t,  \label{U.d-2.ubar.nu.t.in} \\
U_{d-2}^{\overline{\nu}_t}(\PT{z}) &= W_{d-2}(\mathbb{S}^d) \left( 1 + t \right)^{d/2-1} \left( 1 + \xi \right)^{1-d/2} < W_{d-2}(\mathbb{S}^d), \qquad \PT{z} \in \mathbb{S}^d \setminus \Sigma_t. \label{U.d-2.ubar.nu.t.notin}
\end{align}
\end{lem}

\begin{rmk}
It is interesting that the $(d-2)$-potential of $\overline{\nu}_t$ can be expressed using the potential of $\beta_t$ (cf. \eqref{boundary.s.pot})
\begin{equation}
U_{d-2}^{\overline{\nu}_t}(\PT{z}) = W_{d-2}(\mathbb{S}^d) \left( 1 - t^2 \right)^{d/2-1} U_{d-2}^{\beta_t}(\PT{z}), \qquad \PT{z} \in \mathbb{S}^d \setminus \Sigma_t.
\end{equation}
\end{rmk}

%
\begin{rmk}
In the proof of Lemma \ref{NuLem:s.EQ.d-2} and Lemma \ref{EpsLem:s.EQ.d-2} below we shall obtain the balayage measures constructively. Alternatively, one could get this from the potential (in)equalities \eqref{U.d-2.ubar.nu.t.in}, \eqref{U.d-2.ubar.nu.t.notin} and \eqref{U.d-2.ubar.eps.t.tin}, \eqref{U.d-2.ubar.eps.t.notin}.
\end{rmk}

\begin{proof}[Proof of Lemma \ref{NuLem:s.EQ.d-2}]
It is well-known that
\begin{equation} \label{bal.s.EQ.d-2}
\bal_{d-2}(\sigma_d,\Sigma_t) = \sigma_d\big|_{\Sigma_t} + \bal_{d-2}(\sigma_d\big|_{\mathbb{S}^d\setminus\Sigma_t},\Sigma_t).
\end{equation}
By the principle of superposition we have for $\PT{x}\in\partial\Sigma_t$
\begin{align*}
&\bal_{d-2}(\sigma_d\big|_{\mathbb{S}^d\setminus\Sigma_t},\Sigma_t) 
= \int_{\mathbb{S}^d\setminus\Sigma_t} \epsilon_{\PT{y}}(\PT{x}) \dd \sigma_d(\PT{y}) \\
&\phantom{=}= \frac{\omega_{d-1}}{\omega_d} \int_t^1 \left( \int_{\mathbb{S}^{d-1}} \epsilon_{\PT{y}}(\PT{x}) \dd \sigma_{d-1}(\overline{\PT{y}}) \right) \left( 1 - v^2 \right)^{d/2-1} \dd v \\
&\phantom{=}= 2 \frac{\omega_{d-1}}{\omega_d} \left( 1 - t^2 \right)^{d/2-1} \left(\int_t^1 \left( 1 - v^2 \right)^{d/2-1} \left( v - t \right) \int_{\mathbb{S}^{d-1}} \frac{\dd \sigma_{d-1}(\overline{\PT{y}})}{\left|\PT{x}-\PT{y}\right|^d} \dd v \right) \sigma_{d-1}(\overline{\PT{x}}).
\end{align*}
The inner integral can be computed using \eqref{KappaKernel2} with $s=d$
\begin{equation} \label{int.kappa}
\int_{\mathbb{S}^{d-1}} \frac{\dd \sigma_{d-1}(\overline{\PT{y}})}{\left|\PT{x}-\PT{y}\right|^d} = \frac{1}{2\left( v - t \right) \left( 1 + v \right)^{d/2-1}\left( 1 - t \right)^{d/2-1}}.
\end{equation}
Hence,
\begin{align*}
&\bal_{d-2}(\sigma_d\big|_{\mathbb{S}^d\setminus\Sigma_t},\Sigma_t) 
= \frac{\omega_{d-1}}{\omega_d} \left( 1 + t \right)^{d/2-1} \left(\int_t^1 \left( 1 - v \right)^{d/2-1} \dd v \right) \sigma_{d-1}(\overline{\PT{x}}) \\
&\phantom{equals}= \frac{2}{d} \frac{\omega_{d-1}}{\omega_d} \left( 1 + t \right)^{d/2-1} \left( 1 - t \right)^{d/2} \sigma_{d-1}(\overline{\PT{x}}) \FED q_{\overline{\nu}_{t}} \sigma_{d-1}(\overline{\PT{x}}), \qquad \PT{x}\in\partial\Sigma_t.
\end{align*}
Using $W_{d-2}(\mathbb{S}^d) = (4 / d) (\omega_{d-1} / \omega_d )$ and \eqref{bal.s.EQ.d-2} we derive \eqref{NuR:s.EQ.d-2}.

Relation \eqref{U.d-2.ubar.nu.t.in} holds because of the balayage properties. 
Using \eqref{boundary.s.pot} we have 
\begin{align*}
U_{d-2}^{\overline{\nu}_t}(\PT{z}) 
&= \int_{\Sigma_t} \frac{\dd \sigma_d(\PT{x})}{\left| \PT{z} - \PT{x}\right|^{d-2}} + q_{\overline{\nu}_t} U_{d-2}^{\beta_t}(\PT{z}) \\
&= \frac{\omega_{d-1}}{\omega_d} \int_{-1}^t \left( 1 - u^2 \right)^{d/2-1} \int_{\mathbb{S}^{d-1}} \frac{\dd \sigma_{d-1}(\overline{\PT{x}})}{\left| \PT{z} - \PT{x} \right|^{d-2}} \, \dd u +  \frac{q_{\overline{\nu}_t}}{ \left( 1 - t \right)^{d/2-1} \left( 1 + \xi \right)^{d/2-1}} \\
&= \frac{\omega_{d-1}}{\omega_d} \int_{-1}^t \frac{\left( 1 - u^2 \right)^{d/2-1}}{\left( 1 - u \right)^{d/2-1} \left( 1 + \xi \right)^{d/2-1}}  \, \dd u + \frac{q_{\overline{\nu}_t}}{ \left( 1 - t \right)^{d/2-1} \left( 1 + \xi \right)^{d/2-1}} \\
&= W_{d-2}(\mathbb{S}^d) \frac{1+t}{2} \frac{\left( 1 + t \right)^{d/2-1}}{\left( 1 + \xi \right)^{d/2-1}} + W_{d-2}(\mathbb{S}^d) \frac{1-t}{2} \frac{\left( 1 + t \right)^{d/2-1}}{\left( 1 + \xi \right)^{d/2-1}}, 
\end{align*}
from which follows \eqref{U.d-2.ubar.nu.t.notin}
\end{proof}

\begin{lem} \label{EpsLem:s.EQ.d-2} Let $d\geq3$. The measure $\overline{\epsilon}_t=\bal_{d-2}(\delta_{\PT{a}},\Sigma_t)$ is given by
\begin{equation}
\dd \overline{\epsilon}_{t}(\PT{x}) = \overline{\epsilon}_{t}^{\prime}(u) \dd \sigma_{d} \big|_{\Sigma_t}(\PT{x}) + q_{\overline{\epsilon}_{t}} \dd \delta_t(u) \dd\sigma_{d-1}(\overline{\PT{x}}), \label{EpsR:s.EQ.d-2}
\end{equation}
where the density $\overline{\epsilon}_{t}^{\prime}(u)$ and the constant $q_{\overline{\epsilon}_{t}}$ are given by
\begin{equation}
\overline{\epsilon}_{t}^{\prime}(u) \DEF \frac{\left( R^2 - 1 \right)^2 / W_{d-2}(\mathbb{S}^d)}{\left( R^2 - 2 R u + 1 \right)^{d/2+1}}, \quad q_{\overline{\epsilon}_{t}} = \frac{1-t}{2} \frac{\left( R + 1 \right)^2}{r^d} \left( 1 - t^2 \right)^{d/2-1}.
\end{equation}
The $(d-2)$-potential of $\overline{\epsilon}_{t}$ is given by
\begin{align}
U_{d-2}^{\overline{\epsilon}_t}(\PT{z}) &= \left| \PT{z} - \PT{a} \right|^{2-d} = U_{d-2}^{\delta_{\PT{a}}}(\PT{z}), \qquad \PT{z} \in \Sigma_t, \label{U.d-2.ubar.eps.t.tin} \\
U_{d-2}^{\overline{\epsilon}_t}(\PT{z}) &= r^{2-d} \left( 1 + t \right)^{d/2-1} \left( 1 + \xi \right)^{1-d/2} < U_{d-2}^{\delta_{\PT{a}}}(\PT{z}), \qquad \PT{z} \in \mathbb{S}^d \setminus \Sigma_t. \label{U.d-2.ubar.eps.t.notin}
\end{align}
\end{lem}

\begin{proof}
As in the proof of Theorem \ref{SignEq} we evaluate 
\begin{equation}
\overline{\epsilon}_{\PT{a}} \DEF \bal_{d-2}(\delta_{\PT{a}}, \mathbb{S}^d), \qquad \dd \overline{\epsilon}_{\PT{a}}(\PT{x}) = \overline{\epsilon}_{t}^{\prime}(u) \dd \sigma_{d}(\PT{x}).
\end{equation}
Using balayage in steps and \eqref{bal.s.EQ.d-2} we get
\begin{equation}
\bal_{d-2}(\delta_{\PT{a}}, \Sigma_t) = \overline{\epsilon}_{\PT{a}}\big|_{\Sigma_t} + \bal_{d-2}(\overline{\epsilon}_{\PT{a}}\big|_{\mathbb{S}^d\setminus\Sigma_t}, \Sigma_t).
\end{equation}
By the principle of superposition we have for $\PT{x}\in\partial\Sigma_t$
\begin{align*}
&\bal_{d-2}(\overline{\epsilon}_{\PT{a}}\big|_{\mathbb{S}^d\setminus\Sigma_t}, \Sigma_t) 
= \int_{\mathbb{S}^d\setminus\Sigma_t} \overline{\epsilon}_{t}^{\prime}(v) \epsilon_{\PT{y}}(\PT{x}) \dd \sigma_d(\PT{y}) \\
&\phantom{=}= \frac{\omega_{d-1}}{\omega_d} \int_t^1 \left( \int_{\mathbb{S}^{d-1}} \overline{\epsilon}_{t}^{\prime}(v) \epsilon_{\PT{y}}(\PT{x}) \dd \sigma_{d-1}(\overline{\PT{y}}) \right) \left( 1 - v^2 \right)^{d/2-1} \dd v \\
&\phantom{=}= 2 \frac{\omega_{d-1}}{\omega_d} \left( 1 - t^2 \right)^{d/2-1} \left(\int_t^1 \left( 1 - v^2 \right)^{d/2-1} \overline{\epsilon}_{t}^{\prime}(v) \left( v - t \right) \int_{\mathbb{S}^{d-1}} \frac{\dd \sigma_{d-1}(\overline{\PT{y}})}{\left|\PT{x}-\PT{y}\right|^d} \dd v \right) \sigma_{d-1}(\overline{\PT{x}}).
\end{align*}
Applying \eqref{int.kappa} yields
\begin{align}
&\bal_{d-2}(\overline{\epsilon}_{\PT{a}}\big|_{\mathbb{S}^d\setminus\Sigma_t}, \Sigma_t) \notag \\
&\phantom{=}= \frac{\omega_{d-1}}{\omega_d} \frac{\left( R^2 - 1 \right)^2}{W_{d-2}(\mathbb{S}^d)} \left( 1 + t \right)^{d/2-1}  \left( \int_t^1 \frac{\left( 1 - v \right)^{d/2-1}}{\left( R^2 - 2 R v + 1 \right)^{d/2+1}} \dd v \right) \sigma_{d-1}(\overline{\PT{x}}) \notag \\
&\phantom{=}= \frac{2}{d} \frac{\omega_{d-1}}{\omega_d} \frac{\left( R^2 - 1 \right)^2}{W_{d-2}(\mathbb{S}^d)} \left( 1 + t \right)^{d/2-1} \frac{\left( 1 - t \right)^{d/2} }{\left( R^2 - 2 R t + 1 \right)^{d/2}} \sigma_{d-1}(\overline{\PT{x}}) = q_{\overline{\epsilon}_{t}} \sigma_{d-1}(\overline{\PT{x}}), \notag 
\end{align}
where we used the change of variable $w = (R-1)^2 / (1 - v) + 2 R$ to compute the integral in the parenthesis.

Similar computations with the substitution $w = (R+1)^2 / (1 + u) - 2 R$ (see also \eqref{boundary.s.pot}) lead to \eqref{U.d-2.ubar.eps.t.notin}. That is, for $\PT{z} \in \mathbb{S}^d\setminus\Sigma_t$ one has
\begin{align*}
&U_{d-2}^{\overline{\epsilon}_t}(\PT{z}) 
= \int_{\Sigma_t} \frac{\overline{\epsilon}_{t}^{\prime}(u) \dd \sigma_d(\PT{x})}{\left| \PT{z} - \PT{x}\right|^{d-2}} + q_{\overline{\epsilon}_t} U_{d-2}^{\beta_t}(\PT{z}) \\
&\phantom{}= \frac{\omega_{d-1}}{\omega_d} \frac{\left( R^2 - 1 \right)^2}{W_{d-2}(\mathbb{S}^d)} \int_{-1}^t \frac{\left( 1 - u^2 \right)^{d/2-1}}{\left( R^2 - 2 R u + 1 \right)^{d/2+1}}\int_{\mathbb{S}^{d-1}} \frac{\dd \sigma_{d-1}(\overline{\PT{x}})}{\left| \PT{z} - \PT{x} \right|^{d-2}} \, \dd u + q_{\overline{\epsilon}_t} U_{d-2}^{\beta_t}(\PT{z})  \\
&\phantom{}= \frac{\omega_{d-1}}{\omega_d} \frac{\left( R^2 - 1 \right)^2}{W_{d-2}(\mathbb{S}^d)} \int_{-1}^t \frac{\left( 1 + u \right)^{d/2-1}}{\left( R^2 - 2 R u + 1 \right)^{d/2+1} \left( 1 + \xi \right)^{d/2-1}}  \, \dd u +  q_{\overline{\epsilon}_t} U_{d-2}^{\beta_t}(\PT{z}) \\ 
&\phantom{}= \frac{2}{d} \frac{\omega_{d-1}}{\omega_d} \frac{\left( R - 1 \right)^2}{W_{d-2}(\mathbb{S}^d) r^d} \frac{\left( 1 + t \right)^{d/2}}{\left( 1 + \xi \right)^{d/2-1}} + \frac{1-t}{2} \frac{\left( R + 1 \right)^2}{r^d} \frac{\left( 1 + t \right)^{d/2-1}}{\left( 1 + \xi \right)^{d/2-1}} \\
&\phantom{}= \frac{1}{r^{d-2}} \frac{\left( 1 + t \right)^{d/2-1}}{\left( 1 + \xi \right)^{d/2-1}} \left[ \frac{\left( R - 1 \right)^2}{R^2-2R t+1} \frac{1+t}{2}  + \frac{1-t}{2} \frac{\left( R + 1 \right)^2}{R^2-2R t+1} \right] = \frac{1}{r^{d-2}} \frac{\left( 1 + t \right)^{d/2-1}}{\left( 1 + \xi \right)^{d/2-1}}.
\end{align*}

As in the proof of Lemma \ref{NuLem:s.EQ.d-2} the balayage properties imply Equation \eqref{U.d-2.ubar.eps.t.tin}.
\end{proof}

The weak$^*$ convergence in \eqref{weak.star.conv} is shown next.

\begin{lem}
Let $t\in(-1,1)$ be fixed. Then 
\begin{equation}
\nu_{t,s} \stackrel{*}{\longrightarrow} \overline{\nu}_t , \qquad
\quad \epsilon_{t,s} \stackrel{*}{\longrightarrow}
\overline{\epsilon}_t, \qquad \text{as $s\to(d-2)^+$}.
\end{equation}
\end{lem}

\begin{proof}
The result follows easily from the weak$^*$ convergence $\epsilon_{\PT{y},s}\stackrel{*}{\longrightarrow}\epsilon_{\PT{y}}$ as $s\to(d-2)^+$ and the following representation valid for any measure $\mu$ on $\mathbb{S}^d$:
\begin{equation}
\bal_s(\mu, \Sigma_t)(\PT{x}) = \mu\big|_{\Sigma_t}(\PT{x}) + \int_{\mathbb{S}^d\setminus\Sigma_t} \epsilon_{\PT{y},s}(\PT{x}) \dd \mu(\PT{y}).
\end{equation}
\end{proof}

The norms $\|\overline{\nu}_t\|$ and $\|\overline{\epsilon}_t\|$ can be obtained from Lemmas \ref{norm.eps.r} and \ref{LemNuNorm} by taking the limit $s\to(d-2)^+$ (which is justified by the weak$^*$  convergence shown in Lemma \ref{WeakConv}).

\begin{lem} \label{lem:norms.s.EQ.d-2}
Let $d\geq3$. Then
\begin{align}
\left\| \overline{\epsilon}_t \right\| &= \frac{d-2}{4} \left(R+1\right)^2 \int_{-1}^t \frac{\left( 1 + u \right)^{d/2-2}\left( 1 - u \right)^{d/2}}{\left( R^2 - 2 R u + 1 \right)^{d/2}} \dd u, \label{NormEpsB.2} \\
\left\| \overline{\nu}_t \right\| &= \frac{d-2}{4} W_{d-2}(\mathbb{S}^d) \int_{-1}^t \left( 1 + u \right)^{d/2-2} \left( 1 - u \right)^{d/2} \dd u. \label{NuNormB.2}
\end{align}

\end{lem}

\begin{proof}[Completion of the proof of Theorem \ref{ExcepThm}]
Proceeding as in the proof of Theorem \ref{MainThm}, but using now 
($r = r(t) = \sqrt{R^2 - 2 R t + 1}$) 
\begin{equation*}
\overline{\Phi}_{d-2}^\prime(t) = - \left\| \overline{\nu}_t \right\|^\prime \big/ \left\| \overline{\nu}_t \right\|  \left[ \overline{\Phi}_{d-2}(t) - q \left( R + 1 \right)^{2} / r^{d} \right] \FED  - \left\| \overline{\nu}_t \right\|^\prime / \left\| \overline{\nu}_t \right\|  \Delta(t), 
\end{equation*}
it follows that the global minimum of $\overline{\Phi}_{d-2}$ is either the unique solution $t_0\in(-1,1)$ of the equation $\Delta(t)=0$, or $t_0=1$.
In particular, $\Delta(t) \geq 0$ if and only if $t \leq t_0$.

The explicite form \eqref{etabar} follows from Lemmas \ref{NuLem:s.EQ.d-2} and \ref{EpsLem:s.EQ.d-2}. 
If $\overline{\eta}_t\geq0$ then $\Delta(t) \geq 0$, so $t\leq t_0$. On the other hand, it is easy to see that if $t=t_0$, then $\overline{\eta}_{t_0}$ given in \eqref{etabarzero} is $\geq 0$ because of $(R-1)^2 < R^2 - 2 R t_0 + 1 < R^2 - 2 R u + 1$.
Therefore, we have that $t_0 = \max\{ t : \overline{\eta}_t \geq 0 \}$, $\mu_{\overline{Q}_{\PT{a},q}}=\overline{\eta}_{t_0}$, and $\supp(\mu_{\overline{Q}_{\PT{a},q}}) = \Sigma_{t_0}$.
\end{proof}

The proof of Theorem \ref{ExcepThm.log} is also split into several lemmas.

We must check that Theorem \ref{ConnThm} also holds in the case $d=2$ and $s=0$. Then we can make use of the fact that the support $S_{\overline{Q}_{\PT{a},q}}$ of the extremal measure on $\mathbb{S}^2$ associated with the external logarithmic field $\overline{Q}_{\PT{a},q}$ is a spherical cap.

\begin{proof}[Adaptation of the proof of Theorem \ref{ConnThm} for $d=2$ and $s=0$]
Theorem \ref{ConnThm} can be extended to hold for $d=2$ and $s=0$. Instead of the kernel $\kappa(u,\xi)$ given in \eqref{KappaKernel1''} one has to consider
\begin{align}
\kappa_0(u,\xi) &\DEF \int_{\mathbb{S}^1} \log \frac{1}{\left| \PT{z} - \PT{x} \right|} \dd \sigma_1(\overline{\PT{x}}) 
= - \frac{1}{2} \frac{1}{\pi} \int_{-1}^1 \frac{\log\left( 2 - 2 u \xi - 2 \sqrt{1 - u^2} \sqrt{1 - \xi^2} \, \tau \right)}{\sqrt{1-\tau^2}} \dd \tau \notag \\
&\phantom{equals}= - \frac{1}{2} \log \left( 1 - u \xi + \left| \xi - u \right| \right) = \begin{cases} - \frac{1}{2} \log \left( 1 + \xi \right) - \frac{1}{2} \log \left( 1 - u \right) & \xi \geq u, \\  - \frac{1}{2} \log \left( 1 - \xi \right) - \frac{1}{2} \log \left( 1 + u \right) & \xi \leq u. \end{cases} \label{log.int.aux}
\end{align}
This follows from the Funk-Hecke formula and \cite[Lemma~1.15]{ST}.
It is easy to verify that the kernel $\kappa_0(u,\xi)$ is strictly convex for $\xi\in(-1,1)$ for any fixed $u\in(-1,1)$. 
Hence, we may use the arguments of the proof of Theorem \ref{ConnThm} appropriately adapted for $d=2$ and $s=0$.
\end{proof}

It should be emphasized that in the logarithmic case balayage preserves mass. Thus, the logarithmic potentials of a measure and its logarithmic balayage onto a compact set $K$ differ by a constant on $K$. 

\begin{lem} \label{NuLem:s.EQ.0} Let $d=2$ and $s=0$. The measure $\overline{\nu}_{t,0}=\bal_{0}(\sigma_2,\Sigma_t)$ is given by
\begin{equation}
\dd \overline{\nu}_{t,0}(\PT{x}) = \dd \sigma_{2} \big|_{\Sigma_t}(\PT{x}) + \frac{1-t}{2} \dd \delta_t(u) \dd\sigma_{1}(\overline{\PT{x}}) \label{NuR:s.EQ.0}
\end{equation}
and $\| \overline{\nu}_{t,0}(\PT{x}) \|=1$. 
The logarithmic potential of $\overline{\nu}_{t,0}$ is given by
\begin{align*}
U_{0}^{\overline{\nu}_{t,0}}(\PT{z}) &= \frac{1+t}{4} - \frac{\log 2}{2} - \frac{1}{2} \log \left( 1 + t \right), \qquad \PT{z} \in \Sigma_t,  \\
U_{0}^{\overline{\nu}_{t,0}}(\PT{z}) &= \frac{1+t}{4} - \frac{\log 2}{2} - \frac{1}{2} \log \left( 1 + \xi \right), \qquad \PT{z} \in \mathbb{S}^d \setminus \Sigma_t.
\end{align*}
The measure $\overline{\nu}_{t,0}$ is the logarithmic extremal measure on $\Sigma_t$ and
\begin{equation} \label{W.0.Sigma.t}
W_0(\Sigma_t) = \frac{1+t}{4} - \frac{\log 2}{2} - \frac{1}{2} \log \left( 1 + t \right).
\end{equation}
\end{lem}

\begin{proof}
Using relation \eqref{log.int.aux} we show that the measure in \eqref{NuR:s.EQ.0} satisfies the balayage properties. 
Let $\PT{z} \in \Sigma_t$, that is $\xi \leq t$. Then
\begin{align*}
U_0^{\overline{\nu}_{t,0}}(\PT{z}) 
&= U_0^{\sigma_2}(\PT{z}) - U_0^{\sigma_2|_{\mathbb{S}^2\setminus\Sigma_t}}(\PT{z}) + \frac{1-t}{2} U_0^{\sigma_1|_{u=t}}(\PT{z}) \\
&= W_0(\mathbb{S}^2) - \frac{\omega_1}{\omega_2} \int_{t}^1 \left( \int_{\mathbb{S}^1} \log \frac{1}{\left| \PT{z} - \PT{x} \right|} \dd \sigma_1(\overline{\PT{x}}) \right) \dd u + \frac{1-t}{2} \int_{\mathbb{S}^1} \log\frac{1}{\left| \PT{z} - \PT{x} \right|} \Big|_{u=t} \dd \sigma_1(\overline{\PT{x}}) \\
&= \frac{1}{2} - \log 2 + \frac{1}{4} \int_{-1}^t \left[ \log \left( 1 - \xi \right) + \log \left( 1 + u \right) \right] \dd u - \frac{1-t}{4} \left[\log \left( 1 - \xi \right) + \log \left( 1 + t \right) \right] \\
&= \frac{1+t}{4} - \frac{\log 2}{2} - \frac{1}{2} \log \left( 1 + t \right) = W_0(\Sigma_t).
\end{align*}
For $\PT{z} \in \mathbb{S}^2\setminus\Sigma_t$, that is $\xi\geq t \geq u$, we have after a similar computation
\begin{align*}
U_0^{\overline{\nu}_{t,0}}(\PT{z}) 
&= \int_{\Sigma_t} \log\frac{1}{\left| \PT{z} - \PT{x} \right|} \dd \sigma_2(\PT{x}) + \frac{1-t}{2} \int_{\mathbb{S}^1} \log\frac{1}{\left| \PT{z} - \PT{x} \right|} \Big|_{u=t} \dd \sigma_1(\overline{\PT{x}}) \\
&= \frac{\omega_1}{\omega_2} \int_{-1}^t \left( \int_{\mathbb{S}^1} \log \frac{1}{\left| \PT{z} - \PT{x} \right|} \dd \sigma_1(\overline{\PT{x}}) \right) \dd u + \frac{1-t}{2} \int_{\mathbb{S}^1} \log\frac{1}{\left| \PT{z} - \PT{x} \right|} \Big|_{u=t} \dd \sigma_1(\overline{\PT{x}}) \\
&= \frac{1+t}{4} - \frac{\log 2}{2} - \frac{1}{2} \log \left( 1 + \xi \right) = W_0(\Sigma_t) + \frac{1}{2} \log \frac{1+t}{1+\xi} < W_0(\Sigma_t).
\end{align*}
Since 
\begin{equation*}
\left\| \nu_{t,0} \right\| = \int_{\Sigma_t} \dd \sigma_2 + \frac{1-t}{2} \int_{\mathbb{S}^1} \dd \sigma_1 = \frac{\omega_1}{\omega_2} \int_{-1}^t \dd u + \frac{1-t}{2} \int_{\mathbb{S}^1} \dd \sigma_1 = \frac{1+t}{2} + \frac{1-t}{2} = 1,
\end{equation*}
$\nu_{t,0}$ is a probability measure on $\Sigma_t$ which is constant there. By uniqueness of the logarithmic extremal measure $\mu_{\Sigma_t}$ on $\Sigma_t$ one has $\mu_{\Sigma_t} = \nu_{t,0}$.
\end{proof}

\begin{lem} \label{lem:Mhaskar-Saff.log}
Let $d=2$ and $s=0$. Then the Mhaskar-Saff functional $\mathcal{F}_0$ for spherical caps $\Sigma_t$ is given by
\begin{equation} \label{F.0.Sigma.t}
\begin{split}
\mathcal{F}_0(\Sigma_t) &= \left( 1 + q \right) \frac{1+t}{4} + q \frac{\left( R - 1 \right)^2 \log \left( R^2 - 2 R t + 1 \right)}{8R} - \frac{1}{2} \log\left( 1 + t \right) \\
&\phantom{=\pm}- \frac{\log 2}{2} - q \frac{\left( R + 1 \right)^2 \log \left( R + 1 \right)^2}{8R}.
\end{split}
\end{equation}
It has precisely one global minimum $t_0\in(-1,1]$. This minimum is given by
\begin{equation} \label{log.rel.aux.A}
t_0 = \min\left\{ 1, \left( R^2 - 2 R q + 1 \right) / \left[ 2 R \left( 1 + q \right) \right] \right\}.
\end{equation}
\end{lem}

\begin{proof} 
By Lemma \ref{NuLem:s.EQ.0} and $| \PT{x} - \PT{a} |^2 = R^2 - 2 R u + 1$ we obtain (with $\mu_{\Sigma_t,0} = \overline{\nu}_{t,0}$)
\begin{align*}
\int \overline{Q}_{\PT{a},q} \dd \mu_{\Sigma_t,0} 
&= q \int_{\Sigma_t} \log\frac{1}{\left| \PT{x} - \PT{a} \right|} \dd \sigma_2(\PT{x}) + q \frac{1-t}{2} \int_{\mathbb{S}^1} \log\frac{1}{\left| \PT{x} - \PT{a} \right|} \Big|_{u=t} \dd \sigma_1(\overline{\PT{x}}) \\
&= - \frac{q}{2} \frac{\omega_1}{\omega_2} \int_{-1}^t \log \left( R^2 - 2 R u + 1 \right) \dd u - \frac{q}{2} \frac{1-t}{2} \log \left( R^2 - 2 R t + 1 \right) \\
&= q \frac{1+t}{4} - q \frac{\left( R + 1 \right)^2 \log \left( R + 1 \right)^2}{8R} + q \frac{\left( R - 1 \right)^2 \log \left( R^2 - 2 R t + 1 \right)}{8R}.
\end{align*}
Substitution of the last expression and $W_0(\Sigma_t)$ from \eqref{W.0.Sigma.t} into
\begin{equation*}
\mathcal{F}_0(t) \DEF \mathcal{F}_0(\Sigma_t) = W_0(\Sigma_t) + \int \overline{Q}_{\PT{a},q} \dd \mu_{\Sigma_t,0},
\end{equation*}
yields \eqref{F.0.Sigma.t}. Observe, that $\mathcal{F}_0(t)\to\infty$ as
$t\to-1$. Furthermore,
\begin{align*}
\mathcal{F}_0^\prime(t) &= \frac{1+q}{4} - q \frac{\left( R - 1 \right)^2}{4 \left( R^2 - 2 R t + 1 \right)} - \frac{1}{2\left( 1 + t \right)} = \frac{1-t}{4} \left( \frac{2 q R}{R^2 - 2 R t + 1} - \frac{1}{1+t} \right) \\
&= \frac{R(1+q)(1-t)}{2(1+t)(R^2-2Rt+1)} \left[ 1+t-\frac{(R+1)^2}{2R(1+q)}\right].
\end{align*}
If $-1<t<1$, then the sign of $\mathcal{F}_0^\prime(t)$ is given by the sign of the
linear function in the brackets, which is negative at $t=-1$. If $(
R + 1 )^2 \geq 4R(1+q)$, then $\mathcal{F}_0^\prime(t)<0$ everywhere
on $(-1,1)$, and $\mathcal{F}_0(\Sigma_t)$ is strictly monotonically
decreasing on $(-1,1)$ and has a global minimum at $t=1$. Otherwise,
if $( R + 1 )^2 < 4R(1+q)$, then $\mathcal{F}_0^\prime(t)$ has
exactly one zero $t_0 \DEF (R^2-2Rq+1)/[2R(1+q)]$ on $(-1,1)$, and is
negative on $(-1,t_0)$ and positive on $(t_0,1)$. Clearly,
$\mathcal{F}_0 (t)$ achieves global minimum on $(-1,1]$ at $t_0$,
with value 
\begin{equation*}
\mathcal{F}_0(\Sigma_{t_0}) = \frac{\left( R + 1 \right)^2}{8 R} + q
\frac{\left( R - 1 \right)^2}{8 R} \log \frac{q}{1+q} - \frac{1}{2}
\log \frac{\left( R + 1 \right)^2}{R \left( 1 + q \right)} - q
\log\left( R + 1 \right).
\end{equation*}
This completes the proof.
%
\end{proof}

\begin{lem} \label{EpsLem:s.EQ.0} Let $d=2$ and $s=0$. The measure $\overline{\epsilon}_{t,0}=\bal_{0}(\delta_{\PT{a}},\Sigma_t)$ is given by
\begin{equation}
\dd \overline{\epsilon}_{t,0}(\PT{x}) = \frac{\left( R^2 - 1 \right)^2}{\left( R^2 - 2 R u + 1 \right)^{2}} \dd \sigma_{2} \big|_{\Sigma_t}(\PT{x}) + \frac{1-t}{2} \frac{\left( R + 1 \right)^2}{R^2 - 2 R t + 1} \dd \delta_t(u) \dd \sigma_{1}(\overline{\PT{x}}) \label{EpsR:s.EQ.0} 
\end{equation}
and $\|\overline{\epsilon}_{t,0}\|=1$. 
The logarithmic potential of $\overline{\epsilon}_{t,0}$ is given by
\begin{align*}
U_{0}^{\overline{\epsilon}_{t,0}}(\PT{z}) &= U_0^{\delta_{\PT{a}}}(\PT{z}) + \frac{1}{2} \log \frac{R^2 - 2 R t + 1}{2\left( 1 + t \right)} + \frac{\left( R + 1 \right)^2}{8R} \log \frac{\left( R + 1 \right)^2}{R^2 - 2 R t + 1}, \quad \PT{z} \in \Sigma_t, \\
U_{d-2}^{\overline{\epsilon}_{t,0}}(\PT{z}) &= U_0^{\delta_{\PT{a}}}(\PT{z}) + \frac{1}{2} \log \frac{R^2 - 2 R \xi + 1}{2\left( 1 + \xi \right)} + \frac{\left( R + 1 \right)^2}{8R} \log \frac{\left( R + 1 \right)^2}{R^2 - 2 R t + 1}, \quad \PT{z} \in \mathbb{S}^d \setminus \Sigma_t.
\end{align*}
\end{lem}

\begin{proof}
Let $\PT{z}\in\Sigma_t$. We write
\begin{align*}
\begin{split}
U_0^{\overline{\epsilon}_{t,0}}(\PT{z}) 
&= \frac{\omega_1}{\omega_2} \left( \int_{-1}^\xi + \int_\xi^t \right) \frac{\left( R^2 - 1 \right)^2}{\left( R^2 - 2 R u + 1 \right)^2} \left( \int_{\mathbb{S}^1} \log\frac{1}{\left| \PT{z} - \PT{x} \right|} \dd \sigma_1(\overline{\PT{x}}) \right) \dd u \\
&\phantom{=}+ \frac{1-t}{2} \frac{\left( R + 1 \right)^2}{R^2 - 2 R t + 1} \int_{\mathbb{S}^1} \log\frac{1}{\left| \PT{z} - \PT{x} \right|} \Big|_{u=t} \dd \sigma_1(\overline{\PT{x}}).
\end{split} 
\end{align*}
Using relation \eqref{log.int.aux} and Mathematica we arrive at
\begin{equation}
U_0^{\overline{\epsilon}_{t,0}}(\PT{z}) = - \frac{1}{2} \log\left( R^2 - 2 R \xi + 1 \right) + C(R;t), 
\end{equation}
where
\begin{align*}
C(R,t) 
&\DEF 
\frac{1}{2} \log \frac{R^2 - 2 R t + 1}{2\left( 1 + t \right)} + \frac{\left( R + 1 \right)^2}{8R} \log \frac{\left( R + 1 \right)^2}{R^2 - 2 R t + 1}.
\end{align*}

Let $\PT{z}\in\mathbb{S}^2\setminus\Sigma_t$. Then
\begin{align*}
\begin{split}
U_0^{\overline{\epsilon}_{t,0}}(\PT{z}) 
&= \frac{\omega_1}{\omega_2} \int_{-1}^t \frac{\left( R^2 - 1 \right)^2}{\left( R^2 - 2 R u + 1 \right)^2} \left( \int_{\mathbb{S}^1} \log\frac{1}{\left| \PT{z} - \PT{x} \right|} \dd \sigma_1(\overline{\PT{x}}) \right) \dd u \\
&\phantom{=}+ \frac{1-t}{2} \frac{\left( R + 1 \right)^2}{R^2 - 2 R t + 1} \int_{\mathbb{S}^1} \log\frac{1}{\left| \PT{z} - \PT{x} \right|} \Big|_{u=t} \dd \sigma_1(\overline{\PT{x}})
\end{split} 
\end{align*}
Using relation \eqref{log.int.aux} and evaluating the integral one gets after some simplifications
\begin{equation}
U_0^{\overline{\epsilon}_{t,0}}(\PT{z}) = - \frac{1}{2} \log \left[ 2 \left( 1 + \xi \right) \right] + \frac{\left( R + 1 \right)^2}{8R} \log \frac{\left( R + 1 \right)^2}{R^2 - 2 R t + 1},
\end{equation}
which yields the representation outside of $\Sigma_t$. Since 
\begin{equation*}
\frac{R^2 - 2 R \xi + 1}{R^2 - 2 R t + 1} \, \frac{1+t}{1+\xi} < 1 \quad \text{for $\xi > t$,}
\end{equation*}
it follows for $\PT{z}\in\mathbb{S}^2\setminus\Sigma_t$ that
\begin{equation*}
U_0^{\overline{\epsilon}_{t,0}}(\PT{z}) = U_0^{\delta_{\PT{a}}}(\PT{z}) + C(R,t) + \frac{1}{2} \log \left[ \frac{R^2 - 2 R \xi + 1}{R^2 - 2 R t + 1} \frac{1+t}{1+\xi} \right] < U_0^{\delta_{\PT{a}}}(\PT{z}) + C(R,t).
\end{equation*}
Hence, $\overline{\epsilon}_{t,0}$ has the properties of a logarithmic balayage measure.
Finally,
\begin{align*}
\left\| \overline{\epsilon}_{t,0} \right\| &= \int_{\Sigma_t} \frac{\left( R^2 - 1 \right)^2}{\left( R^2 - 2 R u + 1 \right)^2} \dd \sigma_2(\PT{x}) + \frac{1-t}{2} \frac{\left( R + 1 \right)^2}{R^2 - 2 R t + 1} \int_{\mathbb{S}^1} \dd \sigma_1(\overline{\PT{x}}) \\
&= \frac{\omega_1}{\omega_2} \int_{-1}^t \frac{\left( R^2 - 1 \right)^2}{\left( R^2 - 2 R u + 1 \right)^2} \dd u + \frac{1-t}{2} \frac{\left( R + 1 \right)^2}{R^2 - 2 R t + 1} \\
&= \frac{1+t}{2} \frac{\left( R - 1 \right)^2}{R^2 - 2 R t + 1} + \frac{1-t}{2} \frac{\left( R + 1 \right)^2}{R^2 - 2 R t + 1} = 1.
\end{align*}
This completes the proof.
\end{proof}

\begin{proof}[Proof of Theorem \ref{ExcepThm.log}]
Lemmas \ref{NuLem:s.EQ.0} and \ref{EpsLem:s.EQ.0} imply that $\overline{\eta}_{t,0} = ( 1 + q ) \overline{\nu}_{t,0} - q \overline{\epsilon}_{t,0}$ is, indeed, the logarithmic signed equilibrium on $\Sigma_t$ associated with $\overline{Q}_{\PT{a},q}$ as can be seen from its weighted logarithmic potential given in the Theorem. Using $r=\sqrt{R^2 - 2 R t + 1}$ and $\rho=\sqrt{R^2 - 2 R u + 1}$, we can write
\begin{align*} 
\begin{split}
\dd \overline{\eta}_{t,0}(\PT{x}) &= \left[ 1 + q - \frac{q \left( R^2 - 1 \right)^2}{\rho^{4}} \right] \dd \sigma_{2}\big|_{\Sigma_t}(\PT{x}) + \frac{1-t}{2} \left[ 1 + q - \frac{q \left( R + 1 \right)^2}{r^2} \right] \dd \beta_t(\PT{x}),
\end{split}
\end{align*}
where $\PT{x} \in \Sigma_t$.
If $\overline{\eta}_{t,0}\geq0$, then $1 + q - q ( R + 1 )^2 / (R^2 - 2 R t + 1)\geq0$, so $t\leq t_0$. On the other hand, it is easy to see that
if $t=t_0$, then $\overline{\eta}_{t_0,0}$ given in \eqref{etabarzero.log} is $\geq 0$ 
because $\rho\leq\rho$ and $( R - 1 )^2 < R^2 - 2 R u + 1$. Therefore, we have that $t_0 = \max\{ t : \overline{\eta}_{t_0,0} \geq 0 \}$, $\mu_{\overline{Q}_{\PT{a},q}}=\overline{\eta}_{t_0,0}$, and $\supp(\mu_{\overline{Q}_{\PT{a},q}})=\Sigma_{t_0}$.
\end{proof}

\section{Axis-supported Riesz external fields}
\label{sec:8}

In this section we shall prove Theorems \ref{SignEq.axis.supp}, \ref{SignEq.axis.supp.log}, \ref{SignEqThm2}, \ref{SignEqThm2.s.EQ.d-2}, and \ref{ExcepThm.log.axis}.

\begin{proof}[Proof of Theorem \ref{SignEq.axis.supp}]
Direct calculation shows that
\begin{align*}
U_s^{\tilde{\eta}_{\lambda}}(\PT{z}) 
&= \frac{\mathcal{F}_s(\mathbb{S}^d)}{W_s(\mathbb{S}^d)} U_s^{\sigma_d}(\PT{z}) - \int \left( \int_{\mathbb{S}^d} \frac{\left( R^2 - 1 \right)^{d-s} \dd \sigma_d(\PT{x})}{\left| \PT{z} - \PT{x} \right|^{s}\left| \PT{x} - \PT{a} \right|^{2d-s}} \right) \dd \lambda(R) \\
&= \frac{\mathcal{F}_s(\mathbb{S}^d)}{W_s(\mathbb{S}^d)} W_s(\mathbb{S}^d) - \int \frac{\dd \lambda(R)}{\left| \PT{z} - R \PT{p} \right|^{s}} = \mathcal{F}_s(\mathbb{S}^d) - Q(\PT{z}),
\end{align*}
where we used the Kelvin transformation for points (cf. proof of Theorem \ref{SignEq}). 

The second part follows from the uniqueness of the $s$-extremal measure on $\mathbb{S}^d$ associated with $Q$ and the fact that the density is minimal at the North Pole.
\end{proof}

\begin{proof}[Proof of Theorem \ref{SignEq.axis.supp.log}]
The logarithmic potential of $\tilde{\eta}_{\lambda,0}$ is given by
\begin{equation*}
U_0^{\tilde{\eta}_{\lambda,0}}(\PT{z}) = \left( 1 + \left\| \lambda \right\| \right) W_0(\mathbb{S}^d) - \int \left( \int_{\mathbb{S}^d} \frac{\left( R^2 - 1 \right)^d}{\left| \PT{x} - R \PT{p} \right|^{2d}} \log \frac{1}{\left| \PT{z} - \PT{x} \right|} \dd \sigma(\PT{x}) \right) \dd \lambda(R).
\end{equation*}
A Kelvin transformation with center $\PT{a}=R\PT{p}$ and radius $\sqrt{R^2-1}$ (cf. Section \ref{sec:kelvin.transf.}) yields
\begin{align*}
&\int_{\mathbb{S}^d} \frac{\left( R^2 - 1 \right)^d}{\left| \PT{x} - R \PT{p} \right|^{2d}} \log \frac{1}{\left| \PT{z} - \PT{x} \right|} \dd \sigma(\PT{x}) = \int_{\mathbb{S}^d} \frac{\left( R^2 - 1 \right)^d}{\left| \PT{x} - \PT{a} \right|^{d}} \log \frac{1}{\left| \PT{z} - \PT{x} \right|} \, \frac{\dd \sigma(\PT{x})}{\left| \PT{x} - \PT{a} \right|^{d}} \\
&\phantom{equals}= \int_{(\mathbb{S}^d)^*} \left| \PT{x}^* - \PT{a} \right|^d \log \frac{\left| \PT{z}^* - \PT{a} \right|^d \left| \PT{x}^* - \PT{a} \right|^d}{\left( R^2 - 1 \right) \left| \PT{z}^* - \PT{x}^* \right|^d} \, \frac{\dd \sigma(\PT{x}^*)}{\left| \PT{x}^* - \PT{a} \right|^{d}} \\
&\phantom{equals}= \log \frac{\left| \PT{z}^* - \PT{a} \right|}{R^2 - 1} - \int_{(\mathbb{S}^d)^*} \log \frac{1}{\left| \PT{x}^* - \PT{a} \right|} \dd \sigma(\PT{x}^*) + \int_{(\mathbb{S}^d)^*} \log \frac{1}{\left| \PT{z}^* - \PT{x} \right|} \dd \sigma(\PT{x}^*) \\
&\phantom{equals}= \log \frac{1}{\left| \PT{z} - \PT{a} \right|} - \int_{\mathbb{S}^d} \log \frac{1}{\left| \PT{y} - \PT{a} \right|} \dd \sigma(\PT{y}) + U_0^{\sigma}(\PT{z}).
\end{align*}
Hence
\begin{equation*}
\begin{split}
U_0^{\tilde{\eta}_{\lambda,0}}(\PT{z}) &= \left( 1 + \left\| \lambda \right\| \right) W_0(\mathbb{S}^d) - \int \log \frac{1}{\left| \PT{z} - \PT{a} \right|} \dd \lambda(R) \\
&\phantom{=}+ \int_{\mathbb{S}^d} \left( \int \log \frac{1}{\left| \PT{y} - \PT{a} \right|} \dd \lambda(R) \right) \dd \sigma(\PT{y}) - \left\| \lambda \right\| W_0(\mathbb{S}^d),
\end{split}
\end{equation*}
from which follows the first part of the theorem.

The second part follows from the uniqueness of the logarithmic extremal measure on $\mathbb{S}^d$ associated with $Q$ (Lemma \ref{lem:uniqueness} and in particular \cite{Riesz}) and the fact that the density is minimal at the North Pole.
\end{proof}

\begin{proof}[Proof of Theorem \ref{SignEqThm2}]
By construction $\tilde{\eta}_t$ is of total charge one. 
It is easy to verify that the signed measure \eqref{etatilde} has a constant weighted $s$-potential on $\Sigma_t$. Indeed,
\begin{equation*}
U_s^{\tilde{\epsilon}_t}(\PT{x}) = \int U_s^{\bal_s(\delta_{R\PT{p}},\Sigma_t)}(\PT{x}) \, \dd \lambda (R) = \int \frac{\dd \lambda (R)}{\left| \PT{x} - R \PT{p} \right|^s} = Q(\PT{x}), \qquad \PT{x}\in\Sigma_t.
\end{equation*}
Together with $U_s^{\nu_t}(\PT{z})=W_s(\mathbb{S}^d)$ on $\Sigma_t$ we have $U_s^{\tilde{\eta}_t}(\PT{z}) = \tilde{\Phi}_s(t)$ on $\Sigma_t$.
Moreover, by Remark \ref{FuncSignedEqRel}, we also have that $\mathcal{F}_s(\Sigma_t) = \tilde{\Phi}_s(t)$.

By definition of $\nu_t$, $\tilde{\epsilon}_t$, and $\bal_s(\delta_{R \PT{p}},\Sigma_t) = \eps_{t,R}$ (with additional indication of the dependence on the parameter $R$) we can write
\begin{equation*}
\tilde{\eta}_t = \frac{\tilde{\Phi}_s(t)}{W_s(\mathbb{S}^d)} \frac{1}{\left\| \lambda \right\|} \int \nu_t \dd \lambda(R) - \int \epsilon_{t,R} \dd \lambda(R) = \int \left[ \frac{\tilde{\Phi}_s(t)}{W_s(\mathbb{S}^d)} \frac{1}{\left\| \lambda \right\|} \nu_t - \epsilon_{t,R} \right] \dd \lambda(R).
\end{equation*}
Thus, the signed equilibrium is 
\begin{equation*}
\dd \tilde{\eta}_t(\PT{x}) = \left[ \int \tilde{\eta}_{t}^{\prime\prime}(u,R) \, \dd \lambda(R) \right] \frac{\omega_{d-1}}{\omega_d} (1-u^2)^{d/2-1}\, \dd u \dd \sigma_{d-1} (\overline{\PT{x}}), \qquad \PT{x} \in \Sigma_t,
\end{equation*}
where, when using Lemmas \ref{NuLem} and \ref{EpsilonLem}, we have
\begin{equation}
\begin{split} \label{AxisProof1}
\tilde{\eta}_{t}^{\prime\prime}(u,R) &= \frac{1}{W_s(\mathbb{S}^d)} \frac{1}{\left\| \lambda \right\|} \frac{\gammafcn(d/2)}{\gammafcn(d-s/2)}
\left( \frac{1-t}{1-u} \right)^{d/2} \left( \frac{t-u}{1-t} \right)^{(s-d)/2} \\
&\phantom{=\times}\times \Bigg\{ \tilde{\Phi}_s(t)
\HypergeomReg{2}{1}{1,d/2}{1-(d-s)/2}{\frac{t-u}{1-u}}  \\
&\phantom{=\times\pm}- \frac{\left\| \lambda \right\| \left( R + 1 \right)^{d-s}}{r^{d}}
\HypergeomReg{2}{1}{1,d/2}{1-(d-s)/2}{\frac{\left(R-1\right)^{2}}{r^{2}}
\, \frac{t-u}{1-u}}  \Bigg\}. 
\end{split}
\end{equation}

We claim that the density (the integral in square brackets) is either
positive for all $u\in[-1,t]$, or is positive on some interval
$[-1,t_c)$ and negative on $(t_c,t]$. It suffices to consider the
function $h(u)$ obtained by integrating the expression in braces
in \eqref{AxisProof1} against $\dd\lambda(R)$. Using the series
expansion of the hypergeometric functions we get
\begin{equation}
\begin{split}
h(u) &= \sum_{k=0}^\infty \frac{\Pochhsymb{d/2}{k}}{\gammafcn(k+1-(d-s)/2)} \left( \frac{t-u}{1-u} \right)^k \\
&\phantom{=\pm}\times \left\{ \int \left[ \tilde{\Phi}_s(t) - \frac{\left\| \lambda \right\| \left( R + 1 \right)^{d-s}}{r^{d}} \left( \frac{R-1}{r} \right)^{2k} \right] \dd \lambda(R) \right\}.
\end{split}
\end{equation}
The coefficients in braces form an increasing sequence with
positive limit as $k\to \infty$. Hence, either all coefficients are
positive, or the first $n$ are negative and then all others are
positive. 
So, with a substitution $x=(t-u)/(1-u)$ ($x\in A_t\DEF[0,(1+t)/2]$) we obtain
\begin{equation*}
g(x) = \sum_{k=0}^\infty \frac{a_k}{k!} x^k, \qquad \text{$a_k<0$ for $k < n$ and $a_k\geq 0$ for $k\geq n$.}
\end{equation*}
We have that $g^{(n)}(x) > 0$ on $A_t$, so $g^{(n-1)}(x)$ is strictly
increasing on $A_t$. Since $g^{(n-1)}(0) = a_{n-1} < 0$, there is a $\gamma_{n-1}$ in $A_t$
such that $g^{(n-1)}(x)$ is negative on $[0,\gamma_{n-1})$ and positive
on $(\gamma_{n-1},(1+t)/2]$. Indeed, if such a $\gamma_{n-1}$
does not exist, we get a contradiction, because $g^{(n-1)}(x)$ will be negative on $A_t$,
which would imply that $g^{(n-2)}(x)$ is decreasing and negative on $A_t$, and so on. 
This argument yields $g(x)<0$ on $A_t$, which is impossible because the total charge of
$\tilde{\eta}_t$ is one.

By iteration one can show a sequence $\gamma_0>\gamma_1>\cdots>\gamma_{n-1}$ such that $g^{(m)}(x)$ is negative on $[0,\gamma_m)$ and positive on $(\gamma_m,(1+t)/2]$ for every $m=0,1,\dots,n-1$. This establishes our claim ($t_c=\gamma_0$).

We now can complete the proof of the theorem as follows. If $\tilde{\eta}_1$ is
not a positive measure, then there is a $t_1$ such that the density
of $\tilde{\eta}_1$ is positive on $[-1,t_1)$ and negative on
$(t_1,1]$. Then the signed equilibrium for $\Sigma_{t_1}$ is given by
\begin{equation*}
\tilde{\eta}_{t_1} = \tilde{\eta}_1^+ - \bal_s(\tilde{\eta}_1^-,\Sigma_{t_1}) - \left( \left\| \tilde{\eta}_1^- \right\| - \left\| \bal_s(\tilde{\eta}_1^-,\Sigma_{t_1}) \right\| \right) \nu_{t_1} \big/ \left\| \nu_{t_1} \right\|.
\end{equation*}
If it is still not a positive measure, then there exists a $t_2$ such
that $\tilde{\eta}_{t_1}$ has positive density on $[-1,t_2)$ and
negative one on $(t_2,t_1]$. Continuing the argument we derive a
decreasing sequence $\{ t_k \}$ with the property that
$\tilde{\eta}_{t_k}$ is positive on $[-1,t_{k+1})$ and negative on
$(t_{k+1},t_k]$. The limit of this sequence is the number $t_\lambda$
defined in Theorem \ref{SignEqThm2}. 
Thus, $t_\lambda = \max\{ t : \tilde{\eta}_t \geq 0 \}$, $\mu_Q=\tilde{\eta}_{t_\lambda}$, and $\supp(\mu_Q)=\Sigma_{t_\lambda}$.

The Mhaskar-Saff functional $\mathcal{F}_s$ is minimized for $\Sigma_{t_\lambda}$. 
Since $\mathcal{F}_s(\Sigma_t) = \tilde{\Phi}_s(t)$ (cf. Remark \ref{FuncSignedEqRel} and beginning of this proof), we will show similar as in the proof of Theorem \ref{MainThm} above that $t_\lambda$ is, in fact, the unique solution in $(-1,1]$ of the relation
\begin{equation} \label{tilde.Phi.rel.aux}
\Delta(t) \DEF \tilde{\Phi}_s(t) - \int \frac{\left( R + 1 \right)^{d-s}}{\left( R^2 - 2 R t + 1 \right)^{d/2}} \dd \lambda (R) = 0,
\end{equation}
or $t_\lambda = 1$ when such a solution does not exist. 

Using Quotient Rule and $\| \tilde{\eps}_t \|^\prime = \dd \| \tilde{\eps}_t \| / \dd t = \int \| \eps_{t,R} \|^\prime \dd \lambda(R)$, we obtain
\begin{align*}
\tilde{\Phi}_s^\prime(t) = - \left\| \nu_t \right\|^\prime \big/ \left\| \nu_t \right\| \, \Delta(t).
\end{align*}
Observe that $\Delta(t)\to\infty$ as $t\to-1^+$. Hence, by the above relation, $\tilde{\Phi}_s(t)$ is strictly monotonically decreasing on $(-1,t^\prime)$ for some maximal $t^\prime\in(-1,1]$ (cf. \eqref{NuNormB}). If $t^\prime=1$, then $t_\lambda=1$. Otherwise, $t^\prime<1$ and $\tilde{\Phi}_s(t^\prime)=0$ meaning that $t^\prime$ is a solution of \eqref{tilde.Phi.rel.aux}.
Arguing as in the proof of Theorem \ref{MainThm} we have that every solution $t_0\in(-1,1)$ of \eqref{tilde.Phi.rel.aux} is actually a local minimum of $\tilde{\Phi}_s(t)$ because of $\tilde{\Phi}_s^{\prime\prime}(t_0)>0$. We conclude that $\tilde{\Phi}_s(t)$ can have at most one minimum in $(-1,1)$. Consequently $t_\lambda=t^\prime$. We also infer that $\Delta>0$ on $(-1,t_\lambda)$ and $\Delta<0$ on $(t_\lambda,1]$.
This completes the proof.
\end{proof}

\begin{proof}[Proof of Theorem \ref{SignEqThm2.s.EQ.d-2}]
By definition \eqref{barBal} and Lemmas \ref{NuLem:s.EQ.d-2} and \ref{EpsLem:s.EQ.d-2} one can easily see that
\begin{equation*}
U_{d-2}^{\tilde{\overline{\epsilon}}_t}(\PT{z}) = \int U_{d-2}^{\bal_{d-2}(\delta_{R\PT{p}},\Sigma_t)}(\PT{z}) \, \dd \lambda (R) = \int \frac{\dd \lambda (R)}{\left| \PT{z} - R \PT{p} \right|^{d-2}} = Q(\PT{z}), \qquad \PT{z}\in\Sigma_t,
\end{equation*}
and $U_{d-2}^{\tilde{\overline{\nu}}_t}(\PT{z}) = W_{d-2}(\mathbb{S}^d)$ on $\Sigma_t$; hence, the weighted $(d-2)$-potential of the signed measure $\tilde{\overline{\eta}}_t$ is constant on $\Sigma_t$, that is
\begin{equation*}
U_{d-2}^{\tilde{\overline{\eta}}_t}(\PT{z}) + Q(\PT{z}) = \tilde{\overline{\Phi}}_{d-2}(t) = \mathcal{F}_{d-2}(\Sigma_t), \qquad \text{on $\Sigma_t$.}
\end{equation*}
The last relation follows from Remark \ref{FuncSignedEqRel}.
Moreover (with additional indication of the dependence on the parameter $R$), 
\begin{equation*}
\tilde{\overline{\eta}}_t = \int \left[ \frac{\tilde{\overline{\Phi}}_{d-2}(t)}{W_{d-2}(\mathbb{S}^d)} \frac{1}{\left\| \lambda \right\|} \overline{\nu}_t - \overline{\epsilon}_{t,R} \right] \dd \lambda(R).
\end{equation*}
Thus, the signed equilibrium is 
\begin{equation} \label{tilde.bar.eta}
\dd \tilde{\overline{\eta}}_t(\PT{x}) = \left[ \int \tilde{\overline{\eta}}_{t}^{\prime\prime}(u,R) \, \dd \lambda(R) \right] \dd \sigma_d \big|_{\Sigma_t}(\PT{x}) + \left[ \int \tilde{\overline{\eta}}_{t}^{\prime\prime\prime}(u,R) \, \dd \lambda(R) \right] \dd \beta_t(\PT{x}),
\end{equation}
where, when using Lemmas \ref{NuLem:s.EQ.d-2} and \ref{EpsLem:s.EQ.d-2}, we have for $-1 \leq u \leq t$ 
\begin{align}
\tilde{\overline{\eta}}_{t}^{\prime\prime}(u,R) &= \frac{1}{W_{d-2}(\mathbb{S}^d)} \frac{1}{\left\| \lambda \right\|} \left[ \tilde{\overline{\Phi}}_{d-2}(t) - \frac{\left\| \lambda \right\| \left( R^2 - 1 \right)^2}{\left( R^2 - 2 R u + 1 \right)^{d/2+1}} \right], \label{eta.prime.prime} \\
\tilde{\overline{\eta}}_{t}^{\prime\prime\prime}(u,R) &= \frac{1}{\left\| \lambda \right\|} \frac{1-t}{2} \left[ \tilde{\overline{\Phi}}_{d-2}(t) - \frac{\left\| \lambda \right\| \left( R + 1 \right)^2}{\left( R^2 - 2 R t + 1 \right)^{d/2}} \right] \left( 1 - t^2 \right)^{d/2-1}. \label{eta.prime.prime.prime}
\end{align}
It can be shown that the density with respect to $\sigma_d|_{\Sigma_t}$,
\begin{equation*}
g(u) \DEF \int \tilde{\overline{\eta}}_{t}^{\prime\prime}(u,R) \, \dd \lambda(R) = \frac{1}{W_{d-2}(\mathbb{S}^d)} \left[ \tilde{\overline{\Phi}}_{d-2}(t) - \int \frac{\left( R^2 - 1 \right)^2 \dd \lambda(R)}{\left( R^2 - 2 R u + 1 \right)^{d/2+1}} \right],
\end{equation*}
is either positive for all $u\in[-1,t]$, or is positive on some interval $[-1,t_c)$ and negative on $(t_c,t]$. This follows easily from the fact that $g(u)$ is a strictly monotonically decreasing continuous function on $[-1,t]$. 
Now, we turn to the density with respect to $\beta_t$, that is
\begin{equation*}
h(u) \DEF \int \tilde{\overline{\eta}}_{t}^{\prime\prime\prime}(u,R) \, \dd \lambda(R) = \frac{1-t}{2} \left[ \tilde{\overline{\Phi}}_{d-2}(t) - \int \frac{\left( R + 1 \right)^2 \dd \lambda(R)}{\left( R^2 - 2 R t + 1 \right)^{d/2}} \right] \left( 1 - t^2 \right)^{d/2-1}.
\end{equation*}
Observe that non-negativity of the above square bracketed expression implies $g(u)\geq0$ and therefore $\tilde{\overline{\eta}}_t \geq 0$. On the other hand, if $\tilde{\overline{\eta}}_t \geq 0$, then 
\begin{equation} \label{non-neg.rel.s.EQ.d-2}
\tilde{\overline{\Phi}}_{d-2}(t) \geq \int \frac{\left( R + 1 \right)^2 \dd \lambda(R)}{\left( R^2 - 2 R t + 1 \right)^{d/2}}.
\end{equation}
Hence, the last relation holds if and only if $\tilde{\overline{\eta}}_t \geq 0$. 
Note that $\tilde{\overline{\Phi}}_{d-2}(t)\to\infty$ as $t\to-1^+$ and $\tilde{\overline{\Phi}}_{d-2}(1)=\mathcal{F}_{d-2}(\mathbb{S}^d)$. Set $t_\lambda \DEF \{ t : \tilde{\overline{\eta}}_t \geq 0 \}$. Arguing as in the proof of Theorem \ref{SignEqThm2} it can be shown that for $t=t_\lambda$ equality holds in \eqref{non-neg.rel.s.EQ.d-2} and $t_\lambda$ is the unique minimum of $\tilde{\overline{\Phi}}_{d-2}(t)$ on $(-1,1)$ if it exists, or $t_\lambda=1$. In particular, 
\begin{align*}
\frac{\dd \tilde{\overline{\Phi}}_{d-2}(t)}{\dd t} = - \frac{\left\{ \left\| \nu_t \right\| \right\}^\prime}{\left\| \nu_t \right\| } \left[ \tilde{\overline{\Phi}}_{d-2}(t) - \int \frac{\left( R + 1 \right)^{2}}{\left( R^2 - 2 R t + 1 \right)^{d/2}} \dd \lambda (R) \right].
\end{align*}

It remains to show the weak$^*$ convergence in \eqref{weak.star.axis-supp}. This follows from \eqref{weak.star.conv}, since for any function $f$ continuous on $\mathbb{S}^d$ we have
\begin{align*}
\int_{\Sigma_t} f \dd \tilde{\epsilon}_{t,s} 
&= \int_{\Sigma_t} f \dd \left( \int \bal_s(\delta_{R \PT{p}},\Sigma_t) \dd \lambda(R) \right) = \int_{\Sigma_t} f \dd \left( \int \epsilon_{t,s} \dd \lambda(R) \right) \\
&= \int_{\Sigma_t} f \left( \int \dd \epsilon_{t,s} \dd \lambda(R) \right) = \int \left( \int_{\Sigma_t} f \dd \epsilon_{t,s} \right) \dd \lambda(R) \\
&\to \int \left( \int_{\Sigma_t} f \dd \overline{\epsilon}_{t} \right) \dd \lambda(R) = \int_{\Sigma_t} f \dd \tilde{\overline{\epsilon}}_t \qquad \text{as $s\to(d-2)^+$.}
\end{align*}
This completes the proof.
\end{proof}

\begin{proof}[Proof of Theorem \ref{ExcepThm.log.axis}]
First observe that
\begin{equation*}
\left\| \tilde{\overline{\epsilon}}_{t,0} \right\| = \int \left\| \overline{\epsilon}_{t,0} \right\| \dd \lambda(R) = \int \dd \lambda(R) = \left\| \lambda \right\|,
\end{equation*}
which follows from principle of superposition and preservation of mass when using logarithmic balayage. Hence, $\| \tilde{\overline{\eta}}_{t,0} \| = 1$ by construction. The representation \eqref{etabar.log.axis} can be easily obtained using Lemmas \ref{NuLem:s.EQ.0} and \ref{EpsLem:s.EQ.0}. In particular, it follows from Lemma \ref{EpsLem:s.EQ.0} that for $\PT{z}\in\Sigma_t$ there holds
\begin{align*}
U_0^{\tilde{\overline{\epsilon}}_{t,0}}(\PT{z}) 
&= \int U_0^{\overline{\epsilon}_{t,0}}(\PT{z}) \dd \lambda(R) \\
&= \tilde{\overline{Q}}(\PT{z}) + \int \left[ \frac{1}{2} \log \frac{R^2 - 2 R t + 1}{2\left( 1 + t \right)} + \frac{\left( R + 1 \right)^2}{8R} \log \frac{\left( R + 1 \right)^2}{R^2 - 2 R t + 1} \right] \dd \lambda(R).
\end{align*}
By Lemma \ref{NuLem:s.EQ.0}, $U_0^{\overline{\nu}_{t,0}}(\PT{z}) = W_0(\Sigma_t)$ on $\Sigma_t$. 
Hence, for $\PT{z}\in\Sigma_t$ 
\begin{align*}
&U_0^{\tilde{\overline{\eta}}_{t,0}}(\PT{z}) + \tilde{\overline{Q}}(\PT{z}) = \left( 1 + \left\| \lambda \right\| \right) W_0(\Sigma_t) - U_0^{\tilde{\overline{\epsilon}}_{t,0}}(\PT{z}) + \tilde{\overline{Q}}(\PT{z}) = W_0(\Sigma_t) \\
&\phantom{=}+ \int \left[ W_0(\Sigma_t) - \frac{1}{2} \log \frac{R^2 - 2 R t + 1}{2\left( 1 + t \right)} - \frac{\left( R + 1 \right)^2}{8R} \log \frac{\left( R + 1 \right)^2}{R^2 - 2 R t + 1} \right] \dd \lambda(R).
\end{align*}
After substituting $W_0(\Sigma_t)$ with \eqref{W.0.Sigma.t}, the above integral becomes the right-hand side of (cf. proof of Lemma \ref{lem:Mhaskar-Saff.log})
\begin{align*}
&\int \tilde{\overline{Q}} \, \dd \mu_{\Sigma_t} 
= \int \left( \int_{\Sigma_t} \log\frac{1}{\left| \PT{x} - \PT{a} \right|} \dd \mu_{\Sigma_t}(\PT{x}) \right) \dd \lambda(R) \\
&\phantom{=}= \left\| \lambda \right\| \frac{1+t}{4} + \int \left[ \frac{\left( R - 1 \right)^2 \log \left( R^2 - 2 R t + 1 \right)-\left( R + 1 \right)^2 \log \left( R + 1 \right)^2}{8R}\right] \dd \lambda(R).
\end{align*}
Thus,
\begin{equation*}
U_0^{\tilde{\overline{\eta}}_{t,0}}(\PT{z}) + \tilde{\overline{Q}}(\PT{z}) = W_0(\Sigma_t) + \int \tilde{\overline{Q}} \dd \mu_{\Sigma_t} \FED \tilde{\overline{\mathcal{F}}}_0(\Sigma_t) \DEF \tilde{\overline{\mathcal{F}}}_0(t), \qquad \PT{z} \in \Sigma_t.
\end{equation*}
A similar computation shows that for $\PT{z} \in \mathbb{S}^2 \setminus \Sigma_t$
\begin{equation*}
U_0^{\tilde{\overline{\eta}}_{t,0}}(\PT{z}) + \tilde{\overline{Q}}(\PT{z}) = \tilde{\overline{\mathcal{F}}}_0(\Sigma_t) + \frac{1}{2} \log \frac{1+t}{1+\xi} + \int \frac{1}{2} \log \frac{R^2 - 2 R t + 1}{R^2 - 2 R \xi + 1} \dd \lambda(R).
\end{equation*}

The Mhaskar-Saff functional $\tilde{\overline{\mathcal{F}}}_0$ for spherical caps $\Sigma_t$ can be represented as
\begin{equation}
\begin{split} \label{Mashkar-Saff.functional.log.case}
&\tilde{\overline{\mathcal{F}}}_0(t) = \left( 1 + \left\| \lambda \right\| \right) \frac{1 + t}{4} - \frac{\log 2}{2} - \frac{1}{2} \log\left( 1 + t \right) \\
&\phantom{=\pm}+ \int \left[ \frac{\left( R - 1 \right)^2 \log \left( R^2 - 2 R t + 1 \right)-\left( R + 1 \right)^2 \log \left( R + 1 \right)^2}{8R}\right] \dd \lambda(R),
\end{split}
\end{equation}
which yields
\begin{equation}
\frac{\dd}{\dd t} \tilde{\overline{\mathcal{F}}}_0(t) = \frac{1 + \left\| \lambda \right\|}{4} - \frac{1}{2\left( 1 + t \right)} - \frac{1}{4} \int \frac{\left( R - 1 \right)^2}{R^2 - 2 R t + 1} \dd \lambda(R).
\end{equation}
It follows that $\tilde{\overline{\mathcal{F}}}_0^\prime(t) \to - \infty$ as $t \to -1^+$ and $\tilde{\overline{\mathcal{F}}}_0^\prime(1) = 0$. If $t<1$, the equation $\tilde{\overline{\mathcal{F}}}_0^\prime(t) = 0$ is equivalent with each of the following two relations
\begin{align*}
\frac{1-t}{1+t} = \int \frac{2 R \left( 1 - t \right)}{R^2 - 2 R t + 1} \dd \lambda(R), \qquad \frac{2}{1+t} = 1 + \left\| \lambda \right\| - \int \frac{\left( R - 1 \right)^2}{R^2 - 2 R t + 1} \dd \lambda(R),
\end{align*}
which combined give
\begin{equation} \label{prime.rel.}
1 + \left\| \lambda \right\| = \int \frac{\left( R + 1 \right)^2}{R^2 - 2 R t + 1} \dd \lambda(R).
\end{equation}
The above right-hand side is a strictly monotonically increasing function in $t$. Hence, above relation has a unique solution $t_\lambda$ in $(-1,1)$ if such a solution exists (which is a minimum of $\tilde{\overline{\mathcal{F}}}_0(t)$ by the properties of $\tilde{\overline{\mathcal{F}}}_0^\prime(t)$). If there is no such solution, then $\tilde{\overline{\mathcal{F}}}_0^\prime(t)<0$ on $(-1,1)$ and $\tilde{\overline{\mathcal{F}}}_0(t)$ is strictly monotonically decreasing on $(-1,1)$. We conclude, that the Mhaskar-Saff functional $\tilde{\overline{\mathcal{F}}}_0$ is minimized for $\Sigma_{t_\lambda}$, where either $t_\lambda\in(-1,1]$ is the unique solution of equation \eqref{prime.rel.}, or $t_\lambda=1$ if such a solution does not exists.
It follows from the representation \eqref{etabar.log.axis} that $\tilde{\overline{\eta}}_t \geq 0$ if and only if 
\begin{equation} \label{prime.rel.1}
1 + \left\| \lambda \right\| \geq \int \frac{\left( R + 1 \right)^2}{R^2 - 2 R t + 1} \dd \lambda(R).
\end{equation}
Hence, $\tilde{\overline{\eta}}_t \geq 0$ if and only if $t \leq t_\lambda$, that is $t_\lambda = \max\{ t : \tilde{\overline{\eta}}_t \geq 0 \}$,  $\supp(\mu_{\tilde{\overline{Q}}})=\Sigma_{t_\lambda}$, and $\mu_{\tilde{\overline{Q}}}=\tilde{\overline{\eta}}_{t_\lambda,0}$.
\end{proof}

\appendix

\section{}

\begin{lem} \label{IntLemB}
Let $-1\leq a < b < c \leq 1$ and $|y|<1$. Then
\begin{equation}
\begin{split} \label{IntA}
&\int_{a}^{b} \left( u - a \right)^{\beta-1} \left( b - u \right)^{\gamma-1}
\left( c - u \right)^{-\alpha} \HypergeomReg{2}{1}{\alpha,\beta}{\gamma}{y \frac{b-u}{c-u}} \dd u \\
&\phantom{=}= \frac{\gammafcn(\beta)}{\gammafcn(\beta+\gamma-\alpha)\gammafcn(\alpha)}
\left( b - a \right)^{\beta+\gamma-1} \left( c - a \right)^{-\gamma} \left( c - b \right)^{\gamma-\alpha} \left( 1 - x y \right)^{-\beta} \\
&\phantom{=\pm}\times \int_{0}^{1} v^{\beta+\gamma-\alpha-1} \left(
1 - v \right)^{\alpha-1} \left( 1 - x v \right)^{\beta-\gamma}
\left( 1 - \frac{x \left( 1 - y \right)}{1 - x y} v \right)^{-\beta}
\dd v
\end{split}
\end{equation}
for all $\alpha,\beta,\gamma>0$ with $\beta+\gamma>\alpha$. Here
$x\DEF (b-a)/(c-a)$.
\end{lem}

The last integral is the Euler type integral representation of an
Appell $\HyperF_{1}$ function (\cite[5.8(5)]{ERD}, see also
\cite[7.2.4(42)]{PBMIII}).

\begin{proof}
A change of variable $(b-u)/(c-u)=x v$ yields
\begin{equation}
\begin{split} \label{A0}
&\left( b - a \right)^{\beta+\gamma-1} \left( c - a \right)^{-\gamma} \left( c - b \right)^{\gamma-\alpha} \\
&\phantom{times}\times \int_{0}^{1} v^{\gamma-1} \left( 1 - v
\right)^{\beta-1} \left( 1 - x v\right)^{\alpha-\beta-\gamma}
\HypergeomReg{2}{1}{\alpha,\beta}{\gamma}{x y v} \dd v,
\end{split}
\end{equation}
where $0<x<1$. Let $I$ denote the integral above. We substitute the
series expansion of the hypergeometric function and integrate
termwise.
\begin{equation} \label{A1}
I = \sum_{n=0}^{\infty}
\frac{\Pochhsymb{\alpha}{n}\Pochhsymb{\beta}{n}}{\gammafcn(n+\gamma)
n!} x^{n} y^{n} \int_{0}^{1} \frac{v^{n+\gamma-1} \left( 1 - v
\right)^{\beta-1}}{ \left( 1 - x v\right)^{\gamma+\beta-\alpha}} \dd
v,
\end{equation}
where the integral $K_n$ represents the hypergeometric function
(\cite[Eq.~15.3.1]{ABR})
\begin{equation*}
K_n =
\frac{\gammafcn(\beta)\gammafcn(n+\gamma)}{\gammafcn(n+\beta+\gamma)}
\Hypergeom{2}{1}{\gamma+\beta-\alpha,n+\gamma}{n+\beta+\gamma}{x}.
\end{equation*}
Since $\re[n+\beta+\gamma]>\re[\gamma+\beta-\alpha]>0$ by
assumption, another application of \cite[Eq.~15.3.1]{ABR} gives
\begin{equation*}
K_{n} =
\frac{\gammafcn(\beta)\gammafcn(n+\gamma)}{\gammafcn(\beta+\gamma-\alpha)\gammafcn(n+\alpha)}
\int_{0}^{1} v^{\beta+\gamma-\alpha-1} \left( 1 - v
\right)^{n+\alpha-1} \left( 1 - x v \right)^{-n-\gamma} \dd v.
\end{equation*}
Substituting the last formula into \eqref{A1} and reversing the
order of integration and summation, which is justified by uniform
convergences of the series for $0\leq v\leq 1$,  yields
\begin{equation*}
I =
\frac{\gammafcn(\beta)}{\gammafcn(\beta+\gamma-\alpha)\gammafcn(\alpha)}
\int_{0}^{1} \frac{v^{\beta+\gamma-\alpha-1} \left( 1 - v
\right)^{\alpha-1} }{ \left( 1 - x v \right)^{\gamma} }
\sum_{n=0}^{\infty} \frac{\Pochhsymb{\beta}{n}}{n!} \left( x y
\frac{1 - v}{1 - x v} \right)^{n} \dd v.
\end{equation*}
The infinite series equals $( 1 - x v )^\beta ( 1 - x v - x y + x y
v)^{-\beta}$. Thus, we get
\begin{equation*}
\begin{split}
I &= \frac{\gammafcn(\beta)}{\gammafcn(\beta+\gamma-\alpha)\gammafcn(\alpha)} \left( 1 - x y \right)^{-\beta} \\
&\phantom{=\pm}\times \int_{0}^{1} v^{\beta+\gamma-\alpha-1} \left(
1 - v \right)^{\alpha-1} \left( 1 - x v \right)^{\beta-\gamma}
\left( 1 - \frac{x \left( 1 - y \right)}{1 - x y} v \right)^{-\beta}
\dd v.
\end{split}
\end{equation*}

We needed that $\re[\gamma+\beta-\alpha]>0$ and $\re[\alpha]>0$ as
well as $x<1$ and $x(1-y)/(1-x y)<1$. The last formula for $I$ and
\eqref{A0} give \eqref{IntA}.
\end{proof}

{\it Acknowledgement.} The authors are grateful to Robert Scherrer
for pointing out connections with elementary electromagnetic theory.

An important part of this research was done at the Mathematisches Forschungsinstitut Oberwolfach during the first author's stay 
within the Oberwolfach Leibniz Fellows Programme (OWLF) from September 28 -- December 18, 2008. The first author is in particular grateful that the OWLF and the Leibniz-Gemeinschaft supporting this program made it possible to invite his collaborators to visit MFO. We would like
to thank the MFO for the excellent working conditions.

\end{document}